\documentclass[a4paper,12pt]{article}
\usepackage[a4paper,total={6.25in,10in}]{geometry}

\usepackage{newunicodechar}
\newunicodechar{ʹ}{'}

\usepackage[backend=biber,sorting=none,style=numeric-comp,giveninits=true]{biblatex}
\usepackage{amsmath,amsfonts,amssymb,amsthm}
\usepackage{mathtools}
\usepackage{mathabx}
\usepackage{calc}
\usepackage{shuffle}
\usepackage{etoolbox}
\usepackage{cancel}
\usepackage{blkarray}

\usepackage{enumitem}
\usepackage[dvipsnames]{xcolor}

\usepackage{graphicx}
\usepackage{subfigure}
\usepackage[export]{adjustbox}
\usepackage{float}

\usepackage{tikz}
\usepackage{tikz-3dplot}
\usetikzlibrary{backgrounds,shapes,positioning,calc,intersections,arrows.meta}

\usepackage{comment}
\usepackage{imakeidx}

\usepackage{savesym}
\savesymbol{div}
\usepackage{physics}
\restoresymbol{TxF}{div}

\usepackage{hep-title}

\usepackage{hyperref}

\theoremstyle{remark}

\theoremstyle{plain}
\newtheorem{theorem}{Theorem}[section]
\newtheorem{lemma}[theorem]{Lemma} 
\newtheorem{proposition}[theorem]{Proposition}

\newtheorem{cor}[theorem]{Corollary}

\makeatletter
\newtheorem*{rep@theorem}{\rep@title}
\newcommand{\newreptheorem}[2]{%
\newenvironment{rep#1}[1]{%
 \def\rep@title{#2 \ref{##1}}%
 \begin{rep@theorem}}%
 {\end{rep@theorem}}}
\makeatother
\newreptheorem{theorem}{Theorem}
\newreptheorem{cor}{Corollary}
\newreptheorem{proposition}{Proposition}

\theoremstyle{definition}
\newtheorem{definition}[theorem]{Definition}
\newtheorem{eg}[theorem]{Example}

\newtheorem{conj}[theorem]{Conjecture}
\newtheorem{rem}[theorem]{Remark}
\newtheorem{question}[theorem]{Question}

\makeatletter
\newcommand*\bigcdot{\mathpalette\bigcdot@{1}}
\newcommand*\bigcdot@[2]{\mathbin{\vcenter{\hbox{\scalebox{#2}{$\m@th#1\bullet$}}}}}
\makeatother
\makeindex
\begin{document}
\preprint{
\vspace*{50 pt}\hspace*{-40 pt}
MPP-2025-157}
\title{Canonical Forms as Dual Volumes} 
\author[*]{Elia Mazzucchelli \email{eliam@mpp.mpg.de}}
\affiliation[*]{Max-Planck-Institut f\"{u}r Physik, 
Boltzmannstr. 8,
85748 Garching, Germany.}
\author[*]{Prashanth Raman\email{praman@mpp.mpg.de}}
\date{\today}

\maketitle
\thispagestyle{empty}

\begin{abstract}
    We study {\it dual volume} representations of canonical forms for positive geometries in projective spaces, expressing their rational canonical functions as Laplace transforms of measures supported on the convex dual of the semialgebraic set. More generally, for hyperbolic positive geometries the inverse Laplace transform is naturally understood as a distribution supported on the closed dual cone. When this distribution is a non-negative measure, we term the geometry {\it completely monotone}, reflecting the property of its canonical function. We identify a class of positive geometries with hyperbolic algebraic boundary for which this cone-supported distribution exists. In particular, simplex-like minimal spectrahedra are completely monotone, with representing measures governed by the Riesz--Wishart theory of the cone of positive semidefinite matrices and admitting integral descriptions over spectrahedral slices and lower-rank strata. We explicitly compute these measures for positive geometries in the projective plane bounded by lines and conics or by a nodal cubic, revealing periods evaluating to transcendental functions. This dual volume perspective replaces logarithmic differential forms by dual measures in the completely monotone case, forging new connections to hyperbolic differential operators, convex optimization and algebraic statistics.
\end{abstract}

\clearpage

\thispagestyle{empty}

\tableofcontents

\clearpage
\setcounter{page}{1}

\section{Introduction}

Positive geometries arose from geometric formulations of scattering amplitudes in planar $\mathcal{N}=4$ super Yang--Mills theory. At tree level and fixed helicity, amplitudes can be described in terms of canonical functions associated with positive Grassmannians~\cite{ampl_and_pos_gr}; extending this picture to all helicity sectors and loop level integrands led to the amplituhedron~\cite{the_amplituhedron}, a semialgebraic set in the Grassmannian which is believed to be a positive geometry~\cite{Positive_geometries,Ranestad:adjoint}. Since then, positive geometries have appeared in a variety of settings from scattering amplitudes and cosmological correlators~\cite{ABHY_original,ABJM_amplituhedron,Correlahedron,Cosmological_polytopes,Banerjee:2018tun,Raman:2019utu,Aneesh:2019cvt,Jagadale:2020qfa,Jagadale:2022rbl,Jagadale:2023hjr,Cosmoehdra}. In parallel, positive geometries have developed into a mathematical subject at the interface of algebraic geometry, combinatorics and analysis~\cite{Positive_geometries,Lam:_PG_notes}; see~\cite{Ranestad:what_is_PG,Fevola:Pos_Geom} for recent surveys.

A central motivation for positive geometries is the interpretation of scattering amplitudes as geometric volumes. This viewpoint goes back to Hodges~\cite{Hodges}, who represented tree-level gluon amplitudes by volumes of polytopes in projective space and thereby made cancellations of spurious poles geometrically manifest. Here ``volume'' refers to a rational function of the kinematic data, rather than to an ordinary Euclidean volume. We begin by recalling this picture for a general convex polytope.


Let $P\subset\mathbb{P}^m_\mathbb{R}$ be a convex polytope. Polytopes are positive geometries~\cite{Positive_geometries}, with canonical form denoted by $\mathbf{\Omega}_P$. Choose an affine chart containing $P$, and for $x\in\operatorname{int}(P)$ define the shifted polar dual polytope as 
\begin{equation}\label{eq:polar_dual}
    (P-x)^{\circ} := \{ y \in \mathbb{R}^{m} : \langle x^\prime - x, y \rangle \geq -1,\ \forall\, x^\prime \in P \} \, .
\end{equation}
Writing $\mathbf{\Omega}_P(x)=\Omega_P(x)\,dx_1\cdots dx_m$, one has~\cite{Gaetz}
\begin{equation}\label{eq:omega_vol}
    \Omega_P(x)  = {\rm vol}(P-x)^{\circ}  := \int_{(P-x)^{\circ} } dy_1 \cdots dy_m \, , \quad \forall \, x \in \operatorname{int}(P) \, .
\end{equation}
Thus the canonical function of a polytope is literally a dual-volume function. As $x$ approaches $\partial P$, the shifted polar becomes unbounded and, correspondingly, its volume~\eqref{eq:omega_vol} diverges with the singularities of $\Omega_P$ required by the definition of positive geometry~\cite{Positive_geometries}.

The dual-volume formula admits a homogeneous reformulation which makes a stronger positivity property manifest~\cite{Henn:CM}. Let
\begin{equation}
    \widehat{P} := \{\lambda x : x \in P,\ \lambda>0\}
\end{equation}
be the pointed cone over $P$, and let
\begin{equation}\label{eq:dual_P}
    \widehat{P}^* := \{y \in \mathbb{R}^{m+1} : \langle y,x \rangle \geq 0 \text{ for all } x \in \widehat{P}\}
\end{equation}
be its closed dual cone. Then~\cite{Gaetz}
\begin{equation}\label{eq:omega_Laplace}
    \Omega_{\widehat{P}}(x) = \int_{\widehat{P}^*} e^{-\langle x,y\rangle}\,dy_1\cdots dy_{m+1}, \qquad x\in\widehat{P}.
\end{equation}
It follows immediately that $\Omega_{\widehat{P}}$ is \textit{completely monotone}: 
\begin{equation}\label{eq:omega_CM}
    (-1)^r D_{v_1}\cdots D_{v_r}\Omega_{\widehat{P}}(x)\geq0, \qquad x,v_i\in\widehat{P}.
\end{equation}
More generally, the Bernstein--Hausdorff--Widder--Choquet theorem characterizes completely monotone functions on an open convex cone precisely as Laplace transforms of positive Borel measures supported on its closed dual cone; see Theorem~\ref{thm:BHWC}.

The connection between positive integral representations and complete monotonicity was emphasized in~\cite{Henn:CM} in the context of scattering amplitudes and Feynman integrals. In particular, the dual-volume representation of a polytope implies that its canonical function is completely monotone on the cone over the polytope. Applied to Hodges' construction~\cite{Hodges}, this shows that the corresponding tree-level gluon amplitude is completely monotone, since it is represented as the volume of a polytope in $\mathbb{P}^3_{\mathbb{R}}$. This observation provides a starting point for the present work. We therefore ask whether positive geometries beyond polytopes admit similar dual-volume representations.

This perspective is closely related to the search for a \textit{dual amplituhedron}~\cite{Positive_geometries,Herrmann:2020qlt,Ferro}. Such a construction should contain two ingredients: a suitable dual semialgebraic set, and a non-negative measure supported on it whose Laplace transform reproduces the canonical function. For positive geometries in projective space, the first ingredient is provided by convex duality, while for amplituhedra in the Grassmannian the appropriate notion of duality is subtler. A proposal for this geometric ingredient was recently given in~\cite{Mazzucchelli:2025kzy}. In the present paper we focus instead on the second ingredient---the positive measure---restricting first to full-dimensional positive geometries in projective space. Extending this analysis to positive geometries in the Grassmannian is left for future work.

Completely monotone positive geometries also admit a natural interpretation in algebraic statistics. By the Bernstein--Hausdorff--Widder--Choquet theorem, the canonical function is the Laplace transform of a unique positive Borel measure $d\nu_P$ on the dual cone. Although $d\nu_P$ need not have finite total mass, it can be viewed as a base measure whose Laplace transform is a partition function; this then yields an exponential family. This connects completely monotone positive geometries to rational exponential families~\cite{Exponential_varieties}. In addition, the canonical function provides a barrier function for the cone $\widehat P$, linking the same structure to interior-point methods in convex optimization~\cite{guler1996barrier}.

In this paper, we investigate to what extent the dual volume representation of polytopes extends to positive geometries beyond polytopes. Namely, we ask when the canonical function of a positive geometry in projective space can be represented as the Laplace transform of a positive measure supported on its dual cone. More generally, for hyperbolic positive geometries we study the corresponding inverse Laplace distribution and ask when it is in fact a positive measure. We now summarize our main results.

We call a positive geometry $P$ \textit{completely monotone} if its canonical function is completely monotone on the pointed cone $\widehat P$, see Definition~\ref{def:CM_PG}. Complete monotonicity of inverse powers of homogeneous polynomials is closely related to hyperbolicity~\cite{Scott_2014,Pos_Certif}. In Theorem~\ref{thm:rational_fct_CM} we extend this relation to powers of rational functions and obtain the following necessary condition for a positive geometry to be completely monotone.

\begin{theorem}[Corollary \ref{cor:CM_PG_hyperbolic}]
If $(\mathbb{P}^m,P)$ is a completely monotone positive geometry, then its algebraic boundary, the Zariski closure of its Euclidean boundary, is cut out by a hyperbolic polynomial with a hyperbolicity region equal to~$P$.
\end{theorem}

We call a positive geometry \textit{hyperbolic} if its algebraic boundary is cut out by a hyperbolic polynomial and $P$ is a corresponding hyperbolicity region. In this case, the inverse Laplace transform of its canonical function is naturally a distribution supported on the closed dual cone. This follows from the theory of hyperbolic differential operators with constant coefficients~\cite{Garding_1959,1Garding_1970,hormander1}, which we review in Subsection~\ref{subsec:Hyperbolic polynomials}. For a hyperbolic polynomial $p$, the inverse Laplace transform of $p^{-1}$ is the associated Riesz distribution, while for a rational function $q/p$ the corresponding distribution is obtained by applying $q(\partial)$, see Theorems~\ref{thm:hyp_PDE} and~\ref{thm:Riesz_p/q}. When this distribution is positive, it is a measure and the canonical function is completely monotone. Thus every completely monotone positive geometry is hyperbolic, while hyperbolicity provides a broader distributional framework. In particular, complete monotonicity implies hyperbolicity, but the converse need not follow from the existence of the associated distribution.

A particularly tractable class is provided by hyperbolic polynomials admitting symmetric determinantal representations~\cite[Section 4]{Pos_Certif}. Their Riesz measures are governed by the Riesz--Wishart theory on the cone of positive-semidefinite matrices, which we review in Subsection~\ref{subsec:determinantal representations and spectrahedral shadows}. In particular, this gives a class for which complete monotonicity can be established explicitly, and yields the following result for positive geometries.

\begin{theorem}[Corollary~\ref{cor:simplex_det_is_CM}]
    Let $(\mathbb{P}^m,P)$ be a full-dimensional simplex-like positive geometry. If $\widehat{P}$ is a minimal spectrahedral cone, then $P$ is a completely monotone positive geometry.
\end{theorem}

Here \textit{simplex-like} means that the algebraic boundary of $P\subset\mathbb P^m$ has degree ${m+1}$, while \textit{minimal spectrahedral} means that this boundary is defined by a polynomial admitting a symmetric determinantal representation for which $\widehat P$ is a hyperbolicity cone. The half-pizza geometry of Subsection~\ref{subsec:A line and a conic} provides a basic example.

We then compute explicit dual volume representations for several positive geometries in the projective plane, bounded by lines and conics or by a nodal cubic. For geometries bounded by lines and conics, the representing measures are expressed in terms of logarithms, or equivalently inverse tangent functions, while the nodal cubic leads to an elliptic integral. We summarize the main results of Section~\ref{sec:Examples and computations} below.
 
\begin{enumerate}

\item (One line and one conic) We compute using convolutions the Riesz measure for the canonical function of the positive geometry bounded by a line and a conic. This geometry is simplex-like and minimal spectrahedral; see~\eqref{mu_alpha}-\eqref{eq:mu_half_pizza} and Figures~\ref{fig:half_pizzas},~\ref{fig:linequadric}.

\item (Many lines and one conic) We consider \textit{polycons}, semialgebraic sets in the projective plane bounded by lines and a single conic, see Definition~\ref{def:polycon_r_s}. Using a triangulation-based argument, we prove:
\begin{theorem}[Theorem~\ref{thm:polycones_CM}]
Every hyperbolic polycon $(\mathbb P^2,P)$ bounded by lines and one conic is a completely monotone positive geometry, and its representing measure is given explicitly by~\eqref{eq:mu_r_s}.
\end{theorem}
Examples are given in~\ref{eg:curvy_pentagon},~\ref{eg:curvy_k_gons} and Figures~\ref{fig:pizza_slice2},~\ref{fig:curvy_pentagon},~\ref{fig:line_conic_polypols}.

\item (More conics) We then consider hyperbolic polycons bounded by any number of lines and conics. In Proposition~\ref{prop:ext_tr_con_lines_conics} we give an algorithm for computing the inverse Laplace transform of their canonical function from a signed decomposition into simpler geometries. For the examples in~\ref{eg:curvy_2gon} and~\ref{eg:curvy_4gon}, the resulting measures are numerically non-negative, motivating Conjecture~\ref{conj:hyperbolic_polycons_CM}.

\item (Nodal cubic) Finally, we consider the hyperbolic positive geometry bounded by a nodal cubic. It is simplex-like and minimal spectrahedral, and hence completely monotone. We compute its Riesz measure explicitly as a one-fold integral involving a complete elliptic integral; see Figure~\ref{fig:cubic} and~\eqref{eq:cubic_mu_period_2}.


\end{enumerate}

The explicit measures exhibit additional structure beyond their support on the dual cone. For polycons bounded by lines and a single conic, the non-constant part of the measure is expressed through an inverse tangent whose argument is the ratio of two homogeneous polynomials. The vanishing loci of these polynomials display intriguing interpolation properties with the algebraic boundaries of $\widehat P$ and $\widehat P^*$; see Figures~\ref{fig:pizza_slice2},~\ref{fig:curvy_pentagon} and~\ref{fig:line_conic_polypols}. We call these polynomials \textit{dual letters}, see Definition~\ref{def:dual_symbol}. For polycons with more than one conic, the explicit examples provide numerical evidence for non-negativity of the measure and motivate Conjecture~\ref{conj:hyperbolic_polycons_CM}.

The paper is organized as follows. Section~\ref{sec:Review on positive geometries and complete monotonicity} reviews positive geometries, convexity, and complete monotonicity. In Section~\ref{sec:Completely monotone positive geometries} we introduce completely monotone positive geometries and relate them to hyperbolic polynomials, hyperbolic differential operators, and spectrahedral cones. Section~\ref{sec:Examples and computations} contains the explicit computations for geometries bounded by lines and conics and for the nodal cubic. We conclude in Section~\ref{sec:Open problems} with a collection of open questions.

\section{Review of positive geometries and complete monotonicity}
\label{sec:Review on positive geometries and complete monotonicity}

We first review basic notions such as semialgebraic sets and cones and their convexity. We then introduce the definition of positive geometry, and later the definition of a completely monotone function on a cone.

\subsection{Semialgebraic sets, cones and convexity}
\label{subsec:semialgebraic sets and cones}

We fix some notation and elementary definitions.
We denote by $\mathbb{P}^m$ and $\mathbb{P}^m_\mathbb{R}$ the $m$-dimensional complex and real projective spaces, respectively. Given a homogeneous polynomial $p \in \mathbb{R}[x_1,\dots,x_n]$ with real coefficients we denote by $V(p) \subset \mathbb{R}^n$ the cone given by the vanishing locus of $p$ and by $\mathbb{P}V(p) \subset \mathbb{P}_\mathbb{R}^{n-1}$ the real projective variety given by the image of $V(p) \setminus \{0\}$ under the canonical projection map $\pi: \mathbb{R}^{n} \setminus \{0\} \to \mathbb{P}_\mathbb{R}^{n-1}.$

A \emph{basic semialgebraic cone} $C$ in $\mathbb{R}^{n}$ is a subset defined by homogeneous equations and inequalities. A \emph{semialgebraic cone} is a finite boolean combination of basic semialgebraic cones. A \emph{semialgebraic set} in $\mathbb{P}_\mathbb{R}^m$ is the image of a semialgebraic cone under $\pi.$

We call a subset $P \subset \mathbb{P}_\mathbb{R}^m$ \emph{convex} if it is of the form $\mathbb{P}(C)$ for some convex set $C \subset \mathbb{R}^{m+1} \setminus \{0\}$, i.e. the image of $C$ under $\pi$. We also say that a subset $P \subset \mathbb{P}_\mathbb{R}^n$ is \textit{very compact} if there is a real hyperplane $H \subset \mathbb{P}_\mathbb{R}^n$ such that $H \cap P = \emptyset$.

We call a cone $C \subset \mathbb{R}^{n}$ \textit{pointed}, if it does not contain any line, i.e. if $C \cap (-C) = \{0\}$. If $P \subset \mathbb{P}_\mathbb{R}^m$ is very compact, then it is equal to $\mathbb{P}(C)$ where $C$ is the union of two pointed cones $\widehat{P}$ and $-\widehat{P}$. Note that $\widehat{P}$ is defined uniquely up to an overall sign. We call $\widehat{P}$ the (pointed) \textit{cone over $P$}. If $P$ is semialgebraic, $\widehat{P}$ is obtained by homogenizing the equations cutting out $P$ in such an affine chart and requiring that the homogenizing variable is positive. Note that if $P$ is connected, very compact and convex, then $\widehat{P}$ is a pointed convex cone. All projective semialgebraic sets considered below are very compact. Hence, after choosing a hyperplane disjoint from $P$, they may be viewed in an affine chart and lifted to pointed cones in one dimension higher.

Given a cone $C \subset \mathbb{R}^n$, we define its \textit{dual cone} to be the closed convex cone
    \begin{equation}\label{eq:dual_cone_general}
        C^* := \{y \in \mathbb{R}^{n} \, : \, \langle y,x \rangle \geq 0 \, , \ \forall \, x \in C \} \, .
    \end{equation}
If $C$ is a full-dimensional pointed convex cone, then $C^*$ is a closed, full-dimensional, pointed convex cone with nonempty interior ${\rm int}(C^*)$. Throughout the paper, the primal cones $C$ and $\widehat{P}$ are open, while $C^*$ always denotes the corresponding closed convex dual. 
If $P \subset \mathbb{P}_\mathbb{R}^m$ is very compact, we denote by $P^*$ the semialgebraic set given by $\mathbb{P}(\widehat{P}^*) \subset \mathbb{P}^m_\mathbb{R}$.

\subsection{Positive geometries}
\label{subsec:Positive geometries}

We follow the definition of positive geometry in~\cite{Positive_geometries}; see~\cite{Lam:_PG_notes} for a recent review. 

Throughout this section, let $X$ be an $m$-dimensional irreducible complex projective variety defined over $\mathbb{R}$. Let $P$ be an oriented, $m$-dimensional open semialgebraic subset of the real locus $X(\mathbb{R})$. Throughout, we use $P$ for the relative interior of the regular semialgebraic set usually denoted $X_{\geq 0}$ in the positive-geometry literature, and $\overline{P}$ for its Euclidean closure.\footnote{Equivalently, in the notation of~\cite{Positive_geometries}, our $P$ is $X_{>0}=\operatorname{int}(X_{\geq 0})$. Since $X_{\geq 0}$ is assumed to be regular, it is recovered as the closure of its interior. We use the open representative because the completely monotone functions considered below are naturally defined on open cones.} The \textit{algebraic boundary} $\partial_a P$ of $P$ is the Zariski closure in $X$ of the Euclidean boundary of $P$. The irreducible components of $\partial_a P$ are prime divisors $\mathcal{D}_1, \dots , \mathcal{D}_r$ on $X$, see~\cite[Lemma 3.2(a)]{Lam:_PG_notes}. We define $D_i$ as the relative interior in $\mathcal{D}_i(\mathbb{R})$ of $\mathcal{D}_i \cap \overline{P}$. The orientation on $P$ induces an orientation on $D_i$.

\begin{definition}\label{def:PG}
    The pair $(X,P)$ is called a \textit{positive geometry} if there exists a unique meromorphic $m$-form $\mathbf{\Omega}_{P}$ on $X$, called the \textit{canonical form}, satisfying the following axioms.
    \begin{itemize}
        \item If $m>0$, $\mathbf{\Omega}_{P}$ has poles only along $\mathcal{D}_1, \dots, \mathcal{D}_r$. Moreover, $\mathbf{\Omega}_{P}$ has a simple pole along each $\mathcal{D}_i$ whose Poincaré residue $\Res_{\mathcal{D}_i}\mathbf{\Omega}_{P}$ equals the canonical form of the positive geometry $(\mathcal{D}_i, D_i)$. 
        \item If $m=0$, $P$ is a point and $\mathbf{\Omega}_P = \pm 1$. The sign is the orientation.
    \end{itemize}
\end{definition}

\begin{eg}[$X=\mathbb{P}^1$]\label{eg:interval}
    Let $P=(a,b) \subset \mathbb{R} = \{(x,1) \, : \, x \in \mathbb{R}\} \subset \mathbb{P}_\mathbb{R}^1$. This defines a positive geometry in $\mathbb{P}^1$. Its canonical form is given in the local coordinate $x$ by
    \begin{equation}\label{eq:omega_ab}
        \mathbf{\Omega}_{P} = \frac{b-a}{(x-a)(b-x)} \, dx \, .
    \end{equation}
    This form has simple poles at the boundary points $x=a$ and $x=b$. For example, $\Res_{x=a}\mathbf{\Omega}_{P}=1$ and $\Res_{x=b}\mathbf{\Omega}_{P}=-1$, confirming the properties in Definition~\ref{def:PG}.
\end{eg}

More generally, every convex projective polytope $P \subset \mathbb{P}^m$ is a positive geometry~\cite[Section 6.1]{Positive_geometries} with canonical form
\begin{equation}\label{eq:pol_form}
        \mathbf{\Omega}_P(x) =  \frac{q(x)}{\ell_1(x) \cdots \ell_{r}(x)} \, dx \, ,
\end{equation}
 where $\ell_{i}$ are the linear equations cutting out the facets of $P$, and $q$ is a polynomial of degree~$r-m-1$. In ~\eqref{eq:pol_form} we chose affine coordinates $x$ on~$\mathbb{P}^m$. The polynomial $q$ is the \textit{adjoint polynomial} of $P$. Its vanishing on the residual arrangement is reviewed in~\cite[Section 4.4]{Lam:_PG_notes}. For polytopes whose facet hyperplanes form a simple arrangement, Kohn and Ranestad showed that this condition characterizes the adjoint, up to scale~\cite{KohnRanestad2020}. More recently, Br\"user and Weigert proved the corresponding uniqueness statement for arbitrary convex polytopes: the adjoint is uniquely determined, up to scale, by vanishing to the appropriate orders along the residual arrangement, and in fact already by the corresponding conditions on a suitable zero-dimensional subset~\cite{BruserWeigert2025}.

Recently, Brown and Dupont introduced a Hodge-theoretic notion of positive geometry~\cite{Brown:PG_Hodge}, which does not require a priori the choice of a real semialgebraic set representing the canonical form. In this paper we are interested in positivity properties of the canonical function on a real semialgebraic region, so such a choice is essential and we adopt the original framework of~\cite{Positive_geometries}. It would be interesting to determine the precise relation between the Brown--Dupont framework and the dual-volume picture considered here.

\subsection{Complete monotonicity and Choquet's theorem}
\label{subsec:Complete monotoneity and Choquet's theorem}

In this section we review the notion of complete monotonicity and the Bernstein-Hausdorff-Widder-Choquet (BHWC) theorem. The latter allows us to relate dual volume representations to complete monotonicity of canonical functions of positive geometries. For this exposition we follow~\cite{Scott_2014}. We equip $\mathbb{R}^n$ with the standard Euclidean inner product $\langle \cdot , \cdot \rangle$. Let $C \subset \mathbb{R}^n$ be an open convex cone.

\begin{definition}
     A smooth function $f \, \colon \, C \rightarrow \mathbb{R}$ is termed \textit{completely monotone} (CM) if for all $r \geq 0$, all choices of vectors $v_1 , \dots, v_r \in C$, and all $x \in C$, we have
    \begin{equation}\label{eq:CM_ineq}
        (-1)^r \, D_{v_1} \cdots D_{v_r} \, f(x) \geq 0 \, ,
    \end{equation}
    where $D_{v_i}$ denotes a directional derivative. A function $f$ is termed \textit{absolutely monotone} (AM) if the inequality~\eqref{eq:CM_ineq} holds without the factor $(-1)^r$.
\end{definition}

Note that if~\eqref{eq:CM_ineq} holds for all $v_i \in S$ in some set $S \subset \mathbb{R}^n$, then by linearity and continuity it holds also for all $v_i$ in the closed convex cone generated by $S$. This justifies the assumption on $C$ being convex.

Also, given two open convex cones $C_1, \, C_2 \subset \mathbb{R}^n$, a positive linear function $L \colon \mathbb{R}^n \rightarrow \mathbb{R}^n$ such that $L(C_1) \subset C_2$, and a CM function $f  \colon  C_2 \rightarrow \mathbb{R}$, then $f \circ L \colon C_1 \rightarrow \mathbb{R}$ is CM. 

The space of CM functions is a convex cone, closed under pointwise multiplication. It is also closed under locally uniform limits. If $f$ is CM and $\Phi \, \colon \, [0, \infty) \rightarrow [0,\infty)$ is smooth AM when restricted on $(0,\infty)$, then $\Phi \circ f$ is CM. 

An equivalent characterization of completely monotone functions on $(0,\infty)$ is given by the Bernstein–Hausdorff–Widder theorem~\cite{widder2015laplace}, and its multidimensional generalization for functions defined on cones was proven by Choquet~\cite{Choquet}.

\begin{theorem}[Bernstein-Hausdorff-Widder-Choquet]\label{thm:BHWC}
     A function $f \, \colon \, C \rightarrow \mathbb{R}$ is completely monotone if and only if there exists a unique positive Borel measure $d\nu$ supported on the closed dual cone $C^*$ satisfying
    \begin{equation}\label{eq:Choquet}
        f(x) = \int_{C^*} e^{- \langle x,y \rangle} \, d\nu(y) \, .
    \end{equation}
\end{theorem}
The support of $d\nu$ may meet the boundary of $C^*$ and can even be lower-dimensional. In many examples in this paper the measure is absolutely continuous with respect to Lebesgue measure on the interior of $C^*$; in that case there exists a non-negative measurable density $\mu$ such that
\begin{equation}
    d \nu(y) = \mu(y) \, dy \, .
\end{equation}
From Theorem~\ref{thm:BHWC} it is also immediate to deduce the following.

\begin{cor}\label{cor:hol_ext}
    If $f$ is completely monotone on $C$, then it is extendable to a holomorphic function on the complex tube $C + i \, \mathbb{R}^n$ satisfying
    \begin{equation}
        | D_{v_1} \cdots D_{v_r} \, f(x+iy) | \leq (-1)^r \, D_{v_1} \cdots D_{v_r} \, f(x) \, ,
    \end{equation}
for all $r \geq 0$, $x \in C$, $y \in \mathbb{R}^n$ and $v_1 , \dots , v_r \in C$.
\end{cor}
\begin{proof}
    By Theorem~\ref{thm:BHWC},
    \begin{equation}
        f(x+i y)=\int_{C^*}e^{-\langle x+i y,\xi\rangle}\,d\nu(\xi)
    \end{equation}
    for $x\in C$ and $y\in\mathbb{R}^n$. The integral defines a holomorphic function on the tube domain. Differentiating under the integral and using $|e^{-i\langle y,\xi\rangle}|=1$ gives the stated inequality.
\end{proof}

Let us see two elementary examples of application of the BHWC theorem.

\begin{eg}
    Let $f(x) = 1/x$ on $(0,\infty)$. It is immediate to check by differentiating that $f$ is CM. On the other hand, $f$ is CM by Theorem~\ref{thm:BHWC} since $f(x) = \int_{0}^\infty e^{-x \, y} dy$.
\end{eg}

\begin{eg}
    Let $f(x) = \log(x+1)/x$ on $(0,\infty)$. One can show that $f$ has a Laplace integral representation with measure given by an exponential integral
    \begin{equation}
        \mu(y) = \int_{y}^{\infty} \frac{e^{-t}}{t} \, dt \, , \quad \forall \, y \in \mathbb{R}_{+} \, .
    \end{equation}
    As $\mu$ is non-negative on $(0, \infty)$, $f$ is CM. Since $\log(x+1)$ is not CM, if the product of two functions is CM, it is not necessarily true that also the individual factors are CM.
\end{eg}

For recent results on complete monotonicity of logarithmic function families, including a reduction to non-negativity of associated polynomials, see~\cite{MaWeigert2025}.

We will encounter many multidimensional examples of CM functions soon. We now connect the notion of complete monotonicity to positive geometries.

\section{Completely monotone positive geometries}
\label{sec:Completely monotone positive geometries}

In this section we introduce completely monotone positive geometries and relate them to hyperbolic polynomials and their associated differential operators with  constant coefficients. We then study the determinantal case, where Riesz measures on the positive-semidefinite cone provide explicit representing measures and lead to a class of completely monotone positive geometries. We begin by passing from canonical forms in projective space to homogeneous functions on cones.

Let $(\mathbb{P}^m,P)$ be a positive geometry in projective space with canonical form $\mathbf{\Omega}_P$. 
Let $x=(x_1,\dots,x_{m+1})$ be homogeneous coordinates on $\mathbb{P}^m$. We use the standard homogeneous projective top form defined as
\begin{equation}
    \langle x\,d^m x\rangle
    :=\sum_{i=1}^{m+1}(-1)^{i-1}x_i\,
    dx_1\wedge\cdots\wedge\widehat{dx_i}\wedge\cdots\wedge dx_{m+1} \, .
\end{equation}
The canonical form of $P$ can then be written as
\begin{equation}\label{eq:can_form_cone}
    \mathbf{\Omega}_P(x) = \frac{q(x)}{p(x)} \, \langle x\,d^m x\rangle \, ,
\end{equation}
see for example \cite[Remark 1.14]{Brown:PG_Hodge},
where $p,q \in \mathbb{R}[x_1,\dots,x_{m+1}]$ are homogeneous polynomials. Since $\langle x\,d^m x\rangle$ has homogeneous weight $m+1$, projective invariance gives $\deg(p)=\deg(q)+m+1$. We call $\Omega_{\widehat{P}}(x):=q(x)/p(x)$ the \textit{canonical function} of $\widehat{P}$, or of $P$, seen as a real-valued function on the open pointed cone $\widehat{P}\subset\mathbb{R}^{m+1}$. Since $p$ cuts out the algebraic boundary, it does not vanish on $\widehat{P}$.

\begin{definition}[Completely monotone positive geometry]\label{def:CM_PG}
A projective positive geometry $(\mathbb P^m,P)$ is \textit{completely monotone} (CM) if its cone $\widehat P$ is convex and, up to an overall choice of sign, its canonical function $\Omega_{\widehat P}$ is completely monotone on $\widehat P$.
\end{definition}

Since $p$ cuts out the algebraic boundary, complete monotonicity requires ${V(p)\cap\widehat P=\emptyset}$.

\begin{eg}\label{eg:interval_CM}
For the interval of Example~\ref{eg:interval},
\begin{equation}\label{eq:Laplce_interval}
\Omega_{\widehat{(a,b)}}(x)
=\frac{b-a}{(x_1-a x_2)(b x_2-x_1)}
=\int_{\widehat{(a,b)}^*}e^{-x\cdot y}\,dy \, ,
\end{equation}
where the dual cone over the interval is explicitly given by
\begin{equation}
    \widehat{(a,b)}^*
=\{(y_1,y_2)\in\mathbb R^2:
 a y_1+y_2\ge0,\ b y_1+y_2\ge0\} \, .
\end{equation}
Hence the interval $(a,b)$ is a completely monotone positive geometry by Theorem~\ref{thm:BHWC}.
\end{eg}

More generally, combining~\eqref{eq:omega_Laplace} with Theorem~\ref{thm:BHWC} we obtain the following.

\begin{theorem}\label{thm:polytopes_CM}
Every convex projective polytope $(\mathbb P^m,P)$ is a CM positive geometry.
\end{theorem}


\begin{figure}[t]
\begin{minipage}[c]{0.33\linewidth}
\includegraphics[width=\linewidth]{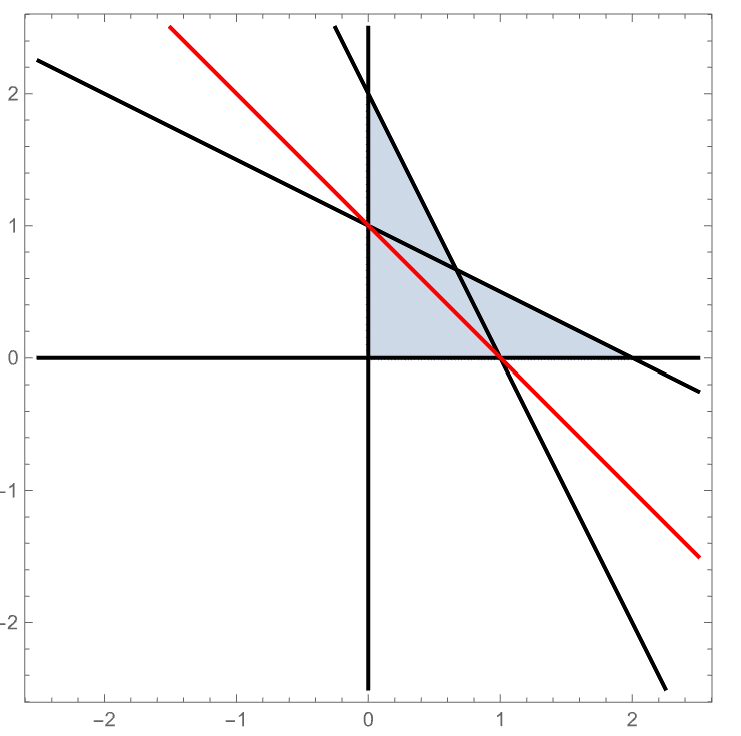}
\end{minipage}
\hfill
\begin{minipage}[c]{0.33\linewidth}
\includegraphics[width=\linewidth]{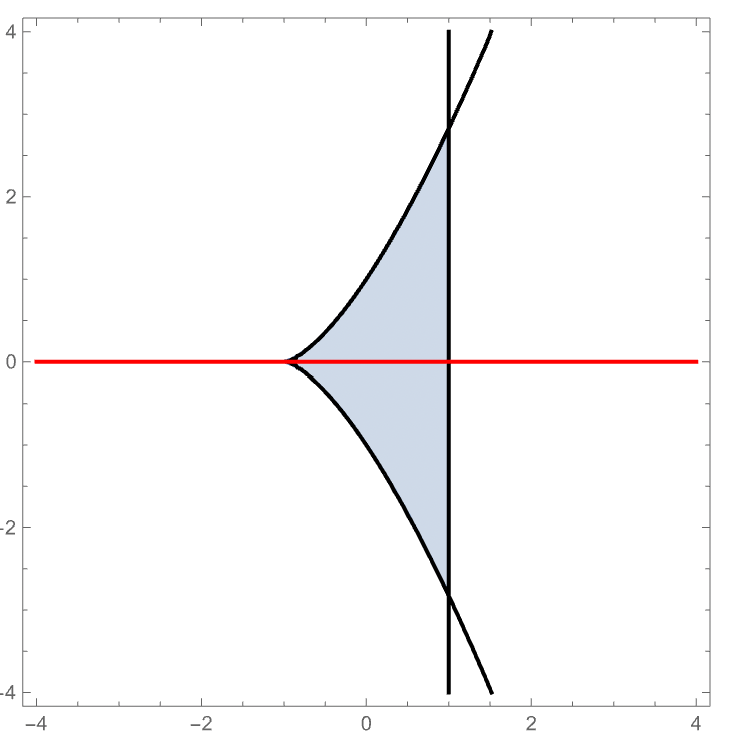}
\end{minipage}%
\hfill
\begin{minipage}[c]{0.33\linewidth}
\includegraphics[width=\linewidth]{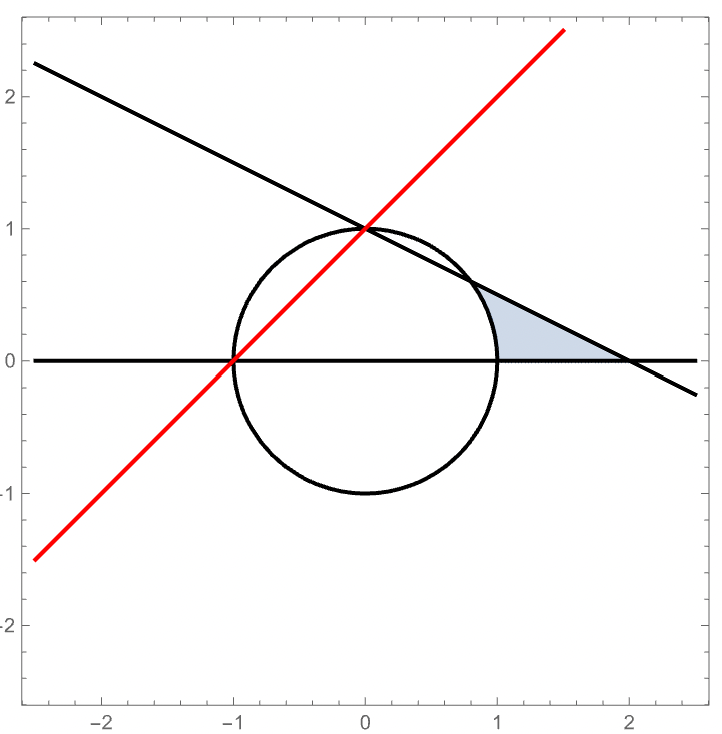}
\end{minipage}%
\caption{Three positive geometries in $\mathbb{P}^2$, given in an affine chart by the shaded regions. The components of the algebraic boundary are colored in black, while the adjoint hypersurface in red. The first two examples are not positively convex, as the algebraic boundary or the adjoint hypersurface intersect the interior of the semialgebraic set. This is not the case for the last example, which is in fact positively convex, but not convex. None of these examples is a completely monotone positive geometry, since any such is convex.}
\label{fig:non_convex_polypols}
\end{figure}

\subsection{Positive convexity}
\label{subsec:Positivity and convexity}

Following~\cite[Section 9]{Positive_geometries}, we recall:

\begin{definition}[Positively convex positive geometry]\label{def:positive_convexity}
A positive geometry $(\mathbb P^m,P)$ is \textit{positively convex} if its canonical function $\Omega_{\widehat P}$ can be chosen positive and regular on $\widehat P$.
\end{definition}

Every completely monotone positive geometry is positively convex. Writing the canonical form as in~\eqref{eq:can_form_cone}, we call $q$ the \textit{adjoint polynomial} and its projective zero locus $\mathcal A(P)$ the \textit{adjoint hypersurface}.

\begin{lemma}\label{lemma:positive_convexity}
A positive geometry $(\mathbb P^m,P)$ is positively convex if and only if
\begin{equation}
    \mathcal A(P)\cap P=\emptyset \, ,
\qquad
\partial_a P\cap P=\emptyset \, .
\end{equation}
\end{lemma}

Hence a positively convex $P$ is a connected component of $\mathbb P^m_{\mathbb R}\setminus\partial_aP(\mathbb R).$

\begin{eg}
    By Lemma~\ref{lemma:positive_convexity}, every projective polytope is positively convex. For more general positive geometries, positive convexity does not coincide with ordinary convexity; see the last example in Figure~\ref{fig:non_convex_polypols}.
\end{eg}

\subsection{Hyperbolic polynomials}
\label{subsec:Hyperbolic polynomials}

We now pass from positive representing measures to a broader distributional framework. Throughout this subsection, a \emph{distribution} means an element of $\mathcal D'(\mathbb R^n)$ in the sense of Schwartz; locally integrable functions are identified with their associated regular distributions, and every positive distribution is a positive Radon measure~\cite{hormander1}. We use the term \emph{Riesz measure} only when the corresponding Riesz distribution is positive~\cite{10.2969/jmsj/05210161}.

In the hyperbolic setting, the inverse Laplace transform of the canonical function is naturally a distribution supported on the closed dual cone $C^*$. When this distribution is positive, Theorem~\ref{thm:BHWC} gives complete monotonicity. We review the hyperbolic polynomials and differential operators with constant coefficients underlying this construction. 

\begin{definition}[Hyperbolic polynomial]\label{def:hyperb_pol}
A homogeneous polynomial $p\in\mathbb R[x_1,\ldots,x_n]$ is \textit{hyperbolic with respect to} $e\in\mathbb R^n$ if $p(e)>0$ and, for every $x\in\mathbb R^n$, the polynomial
\begin{equation}
    t\longmapsto p(x+t e) \, ,
\end{equation}
has only real zeros. The connected component $C_e$ of $\mathbb R^n\setminus V(p)$ containing $e$ is an open convex cone, called the \textit{hyperbolicity cone associated to $e$}. In particular, $C_e\neq\varnothing$. The polynomial $p$ is hyperbolic with respect to every vector in $C_e$, while different hyperbolicity directions may determine different cones.

Equivalently, $C$ is a hyperbolicity cone of $p$ if and only if
\begin{equation}
    p(z)\neq0 \, ,\qquad z\in C+i\mathbb R^n \, .
\end{equation}
We say that $\mathbb PV(p)\subset\mathbb P^{n-1}_{\mathbb R}$ is \textit{hyperbolic} with \textit{hyperbolicity region} $\mathbb P(C)$ when $C$ is a hyperbolicity cone of $p$.
\end{definition}


\begin{eg}
    The polynomial $p(x) = x_2(x_3^2-x_2^2-x_1^2)$ is hyperbolic, see Figure~\ref{fig:half_disk}.
\end{eg}

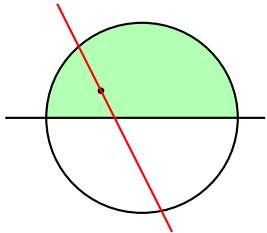
\begin{figure}[t]
    \centering
\begin{tikzpicture}[scale=1.8]

  \fill[green!30] (-1,0) arc[start angle=180, end angle=0, radius=0.7] -- (1,0) -- (-1,0) -- cycle;

  \draw[thick] (-1,0) arc[start angle=180, end angle=540, radius=0.7];
  \draw[thick] (-1.3,0) -- (0.6,0);

  \coordinate (P) at (-0.6,0.2);
  \filldraw[black] (P) circle (0.02);

  \coordinate (Q1) at ($(P)!-1.3!(-1,1)$); 
  \coordinate (Q2) at ($(P)!0.8!(-1,1)$);  

  \draw[red, thick,name path=redline] (Q1) -- (Q2);

\end{tikzpicture}
\caption{In green is the hyperbolicity region of $p(x_1,x_2,x_3)=x_2(x_3^2-x_2^2-x_1^2)$ on the affine slice $x_3=1$, and in black the vanishing locus of $p$. Any line (red) emanating from any point in the green region intersects $V(p)$ in three real points, counting multiplicities. }
    \label{fig:half_disk}
\end{figure}

Hyperbolic polynomials arise naturally in the theory of constant-coefficient partial differential equations and the Cauchy problem~\cite{Garding_1959,1Garding_1970,hormander1}; we follow~\cite{Guler_Hyperbolic_pol}. We use the Fourier transform $\mathcal F[u](\xi)=\int_{\mathbb R^n}e^{-i\langle x,\xi\rangle}u(x)\,dx$, with its standard extension to tempered distributions~\cite{hormander1}, so that for $p\in\mathbb C[\xi_1,\ldots,\xi_n]$ the associated differential operator with constant coefficients $p(-i\partial)$ satisfies
\begin{equation}
    \mathcal F\!\left[p(-i\partial)u\right](\xi)
    =p(\xi)\,\mathcal F[u](\xi)\, .
\end{equation}

\begin{definition}[Fundamental solution]
A \textit{fundamental solution} of $p(-i\partial)$ is a distribution $E\in\mathcal D'(\mathbb R^n)$ satisfying
\begin{equation}\label{eq:fund_sol}
p(-i\partial)E=\delta \, .
\end{equation}
Its support is called the \textit{propagation cone}.
\end{definition}

Formally, $E$ is an inverse Fourier transform of $p^{-1}$, with the latter interpreted distributionally; see~\cite[Section 6.2]{gelʹfand1968generalized}. For hyperbolic $p$, this can be made precise as follows.

\begin{theorem}[Theorem 2.2 \cite{Guler_Hyperbolic_pol}]\label{thm:hyp_PDE}
    Let $p \in \mathbb{R}[x_1,\dots,x_n]$ be a hyperbolic polynomial with hyperbolicity cone $C$. Then there exists a unique fundamental solution $E\in\mathcal{D}'(\mathbb{R}^n)$ to~\eqref{eq:fund_sol} with support contained in the closed dual cone $C^*$. It is the inverse Fourier transform
    \begin{equation}\label{eq:fund_sol_int_repr}
        E=\mathcal{F}^{-1}\!\left[p_-^{-1}\right]
        =(2\pi)^{-n}\int_{\mathbb{R}^n}e^{i\langle\,\cdot\,,\xi\rangle}\,p_-^{-1}(\xi)\,d\xi \, ,
    \end{equation}
    where
    \begin{equation}
        p_-^{-1}:=\lim_{t\to0^+}p(\xi-it e)^{-1}
    \end{equation}
    for any $e\in C$, with the limit and the integral in~\eqref{eq:fund_sol_int_repr} understood distributionally.
\end{theorem}

More generally, for a hyperbolic polynomial $p$ one defines the \emph{Riesz distribution} associated to $p^{-\alpha}$ as the distribution supported on $C^*$ whose Laplace transform is $p^{-\alpha}$; see G$\mathring{\rm a}$rding~\cite[Theorem 3.1]{Garding_1951} and \cite[Theorem 4.8]{Pos_Certif}. When this distribution is positive, it is a positive Radon measure and we call it the \emph{Riesz measure}. For $\alpha=1$, the Riesz distribution is precisely the fundamental solution $E$ in~\eqref{eq:fund_sol_int_repr}.

\begin{rem}[Regularity of the fundamental solution]\label{rem:regularity_of_E}
The study of regularity of fundamental solutions to hyperbolic PDEs with constant coefficients goes back to G$\mathring{\rm a}$rding~\cite{Garding_1951,1Garding_1970}. The distribution $E$ is represented by a smooth function on $\mathbb{R}^n \setminus \operatorname{sing\,supp}(E)$. Its singular support is the projection to $\mathbb{R}^n$ of the \textit{wave front set} $\operatorname{WF}(E)\subset T^*\mathbb{R}^n\setminus 0$. For the hyperbolic fundamental solutions considered here, the possible singular locus is controlled by the characteristic geometry and, in the homogeneous setting, is contained in the real cone over the projective dual variety\footnote{If a complex projective variety is reducible, the dual variety is defined as the union of the projective duals of every irreducible component.} to $\mathbb{P}V(p)\subset\mathbb{P}_\mathbb{R}^{n-1}$; see~\cite[Theorems 10.2.11 and 12.6.2]{hormander2}. In particular, away from this possible singular locus the distribution $E$ may be treated pointwise as a smooth function. Moreover, Wagner~\cite[Proposition 2]{wagner4} proves that if the degree of $p$ is greater than $n$ and the real projective variety $V(p)$ is smooth, then the fundamental solution is represented by a continuous function. The microlocal singular behavior encoded by the wave front set must in general be understood distributionally.
\end{rem}

\begin{eg}[Wave equation]\label{eg:wave_eq}
    The prototypical example of a hyperbolic polynomial is the Lorentz form $p(x) = x_1^2 -x_2^2 - \dots - x_n^2$. The region in $\mathbb{R}^n$ on which $p$ is positive is the union of two open cones, and the hyperbolicity cone can be taken to be any of the two, for example 
    $C:=\{x \in \mathbb{R}^n \, : \, p(x) > 0 \, , \ x_1 > 0 \} $.
    By~\cite[Proposition 5.6]{Scott_2014} $p^{-\alpha}$ is completely monotone on $C$ for every $\alpha > (n-2)/2$, and its Riesz kernel is given by
    \begin{equation}
        \mu_{\alpha}(y) = \Big(\pi^{\frac{n-2}{2}} \, 2^{2\alpha -1} \, \Gamma(\alpha)  \,\Gamma\Big(\alpha - \frac{n-2}{2} \Big)\Big)^{-1} \, (y_1^2 -y_2^2 - \dots - y_n^2)^{\alpha - \frac{n}{2}} \,  , \ \forall \, y \in \operatorname{int}(C^*) \, .
    \end{equation}
\end{eg}

\begin{rem}[The support]
    The fact that the fundamental solution $E$ in~\eqref{eq:fund_sol} is supported on a proper cone is peculiar to hyperbolic polynomials. For comparison, take the elliptic polynomial $p(\xi)=|\xi|^2$. Then $p(-i\partial)=-\sum_{j=1}^n \partial_j^2$ is the negative Laplacian, and for $n\geq3$ a fundamental solution is given by
    \begin{equation}
        E(x)=\frac{|x|^{2-n}}{(n-2)\,\omega_n},
    \end{equation}
    where $\omega_n$ is the surface area of the $(n-1)$-sphere. Its support is all of $\mathbb{R}^n$.
\end{rem}

We began this subsection motivating the relevance of hyperbolic polynomials in relation to the dual volume representation of canonical forms. We now show that this class of polynomials is related to complete monotonicity, when the relevant function is a power of a rational function. In fact, we extend~\cite[Theorem 4.7]{Pos_Certif} to the following result.

\begin{theorem}\label{thm:rational_fct_CM}
    Let $p,q \in \mathbb{R}[x_1, \dots ,x_n]$ be homogeneous coprime polynomials that are positive on an open convex cone $C \subset \mathbb{R}^n$ and such that $f = (q/p)^\alpha$ is completely monotone on $C$ for some $\alpha>0$. Then $p$ is hyperbolic and its hyperbolicity cone contains~$C$.
\end{theorem}

\begin{proof}
    Since $f$ is completely monotone, by Corollary~\ref{cor:hol_ext} it can be extended to a holomorphic function on the tube domain $T= C+i \, \mathbb{R}^n$, which we still denote by $f$. Note that $T$ is open in $\mathbb{C}^n$. We denote by $S_p$ and $S_q$ the intersection of $T$ with the vanishing locus of $p$ and $q$ in $\mathbb{C}^n$, respectively. Since $f$ is holomorphic\footnote{We follow the argument in the proof of \cite[Corollary 2.3]{Scott_2014}. Take a simply connected open subset $U$ of $C$ whose closure is a compact subset of $C$ and satisfying $(U+ i \, \mathbb{R}^n) \cap S_p \neq 0$. We then take a tubular neighborhood $D = U + i \, B_R$ for some $R>0$, where $B_R$ is the open ball in $\mathbb{R}^n$ of radius $R>0$, such that $D \subset T \setminus S_p$ and $\overline{D} \cap S_p \neq \emptyset$. Then, by the identity theorem for holomorphic functions and holomorphic extensions, $f \big|_{D}$ coincides with $(q/p)^\alpha \big|_{D}$, as both functions are equal on $U$. The latter is however singular when approaching a point in $\overline{D} \cap S_p$, and therefore so is $f$, contradicting holomorphicity of $f$ on $T$.} on $T$, we must have that $S_p = \emptyset$ or $\emptyset \neq S_p \subset S_q$. The latter condition contradicts the fact that $p$ and $q$ are coprime\footnote{Note that if $\emptyset \neq S_p \subset S_q$, then $q$ vanishes on a (Euclidean) open subset of the vanishing locus $V$ of $p$ in $\mathbb{C}^n$. Therefore $q$ vanishes on all $V$, which means that $p$ divides $q$.}, hence $S_p = \emptyset$, i.e. $p$ never vanishes on $T$. This is equivalent to $p$ being hyperbolic with hyperbolicity cone containing $C$.
\end{proof}

\begin{cor}\label{cor:CM_PG_hyperbolic}
If $(\mathbb{P}^m,P)$ is a completely monotone positive geometry, then $\partial_a P$ is cut out by a hyperbolic polynomial with hyperbolicity region equal to $P$.
\end{cor}
\begin{proof}
   Let $p$ be the denominator of $\Omega_{\widehat P}$ as in~\eqref{eq:can_form_cone}. Then $\partial_aP$ is cut out by $p$, and by Theorem~\ref{thm:rational_fct_CM} we are left to show that $P$ is the projectivization of a hyperbolicity cone of $p$. Such a projective hyperbolicity region is a connected component $P'$ of $\mathbb P^m_{\mathbb R}\setminus\mathbb PV(p)$. By Theorem~\ref{thm:rational_fct_CM}, we have $P\subset P'$. 
    We now use the argument in the proof of~\cite[Proposition 2.9]{sinn2015algebraic}. Let $U$ be a connected component of $\mathbb{P}_\mathbb{R}^m\setminus \mathbb{P}V(p)$ that meets $P$. The set $U\cap P$ is open in $U$. It is also closed in $U$, because its relative boundary is contained in $\partial P\cap U$, while $\partial P\subset \mathbb{P}V(p)$ and $U\cap \mathbb{P}V(p)=\emptyset$. Hence $U\cap P=U$. Therefore $P$ is a union of connected components of $\mathbb{P}_\mathbb{R}^m\setminus \mathbb{P}V(p)$. Since $P$ is convex and therefore connected, it is a single connected component. Together with $P\subset P^{'}$, this gives $P=P^{'}$.
\end{proof}

This motivates the following definition.
\begin{definition}[Hyperbolic positive geometry]
We call a positive geometry $(\mathbb P^m,P)$ \textit{hyperbolic} if its algebraic boundary $\partial_aP$ is cut out by a hyperbolic polynomial and $P$ is a corresponding hyperbolicity region.
\end{definition}

\begin{rem}[Hyperbolicity of the numerator]
    In the proof of Theorem~\ref{thm:rational_fct_CM}, one could ask when is $S_q = \emptyset$, namely when is $q$ also hyperbolic. Note that if this is the case, then the hyperbolicity cone of $q$ contains $C$. In the context of positive geometries, the question translates to: when is the adjoint hypersurface of a positive geometry hyperbolic with hyperbolicity cone containing the positive geometry? In~\cite[Theorem 3.8]{Polypols} the authors prove that this is the case for every convex polygon in the projective plane. This does not generalize to higher dimensions, not even in the case of convex polytopes, see~\cite[Example 3.13]{Polypols}.
\end{rem}

For a homogeneous hyperbolic polynomial $p$,~\eqref{eq:fund_sol_int_repr} provides the Riesz distribution associated to $p^{-1}$. For a rational function $q/p$, the corresponding distribution is obtained by differentiation. This follows from the elementary compatibility of Fourier transformation with differential operators; see for example \cite[Chapter II]{gelʹfand1968generalized} or~\cite[Chapter VII]{hormander1}.

\begin{theorem}\label{thm:Riesz_p/q}
    Let $p,q \in \mathbb{R}[x_1,\dots,x_n]$ be homogeneous polynomials and assume that $p$ is hyperbolic with hyperbolicity cone $C$. Define
    \begin{equation}\label{eq:fund_sol_int_repr_2}
        T:=q(\partial)E
        =\mathcal{F}^{-1}\!\left[q(i\xi)p_-^{-1}(\xi)\right]
    \end{equation}
    in the sense of distributions. Then $T$ is supported on the closed dual cone $C^*$ and its distributional Laplace transform is
    \begin{equation}\label{eq:distributional_laplace_qp}
        \frac{q(x)}{p(x)}=\mathcal{L}[T](x) \, ,\qquad x\in C \, .
    \end{equation}
    Equivalently, whenever $T$ is represented by a locally integrable function $\mu$,~\eqref{eq:distributional_laplace_qp} becomes
    \begin{equation}
        \frac{q(x)}{p(x)}=\int_{C^*}e^{-\langle x,y\rangle}\,\mu(y)\,dy \, .
    \end{equation}
    If $T$ is positive, then it is a positive measure and $q/p$ is completely monotone.
\end{theorem}

\subsection{Determinantal representations}
\label{subsec:determinantal representations and spectrahedral shadows}

We recall only the basic facts about spectrahedral cones we need; see~\cite{BPT2013} for background.

\begin{definition}[Spectrahedral cone]
A cone $C\subset\mathbb R^n$ is \textit{spectrahedral} if there exist linearly independent real symmetric $m\times m$ matrices $A_1,\ldots,A_n$ such that
\begin{equation}\label{eq:spectrahedron}
C=\{x\in\mathbb R^n \, : \, A(x):=x_1A_1+\cdots+x_nA_n\succ0\} \, ,
\end{equation}
where $A \succ0$ means that the matrix $A$ is positive definite.
\end{definition}

Every spectrahedral cone is convex.

\begin{eg}[The ice-cream cone]
Consider
\begin{equation}\label{eq:ice_cream_C}
C=\big\{(x_1,x_2,x_3)\in\mathbb R^3:
 x_3^2-x_2^2-x_1^2>0,\ x_3>0\big\} \, .
\end{equation}
It admits the spectrahedral representation
\begin{equation}\label{eq:sps_matrix_disk}
A(x)=
\begin{pmatrix}
 x_3-x_2 & x_1\\
 x_1 & x_3+x_2
\end{pmatrix}\succ0 \, .
\end{equation}
Indeed, Sylvester's criterion gives precisely the inequalities defining $C$. Since $\operatorname{Sym}_2(\mathbb R)$ is three-dimensional, this identifies $C$ with the cone of $2\times2$ positive-definite matrices.
\end{eg}

Let $N=\binom{m+1}{2}$ and identify $\mathbb R^N$ with $\operatorname{Sym}_m(\mathbb R)$. Denote by $S_m$ the cone of positive-definite matrices and by $S_m^*$ the closed cone of positive-semidefinite matrices, dual to $S_m$ with respect to the trace pairing. The map
\begin{equation}\label{eq:iota_emb}
A:\mathbb R^n\hookrightarrow\operatorname{Sym}_m(\mathbb R) \, ,
\qquad
A(x)=x_1A_1+\cdots+x_nA_n \, ,
\end{equation}
identifies $C$ with the linear slice $A(\mathbb R^n)\cap S_m$. Dually, denoting by $\operatorname{tr}$ the trace,
\begin{equation}\label{eq:L_projection}
L=A^*:\operatorname{Sym}_m(\mathbb R)\longrightarrow\mathbb R^n \, ,
\qquad
L(B)=\big(\operatorname{tr}(A_1B),\dots,\operatorname{tr}(A_nB)\big) \, ,
\end{equation}
and $C^*$ is a linear image of the cone of positive-semidefinite matrices,
\begin{equation}
    L(S_m^*)=C^* \, ,
\end{equation}
also called a \textit{spectrahedral shadow}. 
The determinant polynomial
\begin{equation}\label{eq:spectr_pol}
p(x)=\det A(x)
\end{equation}
vanishes on $\partial C$, hence $\partial_a C\subseteq V(p)$. If $e\in C$, then $A(e)\succ0$, and the polynomial
\begin{equation}
    t\longmapsto\det(A(x)+tA(e))
\end{equation}
has only real zeros: after conjugation by $A(e)^{-1/2}$, these are the negatives of the eigenvalues of a real symmetric matrix. Thus $p$ is hyperbolic with respect to $e$. When $\partial_a C=V(p)$, the cone $C$ is the corresponding hyperbolicity cone.

\begin{definition}[Minimal spectrahedral cone]\label{def:minimal_spec}
We call a spectrahedral cone $C$ \textit{minimal} if its algebraic boundary is exactly
\begin{equation}
    \partial_a C=V(\det A(x)) \, .
\end{equation}
\end{definition}

\begin{eg}[V\'amos polynomial]
The specialized V\'amos polynomial
\begin{equation}\label{eq:vamos_pol}
q(x)=x_1^2x_2^2+4(x_1+x_2+x_3+x_4)
(x_1x_2x_3+x_1x_2x_4+x_1x_3x_4+x_2x_3x_4)
\end{equation}
is hyperbolic with respect to $e=(1,1,0,0)^T$, but no power of $q$ admits a symmetric determinantal representation~\cite{Kummer_2015}. Its hyperbolicity cone containing $e$ is nevertheless spectrahedral, and hence is not minimal in the sense above.
\end{eg}

This illustrates the distinction between hyperbolicity cones and determinantal representations. The generalized Lax conjecture asks whether every hyperbolicity cone is spectrahedral. It is open in general, but holds for $n=3$~\cite{Helton_Linear} and for elementary symmetric polynomials~\cite{branden2014hyperbolicity}.




Let us go back to the polynomial in~\eqref{eq:spectr_pol} associated to a spectrahedral cone. Raised to the appropriate negative power, this function is known to be completely monotone on its associated spectrahedral cone.
\begin{proposition}[Corollary 4.2 \cite{Pos_Certif}]\label{Prop:det_is_CM}
Let $\alpha \in \{0,\tfrac{1}{2},1,\tfrac{3}{2},\dots,\tfrac{m-1}{2}\}$ or $\alpha > \tfrac{m-1}{2}$. Then the function $p^{-\alpha}$ for $p$ as in~\eqref{eq:spectr_pol} is completely monotone on its spectrahedral cone $C$~\eqref{eq:spectrahedron}.
\end{proposition}
Let $\nu_\alpha$ denote the Riesz measure on the closed positive-semidefinite cone $S_m^*$, characterized by the equation
\begin{equation}\label{eq:det_riesz_psd}
\det(X)^{-\alpha}
=\int_{S_m^*}e^{-\operatorname{tr}(XB)}\,d\nu_\alpha(B) \, ,
\qquad X\in S_m \, .
\end{equation}
The corresponding Wishart probability distribution is obtained from $\nu_\alpha$ by exponential tilting and normalization; see~\cite{Pos_Certif,Scott_2014}. Setting $X=A(x)$ this gives
\begin{equation}\label{eq:det_nu}
p(x)^{-\alpha}
=\int_{S_m^*}e^{-\operatorname{tr}(A(x)B)}\,d\nu_\alpha(B)
=\int_{S_m^*}e^{-\langle x,L(B)\rangle}\,d\nu_\alpha(B) \, ,
\qquad x\in C \, .
\end{equation}
Thus the Riesz measure for $p^{-\alpha}$ on $C^*$ is the push-forward $\rho_\alpha:=L_*\nu_\alpha$, satisfying
\begin{equation}\label{eq:riesz_rho}
    p(x)^{-\alpha}
=\int_{C^*}e^{-\langle x,y\rangle}\,d\rho_\alpha(y) \, .
\end{equation}
Whenever $\rho_\alpha$ is absolutely continuous on $\operatorname{int}(C^*)$, we write
\begin{equation}\label{eq:density_pf}
    d\rho_\alpha(y)=\mu_\alpha(y)\,dy \, .
\end{equation}

There are two regimes in the so-called \emph{Wallach--Gindikin set} for $\alpha$ in Proposition~\ref{Prop:det_is_CM}.
\begin{enumerate}
\item First, if $\alpha>\frac{m-1}{2},$
then $\nu_\alpha$ is absolutely continuous and
\begin{equation}\label{eq:Reisz_det_1}
d\nu_\alpha(B)
=\Gamma_m(\alpha)^{-1}\det(B)^{\alpha-\frac{m+1}{2}}\,dB \, ,
\qquad
\Gamma_m(\alpha)
:=\pi^{\frac{m(m-1)}{4}}\prod_{j=0}^{m-1}
\Gamma\!\left(\alpha-\frac{j}{2}\right) \, ,
\end{equation}
with $dB$ Lebesgue measure on $\operatorname{Sym}_m(\mathbb R)$. We normalize the translation-invariant measures $dB_y$ on the fibers $L^{-1}(y)$ such that
\begin{equation}
    \int_{\operatorname{Sym}_m(\mathbb R)}f(B)\,dB
=\int_{\mathbb R^n}\left(\int_{L^{-1}(y)}f(B)\,dB_y\right)dy \, .
\end{equation}
Whenever the push-forward admits a density as in~\eqref{eq:density_pf},
\begin{equation}\label{eq:Reisz_det}
\mu_\alpha(y)
=\Gamma_m(\alpha)^{-1}
\int_{L^{-1}(y)\cap S_m^*}
\det(B)^{\alpha-\frac{m+1}{2}}\,dB_y \, ,
\qquad y\in\operatorname{int}(C^*) \, .
\end{equation}
Hence $\mu_\alpha$ is a weighted fiber integral over the spectrahedral slice $L^{-1}(y)\cap S_m^*$.

\item At the discrete points $\alpha=\frac r2$ with
$r=0,1,\dots,m-1,$
$\nu_\alpha$ is supported on positive-semidefinite matrices of rank at most $r$. For $r=0$, $\nu_0=\delta_0$. For $1\leq r\leq m-1$, let
\begin{equation}
    \phi_r:\mathbb R^{m\times r}\longrightarrow S_m^* \, ,
\qquad
\phi_r(Z)=ZZ^T \, .
\end{equation}
Then $d\nu_{r/2}$ in~\eqref{eq:det_nu} is obtained as the push-forward along $\phi_r$ of $dZ$
and hence the Riesz measure in~\eqref{eq:riesz_rho} becomes
\begin{equation}\label{eq:singular_riesz_pushforward}
\rho_{r/2}
=(L\circ\phi_r)_*\!\left(\pi^{-mr/2}\,dZ\right).
\end{equation}

\end{enumerate}

We now apply Proposition~\ref{Prop:det_is_CM} to the following class of positive geometries.

\begin{definition}[Simplex-like positive geometry]
A positive geometry $(\mathbb P^m,P)$ is \textit{simplex-like} if the polynomial $q$ in~\eqref{eq:can_form_cone} has degree zero, or equivalently, if $\mathcal A(P)=\emptyset$.
\end{definition}

This agrees with the terminology of~\cite[Section 5]{Positive_geometries}: the canonical form has no zero divisor precisely when $q$ is a nonzero constant. Since $\deg p=\deg q+m+1$, a simplex-like positive geometry in $\mathbb P^m$ has algebraic boundary of degree $m+1$.

\begin{cor}\label{cor:simplex_det_is_CM}
Let $(\mathbb P^m,P)$ be a full-dimensional simplex-like positive geometry. If $\widehat P$ is a minimal spectrahedral cone, then $P$ is a completely monotone positive geometry.
\end{cor}

\begin{proof}
Since $\widehat P$ is minimal spectrahedral, its algebraic boundary is cut out by $p(x)=\det A(x)$ as in~\eqref{eq:spectr_pol}. Since $P$ is simplex-like, its canonical function is, up to a positive scalar factor, equal to $p^{-1}$, which is completely monotone by Proposition~\ref{Prop:det_is_CM}.
\end{proof}

For a minimal spectrahedral positive geometry which is not simplex-like, Theorem~\ref{thm:Riesz_p/q} represents the canonical function by applying the numerator differential operator to the Riesz distribution of the denominator. Complete monotonicity then amounts to positivity of the resulting distribution. We encounter this situation for polycons in Section~\ref{sec:Examples and computations}.

We close this section by mentioning the polyhedral case, where $p=\prod_{i=1}^r\ell_i$
is a product of linear forms. Then $p$ admits the diagonal determinantal representation $\det\operatorname{diag}(\ell_1,\dots,\ell_r)$, and its hyperbolicity cone is polyhedral. The corresponding Riesz measures, more generally for $\prod_i\ell_i^{-\alpha_i}$, are described by Aomoto--Gelfand hypergeometric functions; see~\cite[Theorems 3.3 and 7.4]{Pos_Certif}. For $\alpha_i=1$, the Riesz kernel is piecewise polynomial on the chamber complex~\cite[Example 3.4]{Pos_Certif}. Applying the numerator operator $q(\partial)$ as in Theorem~\ref{thm:Riesz_p/q} recovers the canonical function $q/p$. For a projective polytope, Theorem~\ref{thm:polytopes_CM} identifies the resulting measure with Lebesgue measure on $C^*$.

\section{Examples and computations of measures}
\label{sec:Examples and computations}

We now compute dual-volume representations for planar positive geometries bounded by lines and conics, and for one example bounded by a nodal cubic. For a single conic and any number of lines we prove complete monotonicity and obtain explicit elementary densities, naturally expressed in terms of inverse tangent functions; for several conics we give a triangulation-based construction and formulate a positivity conjecture. The nodal cubic leads instead to an elliptic integral.

Whenever a representing measure is absolutely continuous, formulas for its density $\mu(y)$ are understood on the interior of the relevant dual region unless stated otherwise.

\subsection{A line and a conic}
\label{subsec:A line and a conic}

Consider one line and one conic in $\mathbb P^2$. After a projective transformation, the conic reads
\begin{equation}
    x_3^2-x_2^2-x_1^2=0 \, ,
\end{equation}
and together with a linear form, we define the hyperbolic polynomial
\begin{equation}
    p_a(x)=(x_3^2-x_2^2-x_1^2)(x_2+a x_3) \, ,\qquad a\geq0 \, .
\end{equation}
For $0\leq a<1$, its hyperbolicity cone is given by
\begin{equation}\label{eq:pizza_cone}
\widehat P_a=
\{x\in\mathbb R^3:
 x_3^2-x_2^2-x_1^2>0,\,
 x_2+a x_3>0,\,
 x_3>0\} \, ,
\end{equation}
while for $a\geq1$ it is the ice-cream cone $C$ of~\eqref{eq:ice_cream_C}. For $0\leq a<1$, $P_a=\mathbb P(\widehat P_a)$ is a simplex-like minimal spectrahedral positive geometry and hence is completely monotone by Corollary~\ref{cor:simplex_det_is_CM}. At $a=1$ the line is tangent to the conic, while for $a>1$ the intersection points are complex; see Figure~\ref{fig:half_pizzas}. 
\begin{figure}[t]
\centering
\subfigure[$a=1/2$]{
  \includegraphics[width=0.25\textwidth]{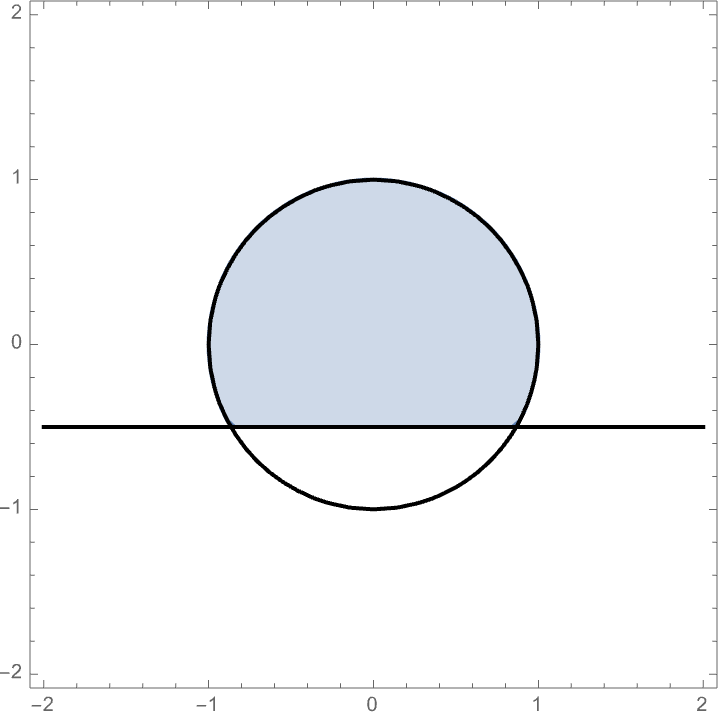}}
\hspace{0.2in}
\subfigure[dual for $a=1/2$]{
  \includegraphics[width=0.20\textwidth]{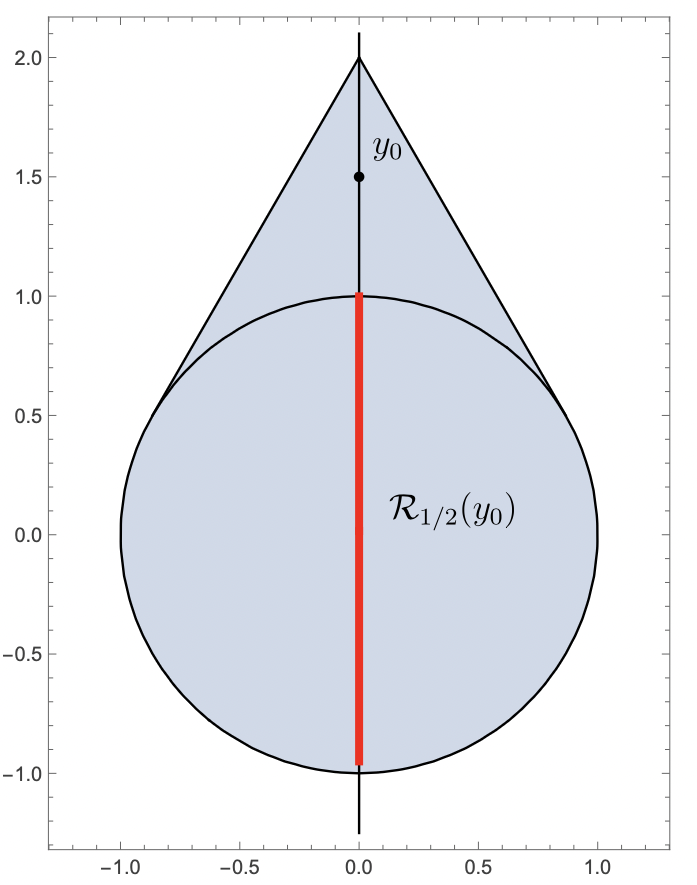}}
\subfigure[$a=1$]{
  \includegraphics[width=0.25\textwidth]{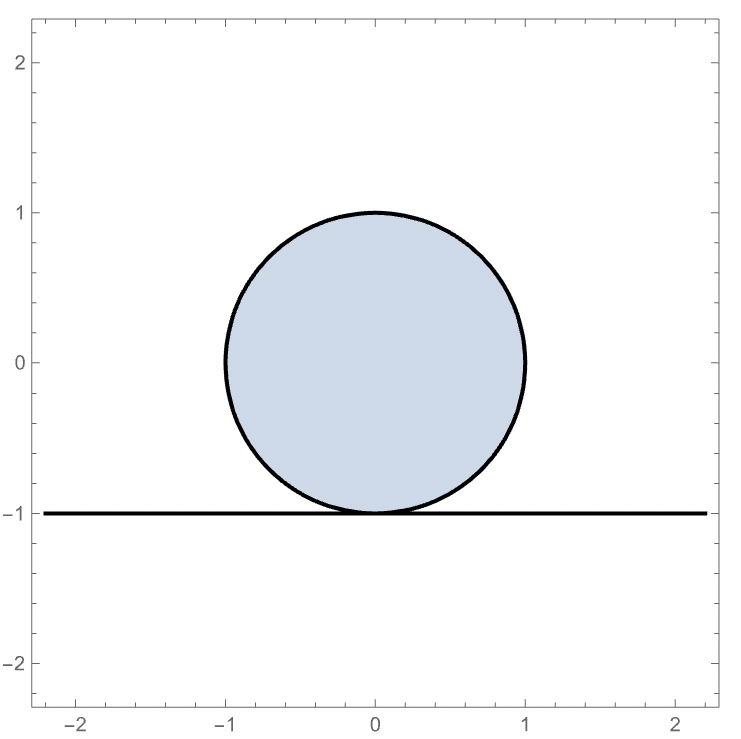}}
\hspace{0.2in}
\subfigure[$a=3/2$]{
  \includegraphics[width=0.25\textwidth]{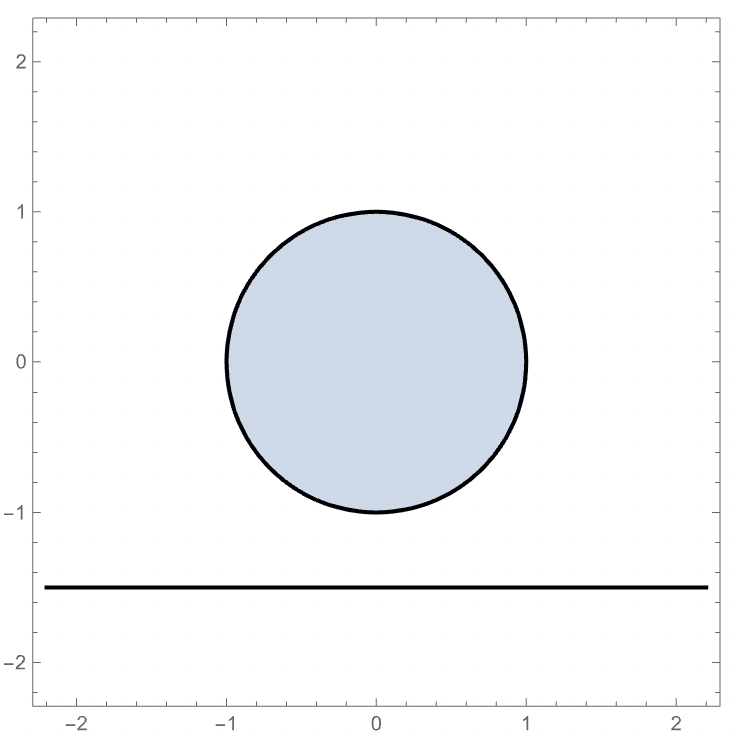}}
\caption{In (a), (c) and (d) we show $\widehat{P}_a$ from~\eqref{eq:pizza_cone} for different values of $a$ on the slice $x_3=1$. In (b) we show the dual $\widehat{P}_a^*$ for $a=1/2$. For a representative point $y_0\in\widehat{P}_{1/2}^*$, we show the line given by the projective image of the family $y_0-t\ell_{1/2}$, while the red segment is the image of the integration region $\mathcal{R}_{1/2}(y_0)$ as in~\eqref{eq:reg_Ra}.}
\label{fig:half_pizzas}
\end{figure}

By~\cite[Proposition 5.7]{Pos_Certif}, the Riesz measure of $p_a^{-1}$ is the convolution of the measures associated with its linear and quadratic factors, the latter given in Example~\ref{eg:wave_eq}. Set
\begin{equation}
    c_a:=\begin{cases}
2\sqrt{1-a^2} \, , & 0\leq a<1 \, ,\\
1 \, , & a\geq1 \, ,
\end{cases}
\end{equation}
and let $\mu_a$ be the Riesz measure of $c_a p_a^{-1}$. For $0\leq a<1$, the canonical function is 
\begin{equation}\label{eq:form_half_pizza_a}
\Omega_{\widehat P_a}(x)=\frac{2\sqrt{1-a^2}}{p_a(x)}
\end{equation}
while for $a\geq1$ we retain $p_a^{-1}$ for comparison. The convolution gives
\begin{equation}\label{eq:line_qudric_measure_1i}
\mu_a(y)=c_a\int_{\mathcal R_a(y)}
\frac{dt}{2\pi\sqrt{q_a(t,y)}} \, ,
\end{equation}
with the quadratic polynomial
\begin{equation}
    q_a(t,y)=(y_3-a t)^2-(y_2-t)^2-y_1^2 \, .
\end{equation}
For $a>0$, the integration domain in~\eqref{eq:line_qudric_measure_1i} is
\begin{equation}\label{eq:reg_Ra}
    \mathcal R_a(y)=\{t:0<t<y_3/a,\ q_a(t,y)>0\} \, ,
\end{equation}
while for $a=0$ we set $\mathcal R_0(y)=\{t>0:q_0(t,y)>0\}$.

For $a>0$, let $r_-(y)<r_+(y)$
denote the two real roots of $q_a(t,y)$ in $t$. Their discriminant is given by the polynomial
\begin{equation}
    \Delta_a(y)
=4\bigl(y_3^2-y_1^2-2ay_2y_3+a^2y_1^2+a^2y_2^2\bigr) \, ,
\end{equation}
which is positive for $y\in\widehat P_a^*.$
Then, $\mathcal R_a(y)$ is determined by the position of these roots relative to $0$ and $y_3/a$. We distinguish three cases.

\begin{enumerate}[label=(\alph*)]
\item $0<a<1$. In this case the integration domain
\begin{equation}
    \mathcal R_a(y)=(r_-(y),r_+(y))\cap(0,y_3/a) \, .
\end{equation}
On $\widehat P_a^*$, one has $y_3/a>r_\pm(y)$ and $r_+(y)>0$, while $r_-(y)>0$ precisely for $y\in\widehat P_a^*\setminus C^*$. Hence, after performing the integral in~\eqref{eq:line_qudric_measure_1i} we find
\begin{equation}\label{mu_alpha}
\mu_a(y)=\begin{cases}
\dfrac{1}{2}+\dfrac{1}{\pi}\arctan\!\left(\dfrac{y_2-a y_3}{\sqrt{(1-a^2)(y_3^2-y_2^2-y_1^2)}}\right) \, , & y\in C^* \, ,\\[6pt]
1 \, , & y\in\widehat P_a^*\setminus C^* \, .
\end{cases}
\end{equation}
The density $\mu_a$ is smooth on $\operatorname{int}(\widehat P_a^*)$.

\item $a=1$. The polynomial $q_a(t,y)$ is linear in $t$, and
\begin{equation}
    \mathcal R_1(y)=
\left(0,\frac{y_3^2-y_1^2-y_2^2}{2(y_3-y_2)}\right) \, .
\end{equation}
Therefore the measure evaluates to
\begin{equation}\label{mu_a1}
\mu_1(y)=\frac{\sqrt{y_3^2-y_1^2-y_2^2}}{2\pi(y_3-y_2)} \, ,
\qquad y\in\operatorname{int}(C^*)\setminus\{y_3=y_2\} \, .
\end{equation}
It vanishes at the boundary of $C^*$ on $y_3=y_2$, where it has a $(y_3-y_2)^{-1/2}$ singularity.

\item $a>1$. In this case we have that
\begin{equation}
    r_-(y)<y_3/a<r_+(y) \, ,
\qquad
\mathcal R_a(y)=(0,r_-(y)) \, ,
\end{equation}
and hence we compute
\begin{equation}\label{measure_alpha_bigger_one}
\mu_a(y)=\frac{1}{4\pi\sqrt{a^2-1}}
\log\!\left(
\frac{a y_3-y_2+\sqrt{(a^2-1)(y_3^2-y_1^2-y_2^2)}}
{a y_3-y_2-\sqrt{(a^2-1)(y_3^2-y_1^2-y_2^2)}}
\right).
\end{equation}
The density has a logarithmic singularity at $y_1=0$, $y_2=y_3/a$.
\end{enumerate}

\begin{figure}[t]
\centering
\subfigure[$a=1/2$]{
  \includegraphics[width=0.295\linewidth]{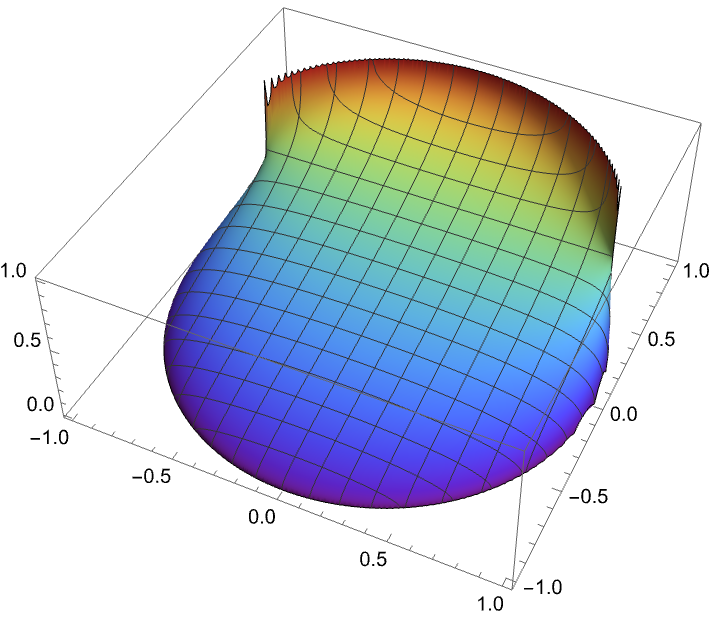}
  \label{fig:linearquadricmeas1}}
\subfigure[$a=1$]{
  \includegraphics[width=0.295\linewidth]{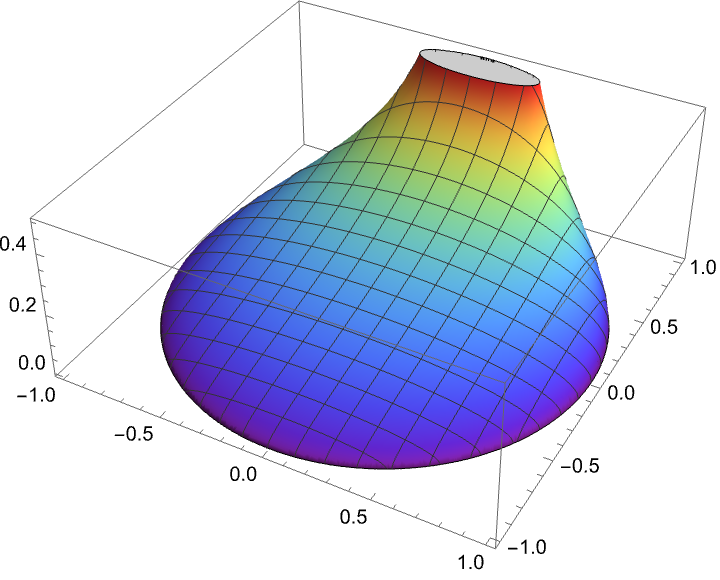}
  \label{fig:linearquadricmeas2}}
  \subfigure[$a=10$]{
  \includegraphics[width=0.295\linewidth]{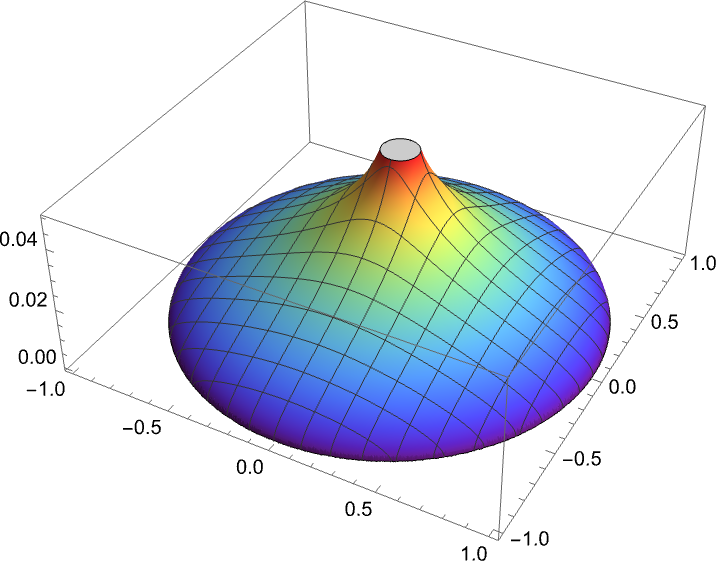}
  \label{fig:linearquadricmeas3}}
 \caption{The graph of $\mu_{a}$ in~\eqref{mu_alpha}, ~\eqref{mu_a1} and ~\eqref{measure_alpha_bigger_one}, respectively, for different values of $a$, and plotted on the slice $C^* \cap \{y_3=1\}$, where $C^*=C$ is in~\eqref{eq:ice_cream_C}. Note that for (b) and (c) the gray cuts on the vertical axis happen at locations where the measure is singular. We observe that the measure is non-negative, and smooth on the interior of $C^*$ for $a<1$ (when the line intersects the conic in two real points), while it develops a singularity for $a\geq 1$. The different colors in the plots indicate the level sets of $\mu_{a}$. }
 \label{fig:linequadric}
\end{figure}

Let us comment on the form of~\eqref{mu_alpha} for a general conic and line. Let $Q\subset\mathbb{P}^2_\mathbb{R}$ be a non-degenerate conic and let $L=V(\ell)\subset\mathbb{P}^2_\mathbb{R}$ be a line intersecting $Q$ in two distinct real points, where $\ell\in(\mathbb{R}^3)^*$ is a nonzero covector and $[\ell]\in(\mathbb{P}^2_\mathbb{R})^*$ is the corresponding point of the dual projective plane. Denote by $q \in \mathbb{R}[x_1,x_2,x_3]$ the quadratic polynomial cutting out $Q$, and write $\ell(x)=\langle\ell,x\rangle$. Assume that $q(x)>0$ defines the hyperbolicity region $C$ of $q$. Consider the semialgebraic cone
\begin{equation}
    \widehat{P} = \{ x \in \mathbb{R}^3 \, : \, q(x)>0 \, , \ \ell(x) > 0 \, , \ x_3 > 0 \} \subset C \, .
\end{equation}
This is the cone over a positive geometry in $\mathbb{P}^2$ with canonical function equal to
\begin{equation}\label{eq:rat_fct_half_pizza}
   c(\ell,q) \,  \frac{1}{\ell(x) q(x) } \, ,
\end{equation}
where the constant $c(\ell,q)$ is such that the residue of~\eqref{eq:rat_fct_half_pizza} on any of the two intersection points $L \cap Q$ is equal to $\pm 1$. If we write $q(x) = \langle x, Ax \rangle$, where $A$ is a real symmetric $3 \times 3$ matrix, then $c(\ell,q) = \sqrt{-\frac{1}{2} \varepsilon^{abc} \varepsilon^{ijk} \, A_{ai} \, A_{bj} \, \ell_c \, \ell_k }$, where $\varepsilon^{ijk}$ is the $3$-dimensional Levi-Civita tensor, see~\cite[Eq. (5.20)]{Positive_geometries}. In the following, we assume without loss of generality that $\det(A) = 1$.
Using the coordinates above to identify the dual vector space with $\mathbb{R}^3$, the dual cone $\widehat{P}^*$ is the convex hull of $C^*$ and the ray represented by the covector~$\ell$. Then, denoting by $\mu_{\widehat{P}}$ the measure representing the canonical function on $\widehat{P}$, ~\eqref{mu_alpha} becomes
\begin{equation}\label{eq:mu_half_pizza}
        \mu_{\widehat{P}}(y) = \begin{cases}
            \frac{1}{2} + \frac{1}{\pi} \arctan(\frac{\widetilde{\ell}(y)}{ \sqrt{q^*(y)}}) \, , \quad \ y \in C^* \, , \\
            1 \, , \qquad \qquad \qquad \qquad \qquad \ \,  y \in \widehat{P}^{*} \setminus C^* \, ,
        \end{cases}
\end{equation}
where $q^*(y)=\langle y,A^{-1}y\rangle$ cuts out the conic $Q^*\subset(\mathbb{P}^2_\mathbb{R})^*$ projectively dual to $Q$, and
\begin{equation}\label{eq:ell_tilde}
    \widetilde{\ell}=-\frac{A^{-1}\ell}{c(\ell,q)} \, .
\end{equation}
The geometric meaning of $\widetilde{\ell}$ is the following. From the dual point $[\ell]\in(\mathbb{P}^2_\mathbb{R})^*$ there are two tangent lines to $Q^*$, touching it at points $t_1$ and $t_2$. The line $\widetilde{L}=V(\widetilde{\ell})\subset(\mathbb{P}^2_\mathbb{R})^*$ is the line through $t_1$ and $t_2$. Thus $\widetilde{\ell}$ is a covector defining this auxiliary line and should not be confused with the projective-dual point $[\ell]$. If the conic is the one cutting out the ice-cream cone~\eqref{eq:ice_cream_C}, namely $q(x)=x_3^2-x_2^2-x_1^2$, then $q^*=q$ and $C=C^*$. In this case,
\begin{equation}
    \widetilde{\ell}=\sqrt{\frac{\ell_1^2+\ell_2^2}{\ell_1^2+\ell_2^2-\ell_3^2}}\,(\ell_1,\ell_2,-\ell_3) \, .
\end{equation}

\subsection{Two lines and a conic}
\label{subsec:wo lines and a conic}

\begin{figure}[t]
\begin{minipage}[c]{0.32\linewidth}
\includegraphics[width=\linewidth]{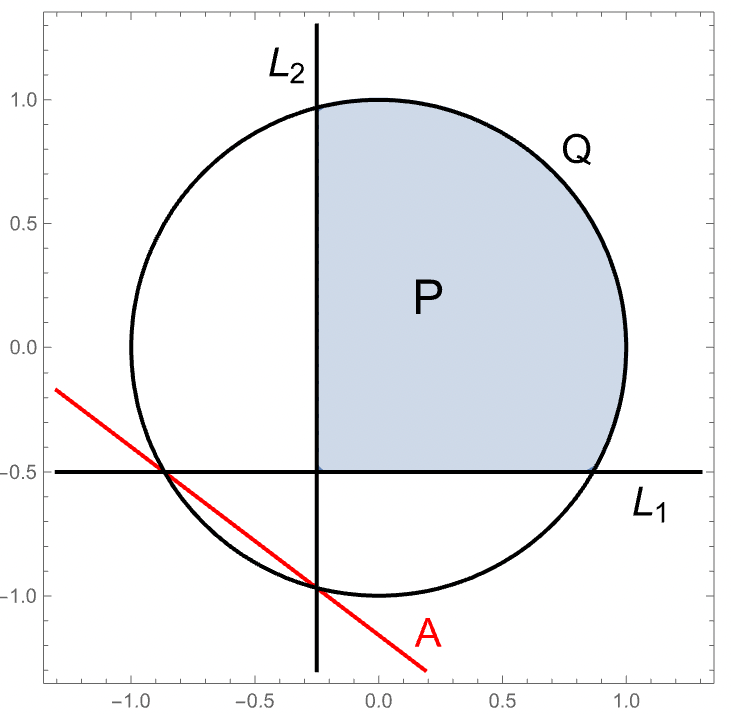}
\end{minipage}
\hfill
\begin{minipage}[c]{0.31\linewidth}
\includegraphics[width=\linewidth]{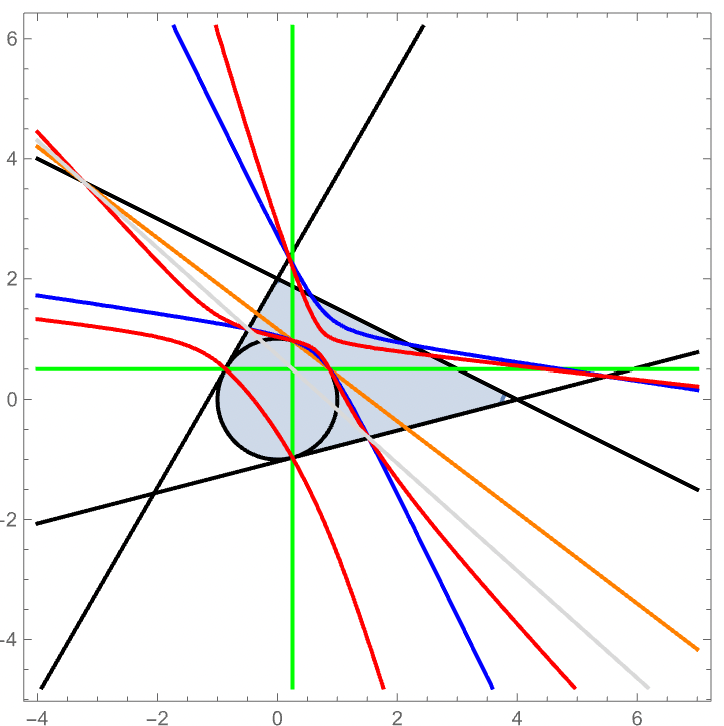}
\end{minipage}%
\hfill
\begin{minipage}[c]{0.35\linewidth}
\includegraphics[width=\linewidth]{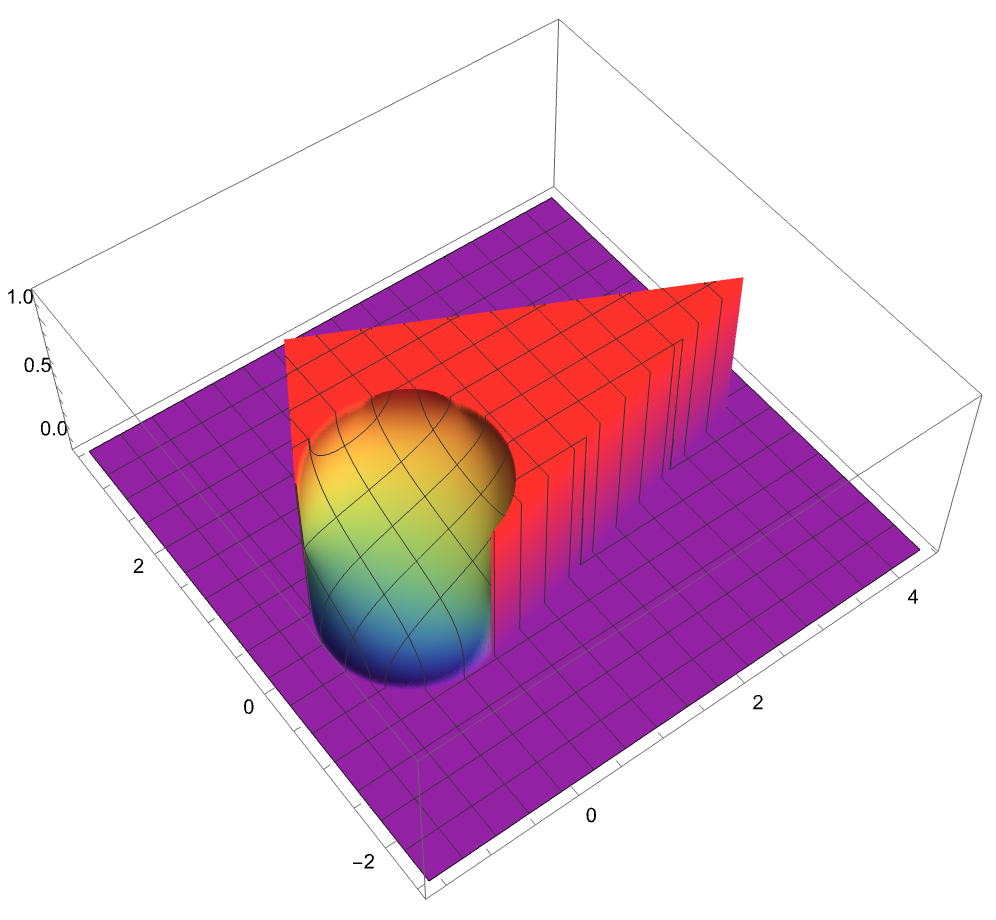}
\end{minipage}%
\caption{On the left, a positive geometry $P$ inside the hyperbolicity region $\mathbb{P}(C)$ of a conic $Q$, bounded by $Q$ and two lines $L_1$, $L_2$ cut out by $\ell_1 = (1,0,1/4)$, $\ell_2=(0,1,1/2)$, respectively. The adjoint line of $P$, cut out by $\alpha$, is drawn in red. In the middle, we show the dual $P^*$ with many curves: in green the lines $\mathbb{P}V(\widetilde{\ell}_i)$, in orange $\mathbb{P}V(\alpha^*)$, and in red and blue a cubic and conic, respectively, appearing in the argument of the inverse tangent in the last line of~\eqref{eq:mu_pizza_slice2}. The graph of the measure~\eqref{eq:mu_pizza_slice2} representing the canonical function~\eqref{eq:can_fct_pizza_slice} of $P$ is plotted on the right. This is non-negative, supported on $P^*$ and constant equal to one on $P^* \setminus \mathbb{P}(C^*)$.}
\label{fig:pizza_slice2}
\end{figure}
We next consider a positive geometry $P\subset\mathbb P^2_{\mathbb R}$ bounded by the conic $Q$ fixed above and two lines $L_i=V(\ell_i)$, with $\ell_i\in(\mathbb R^3)^*$, see Figure~\ref{fig:pizza_slice2}. By~\cite[Theorem 2.15]{Polypols}, its canonical function is, up to the canonical normalization,
\begin{equation}\label{eq:can_fct_pizza_slice}
\Omega_{\widehat P}(x) \ \propto \  
\frac{\alpha(x)}{\ell_1(x)\ell_2(x)q(x)} \, ,
\end{equation}
where $\alpha$ is linear. The dual $P^*$ is the convex hull of $Q^*$ and the dual points $[\ell_1],[\ell_2]$.\begin{figure}[t]
\begin{minipage}[c]{0.32\linewidth} \label{fig:quad2lines1}
\includegraphics[width=\linewidth]{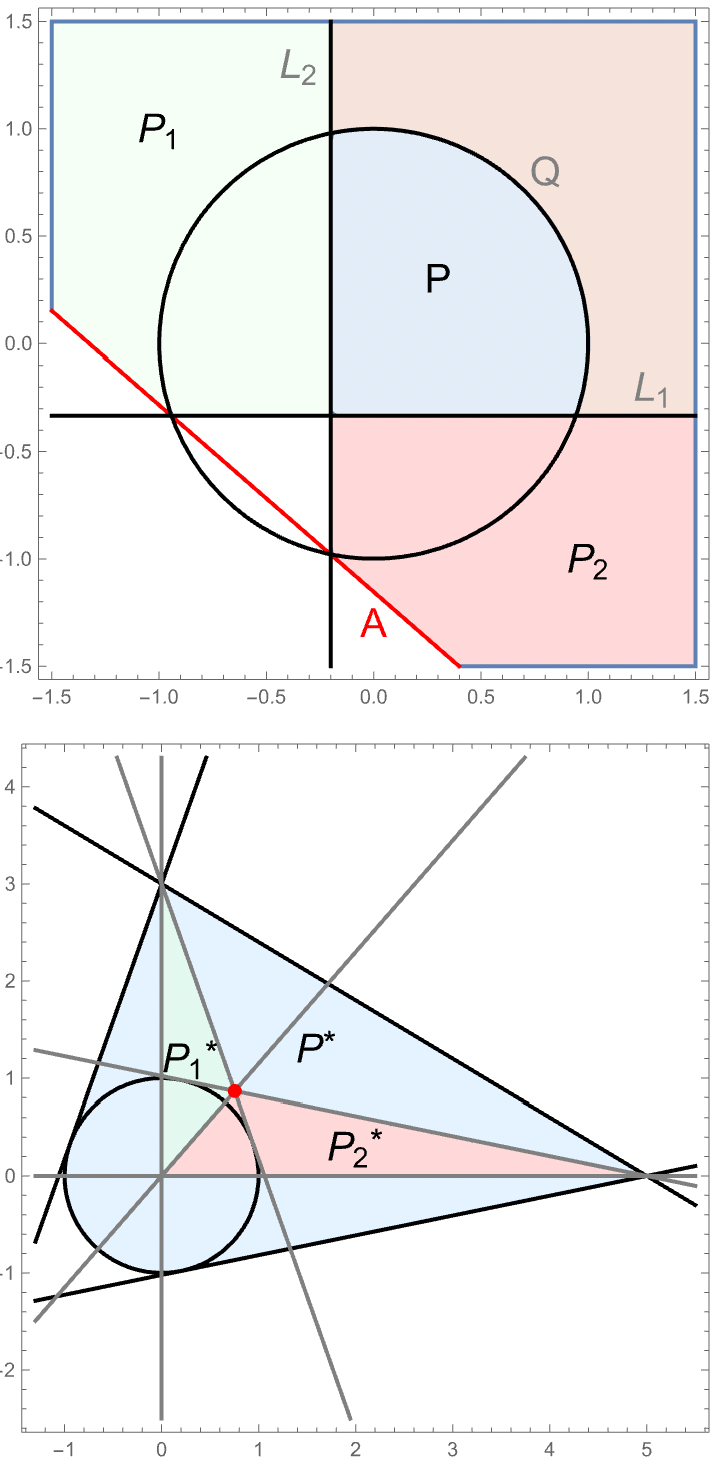}
\end{minipage}
\hfill
\begin{minipage}[c]{0.327\linewidth}\label{fig:quad2lines2}
\includegraphics[width=\linewidth]{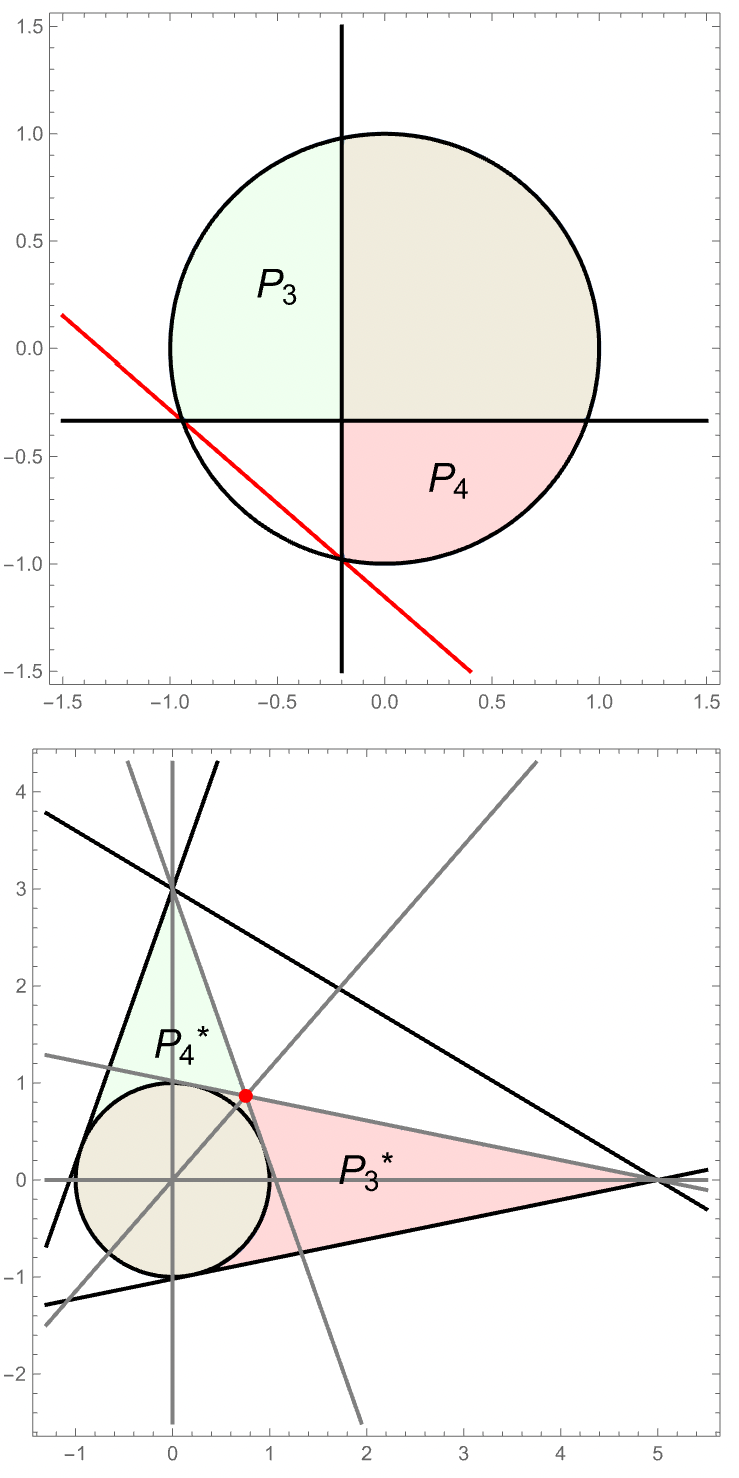}
\end{minipage}%
\hfill
\begin{minipage}[c]{0.325\linewidth}\label{quad2lines3}
\includegraphics[width=\linewidth]{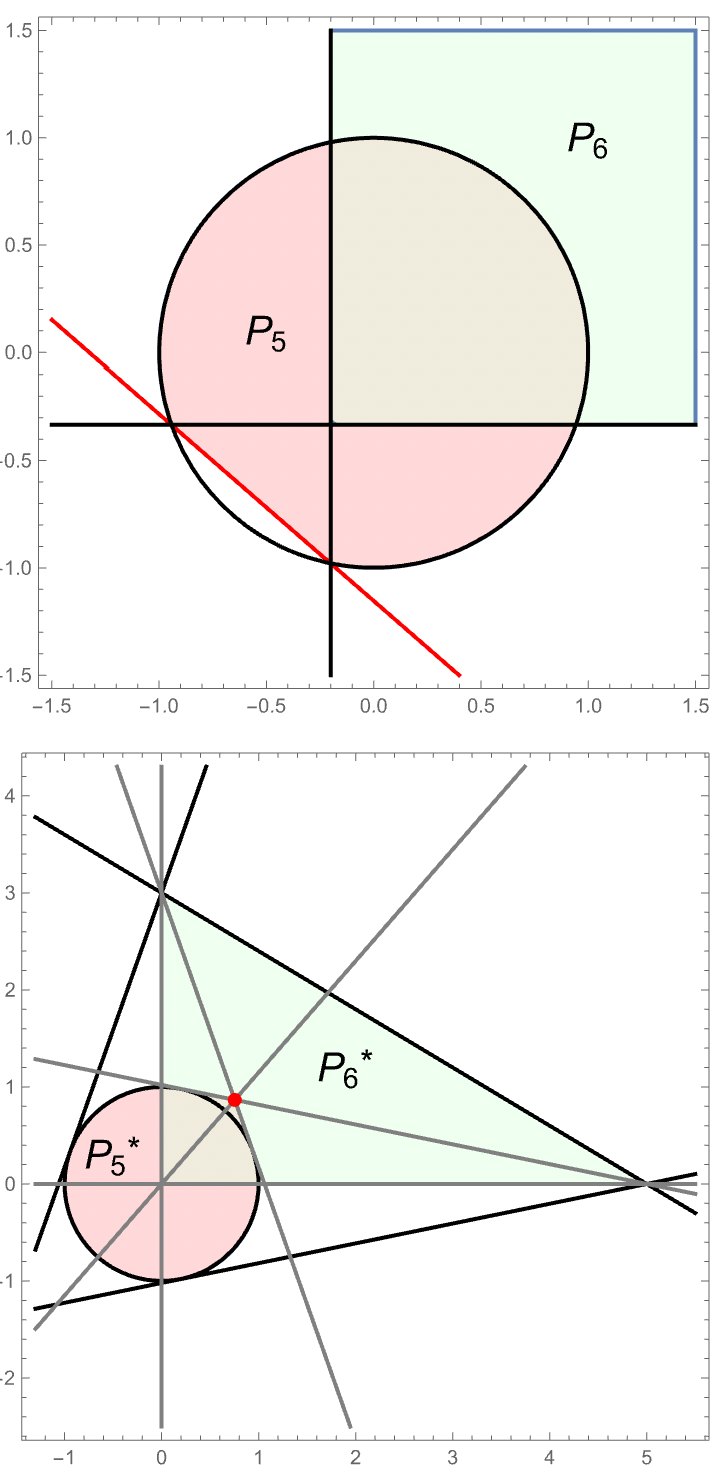}
\end{minipage}%
\caption{We show the semialgebraic sets $P_i$, above, and their duals $P_i^*$, below, for $i=1,\dots,6$ as in \eqref{eq:six_cones}. These realize the canonical form triangulation~\eqref{eq:ext_tr_pizza} of the positive geometry $P$, appearing in blue in the first picture above with its adjoint line in red. Its dual $P^*$ is in blue in the first picture below, and in the pictures below we mark in red the point projectively dual to the adjoint line. Note that the latter lies outside the conic $Q^*$; for instance, $P_5^*$ is equal to the convex hull of $Q^*$ and this point, and in the figure $P_5^*$ looks like the disk with a small horn. The sets $P^*_i$ form a \textit{signed} triangulation of $P^*$, according to~\eqref{eq:ext_tr_pizza}.}
\label{fig:pizza_tr}
\end{figure}

The measure follows from a canonical-form triangulation~\cite[Section 3.2]{Positive_geometries}. Define
\begin{equation}\label{eq:six_cones}
\begin{aligned}
C_1 &= \{\alpha>0,\ \ell_1>0,\ x_3>0\} \, ,&
C_2 &= \{\alpha>0,\ \ell_2>0,\ x_3>0\} \, ,\\
C_3 &= \{q>0,\ \ell_1>0,\ x_3>0\} \, ,&
C_4 &= \{q>0,\ \ell_2>0,\ x_3>0\} \, ,\\
C_5 &= \{q>0,\ \alpha>0,\ x_3>0\} \, ,&
C_6 &= \{\ell_1>0,\ \ell_2>0,\ x_3>0\} \, .
\end{aligned}
\end{equation}
Each $P_i=\mathbb P(C_i)$ is bounded either by three lines or by one line and one conic, and
\begin{equation}\label{eq:ext_tr_pizza}
\Omega_{\widehat P}
=-\Omega_{C_1}-\Omega_{C_2}
+\Omega_{C_3}+\Omega_{C_4}
-\Omega_{C_5}+\Omega_{C_6} \, ,
\end{equation}
with each $\Omega_{C_i}$ chosen positive on $C_i$. Since $\widehat P\subset C_i$, one has $C_i^*\subset\widehat P^*$, so the corresponding representing measures combine with the same signs. The signed decomposition~\eqref{eq:ext_tr_pizza} gives
\begin{equation}
    \mu_{\widehat P}=1
\qquad\text{on }\widehat P^*\setminus C^* \, ,
\end{equation}
while on $\operatorname{int}(C^*)$ only $C_3,C_4,C_5$ contribute. Thus
\begin{equation}\label{eq:mu_pizza_slice2}
\begin{aligned}
\mu_{\widehat P}(y)
&=\frac12+\frac1\pi\arctan\!\left(
\frac{\widetilde\ell_1(y)}{\sqrt{q^*(y)}}\right)
+\frac1\pi\arctan\!\left(
\frac{\widetilde\ell_2(y)}{\sqrt{q^*(y)}}\right)
-\frac1\pi\arctan\!\left(
\frac{\alpha^*(y)}{\sqrt{q^*(y)}}\right)\\
&=\frac12-\frac1\pi\arctan\!\left(
\frac{\widetilde{\ell}_1\widetilde{\ell}_2\alpha^*
+(\widetilde{\ell}_1+\widetilde{\ell}_2-\alpha^*)q^*}
{\sqrt{q^*}\big((\widetilde{\ell}_1+\widetilde{\ell}_2)\alpha^*
-\widetilde{\ell}_1\widetilde{\ell}_2+q^*\big)}
\right),
\end{aligned}
\end{equation}
for $y\in\operatorname{int}(C^*)$, with all polynomials in the second line evaluated at $y$. The second equality uses the inverse-tangent addition formula\footnote{We record the principal-branch identity
\begin{equation}\label{eq:arctan_addition}
\arctan(a)+\arctan(b)
=\arctan\!\left(\frac{a+b}{1-ab}\right)+\delta(a,b)\pi,
\end{equation}
where $\delta(a,b)=0$ for $ab<1$, $\delta(a,b)=1$ for $ab>1$ with $a,b\ge0$, and $\delta(a,b)=-1$ for $ab>1$ with $a,b<0$.}, with its branch fixed continuously on $\operatorname{int}(C^*)$. Since the resulting inverse tangent lies in $(-\pi/2,\pi/2)$, the density is non-negative, and $P$ is completely monotone.

The numerator and denominator appearing in the second line define the varieties plotted in Figure~\ref{fig:pizza_slice2}, which exhibit the interpolation properties discussed below.

\subsection{More lines and a conic}
\label{subsec:Any number of lines and a conic}

We now extend the construction to any number of lines bounding a single conic $Q$, with hyperbolicity cone $C$. We use the class of \textit{polycons} introduced in~\cite{Polypols}.

\begin{definition}\label{def:polycon_r_s}
A \textit{polycon of type $(r,s)$} is a full-dimensional semialgebraic set $P(r,s)\subset\mathbb P(C)$
whose boundary consists of $r\geq1$ arcs of $Q$ and $s$ line segments, each boundary line meeting $Q$ in two distinct real points.
\end{definition}
Thus $P(r,s)$ has $r$ curved and $s$ linear edges, with $r\leq s$, and is a positive geometry by~\cite[Theorem 2.15]{Polypols}.

The measure for a general polycon will be reduced to the one-line/one-conic case through canonical-form triangulations~\cite[Section 3.2]{Positive_geometries}: a finite family of semialgebraic sets $\{P_i\}$ \textit{canonical form triangulates} $P$ if
\begin{equation}
    \mathbf\Omega_P=\sum_i\mathbf\Omega_{P_i} \, .
\end{equation}
\begin{figure}[t]
\centering
\subfigure[]{
  \includegraphics[width=0.30\linewidth]{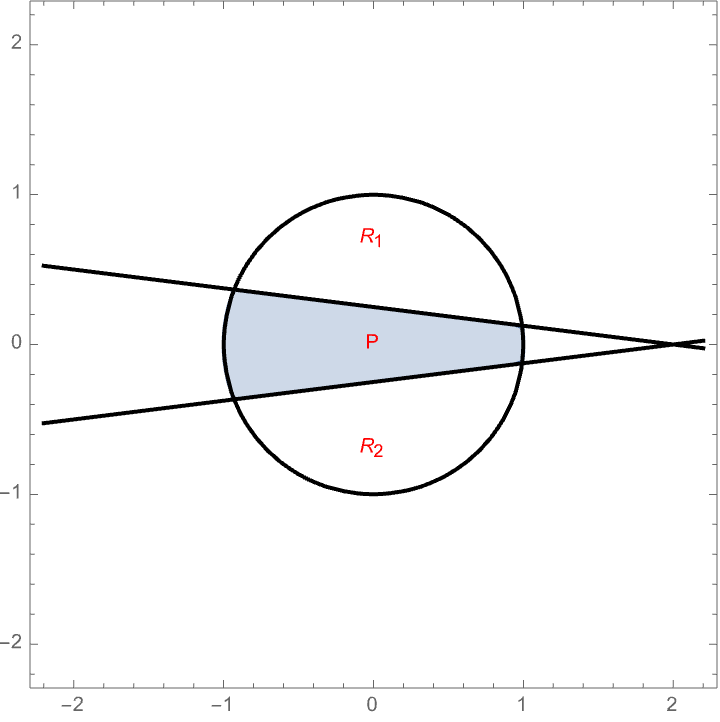}
  \label{fig:lines_circ_1}}
  \hspace{0.3in}
\subfigure[]{
  \includegraphics[width=0.30\linewidth]{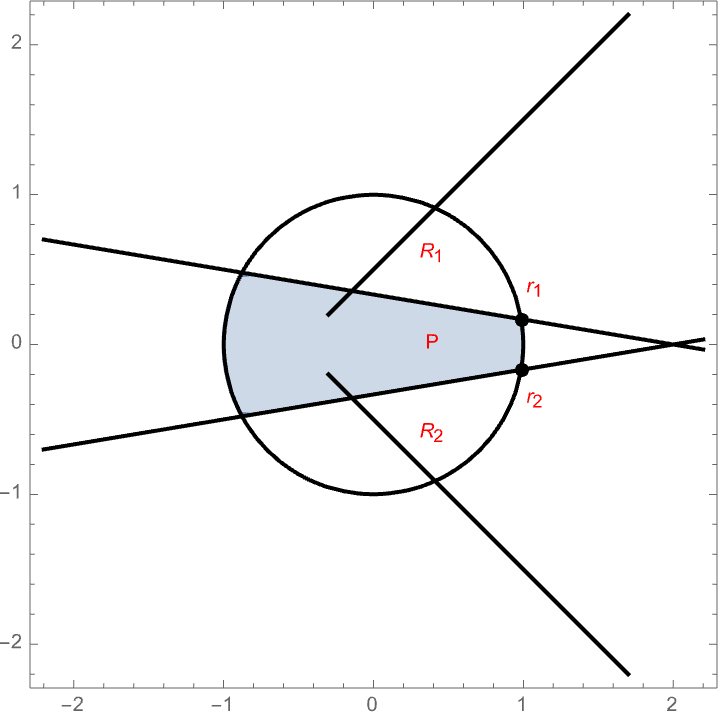}
  \label{fig:lines_circ_2}}
 \caption{Pictorial representation of the iterative argument for building a canonical form triangulation of any polycon $P(r,s)$. On the left, we have a polycon $P(2,2)$.}
\end{figure}

\begin{proposition}\label{prop:ext_tr_con_lines}
    Let $P=P(r,s)$ be a connected convex polycon of type $(r,s)$. Assume moreover that the boundary vertices of $P$ lying on $Q$ are pairwise distinct; equivalently, no two distinct boundary lines of $P$ meet $Q$ at the same boundary point of~$P$. Then, there exists a finite collection of polycons $\{P_i=P(r_i,s_i)\}$, satisfying the following for every $i$:
    \begin{enumerate}
    \item $(r_i,s_i)= (1,1)$, or $(r_i,s_i) = (0,3)$,
        \item $P\subset P_i$,
        \item $P_i$ is convex,
        \item $\Omega_{\widehat{P}} = \sum_i \varepsilon_i \, \Omega_{\widehat{P}_i}  $, where $ \varepsilon_i \in \{\pm 1\}$, with the convention that $\Omega_{\widehat{P}_i}$ is positive on~$\widehat{P}_i$.
    \end{enumerate} 
    In particular, the measures representing the corresponding canonical functions satisfy
    \begin{equation}\label{eq:measure_sum}
        \mu_{\widehat{P}} = \sum_i \varepsilon_i \, \mu_{\widehat{P}_i} \, .
    \end{equation}
\end{proposition}

\begin{proof}
    Our argument uses the fact that the region $\mathbb{P}(C)$ has vanishing canonical form in the sense of~\cite[Section 2.2]{Positive_geometries}. In particular, if a finite family of positive geometries $\{P_i\}$ form a \textit{signed triangulation} of $\mathbb{P}(C)$, then the sum of their canonical forms is equal to zero~\cite[Section 3]{Positive_geometries}. We work by induction on $s$. 
    If $s=1$, the claim is vacuously true. If $s=2$, and the two lines intersect in a point inside $C$, then this case is explicitly worked out in~\eqref{eq:ext_tr_pizza}. If instead the two lines intersect outside $C$, we take $P$ to be bounded by two arcs of $Q$ and the two lines. There are three regions inside $C$: $P$, $R_1$ and $R_2$, see Figure~\ref{fig:lines_circ_1}. Pick $P_i=P \cup R_i$ for $i=1,2$. Since the canonical form of $\mathbb P(C)$ vanishes, the signed-triangulation identity gives $\Omega_{\widehat{P}_1}+\Omega_{\widehat{P}_2}=\Omega_{\widehat{P}}$, where we always take the sign of the forms such that they are positive on the respective semialgebraic sets. For $s\geq 3$, we argue in a similar way. By assumption, $P$ has a boundary component on~$Q$. Take any connected component $S$ of $\partial P \cap Q$. The relative boundary of $S$ in $Q$ consists of two points $r_1$ and $r_2$, which are by assumption distinct, and given by the intersection of $Q$ with two lines $L_1$ and $L_2$ forming the boundary of $P$. Let us denote by $R_1$ and $R_2$ the regions inside $C$, different from $P$, containing $r_1$ and $r_2$, respectively, see Figure~\ref{fig:lines_circ_2}. Then, $P_i = P \cup R_i$ for $i=1,2$ and $P_3 = P \cup R_1 \cup R_2$ are all convex sets containing $P$. Moreover, both $P_1$ and $P_2$ are bounded by $s-1$ lines, while $P_3$ by $s-2$ lines. Finally, the same vanishing gives $\Omega_{\widehat{P}_1}+\Omega_{\widehat{P}_2}-\Omega_{\widehat{P}_3}=\Omega_{\widehat{P}}$, and the claim follows by induction on~$s$.
\end{proof}

\begin{eg}[Minimally-curvy $(s+1)$-gon]\label{eg:curvy_pentagon}
Consider a polycon $P(1,s)$ with one curved edge and $s$ consecutive linear edges $L_1,\ldots,L_s$, cut out by $\ell_1,\ldots,\ell_s$. Proposition~\ref{prop:ext_tr_con_lines} gives the canonical-form decomposition
\begin{equation}\label{eq:tra_curvy_pent}
\Omega_{\widehat P(1,s)}
=
\sum_{i=1}^{s-1}\Omega_{i,i+1}
-
\sum_{i=2}^{s-1}\Omega_i \, ,
\end{equation}
where $\Omega_I$ denotes the canonical form of the polycon bounded by $Q$ and the lines $L_i$, $i\in I$, chosen to contain $P(1,s)$. For brevity, we write $\Omega_i:=\Omega_{\{i\}}$ and $\Omega_{i,i+1}:=\Omega_{\{i,i+1\}}$.

\begin{figure}[t]
\begin{minipage}[c]{0.31\linewidth}
\includegraphics[width=\linewidth]{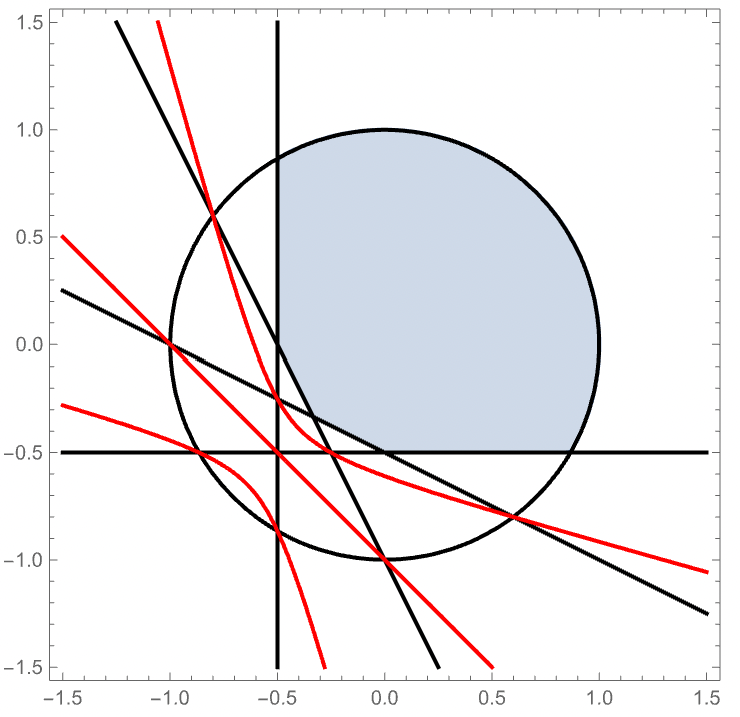}
\end{minipage}
\hfill
\begin{minipage}[c]{0.31\linewidth}
\includegraphics[width=\linewidth]{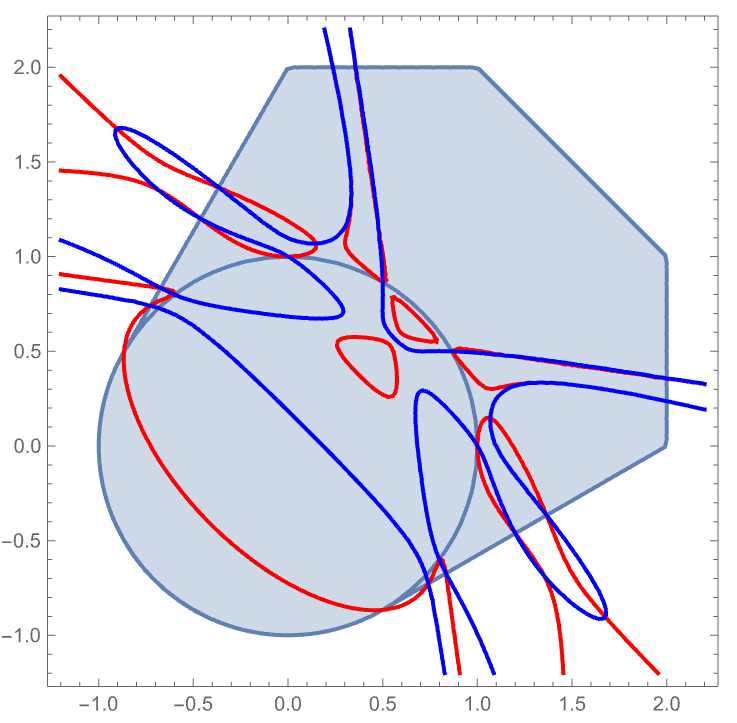}
\end{minipage}%
\hfill
\begin{minipage}[c]{0.37\linewidth}
\includegraphics[width=\linewidth]{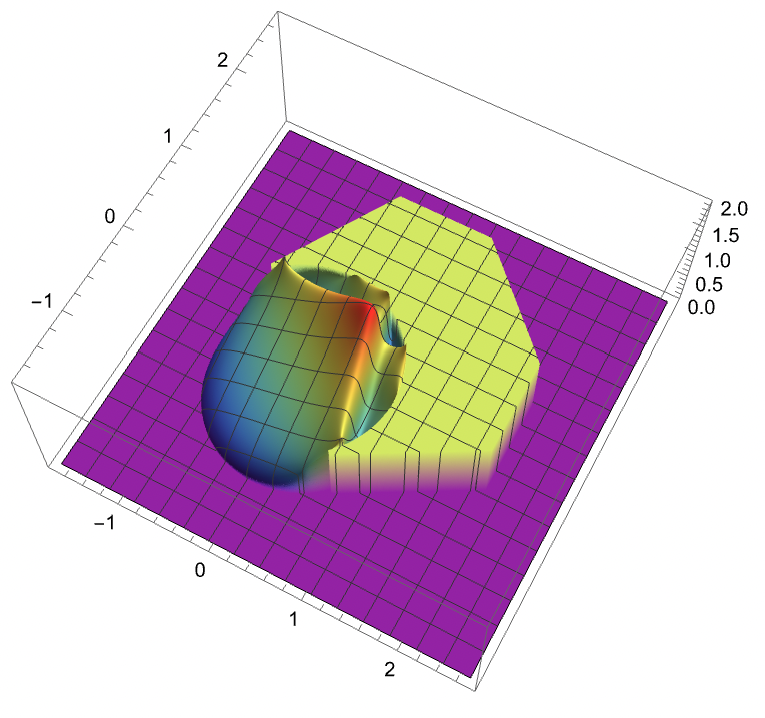}
\end{minipage}%
\caption{On the left, the polycon $P=P(1,4)$ from Example~\ref{eg:curvy_pentagon} with its adjoint cubic curve in red. In the middle, its dual $P^*$, with additional curves: in red is the vanishing locus of the degree-seven polynomial~$f_{\widehat{P}}$, and in blue that of the degree-six $g_{\widehat{P}}$, see~\eqref{eq:dual_symbol}. On the right, a plot of the graph of the measure, see~\eqref{eq:mu_penta}. Note that the function is non-negative, supported on $P^*$, equal to one on $P^* \setminus C^*$, and continuous on $\operatorname{int}(P^*)$; its extension by zero outside $P^*$ need not be continuous across $\partial P^*$.}
\label{fig:curvy_pentagon}
\end{figure}

Applying~\eqref{eq:measure_sum}, the representing density is $1$ on $\widehat P(1,s)^*\setminus C^*,$
while for $y\in\operatorname{int}(C^*)$,
\begin{equation}\label{eq:mu_penta}
\mu_{\widehat P(1,s)}(y)
=
\frac12
+\frac1\pi\sum_{i=1}^s
\arctan\!\left(\frac{\widetilde\ell_i(y)}{\sqrt{q^*(y)}}\right)
-\frac1\pi\sum_{i=1}^{s-1}
\arctan\!\left(\frac{\widetilde\ell_{i,i+1}(y)}{\sqrt{q^*(y)}}\right),
\end{equation}
where $\ell_{i,i+1}$ is the line through the two points of $(L_i\cap Q)\cup(L_{i+1}\cap Q)$
not lying on $P(1,s)$.Combining the inverse-tangent terms in~\eqref{eq:mu_penta} gives
\begin{equation}\label{eq:dual_symbol}
\mu_{\widehat P(1,s)}(y)
=
\frac12+\frac1\pi
\arctan\!\left(
\frac{f_{\widehat P(1,s)}(y)}
{g_{\widehat P(1,s)}(y)\sqrt{q^*(y)}}
\right)+k(y) \, ,
\qquad y\in\operatorname{int}(C^*) \, ,
\end{equation}
where $f_{\widehat P(1,s)}$ and $g_{\widehat P(1,s)}$ are homogeneous of degrees $2s-1$ and $2s-2$, respectively, and $k:\operatorname{int}(C^*)\longrightarrow\mathbb Z$
is piecewise constant, encoding the branch of the inverse tangent.

For $s=4$, take the example of Figure~\ref{fig:curvy_pentagon},
\begin{equation}
    \ell_1=(1,0,1/2) \, ,\quad
\ell_2=(0,1,1/2) \, ,\quad
\ell_3=(1,2,1) \, ,\quad
\ell_4=(2,1,1) \, .
\end{equation}
Here $f_{\widehat P(1,4)}$ and $g_{\widehat P(1,4)}$ have degrees seven and six. Their vanishing loci and the corresponding density are shown in Figure~\ref{fig:curvy_pentagon}; in this example $k(y)$ is nonconstant on $\operatorname{int}(C^*)$.
\end{eg}

From~\eqref{eq:mu_penta} it is not obvious that $\mu_{\widehat{P}(1,s)}$ is non-negative on $\operatorname{int}(C^*)$, but this is true.

\begin{lemma}\label{lemma:mu_nonnegative}
    The function in~\eqref{eq:mu_penta} is non-negative. In particular, any polycon $(\mathbb{P}^2,P(1,s))$ of type $(1,s)$ is a completely monotone positive geometry.
\end{lemma}

\begin{proof}
    We prove the claim by induction over $s \geq 1$. For $s=1$, the claim follows from the explicit expression in~\eqref{eq:mu_half_pizza}. Let $s>1$. Then, we can write ~\eqref{eq:mu_penta} as
    \begin{equation}\label{eq:mu_s_s-1}
         \mu_{\widehat{P}(1,s)}(y) = \mu_{\widehat{P}(1,s-1)}(y)  + \frac{1}{\pi} \arctan\left( \frac{\widetilde{\ell}_s(y)}{\sqrt{q^*(y)}}\right) - \frac{1}{\pi} \arctan\left( \frac{\widetilde{\ell}_{s-1,s}(y)}{\sqrt{q^*(y)}}\right) \, , \quad \forall \, y \in \operatorname{int}(C^*) \, .
    \end{equation}
    By induction hypothesis, $\mu_{\widehat{P}(1,s-1)}$ is non-negative on $\operatorname{int}(C^*)$, and we show that also the other term in~\eqref{eq:mu_s_s-1} is non-negative. For that, we use the addition formula for inverse tangent functions in~\eqref{eq:arctan_addition} with $a=\widetilde{\ell}_{s}(y)/\sqrt{q^*(y)}$ and $b=-\widetilde{\ell}_{s-1,s}(y)/\sqrt{q^*(y)}$. Then, once the two terms are added the argument of the inverse tangent is equal to
    \begin{equation}\label{eq:added_argument}
        \frac{a+b}{1-ab} = \frac{\widetilde{\ell}_s(y)-\widetilde{\ell}_{s-1,s}(y)}{\widetilde{\ell}_s(y)\widetilde{\ell}_{s-1,s}(y) + q^*(y)}\sqrt{q^*(y)} \, , \quad \forall \, y \in \operatorname{int}(C^*) \, .
    \end{equation}
    If $ab < 1$, then $\delta(a,b)$ in~\eqref{eq:arctan_addition} is equal to zero. We argue that in this case $\widetilde{\ell}_s(y)-\widetilde{\ell}_{s-1,s}(y) \geq 0$, so that~\eqref{eq:added_argument} is non-negative, and hence so is~\eqref{eq:mu_s_s-1}, as $\arctan(z) \geq 0$ for $z \geq 0$. In fact, $\widetilde{\ell}_s(y)-\widetilde{\ell}_{s-1,s}(y) =0 $ defines the equation of the line $L$ through $\widetilde{\ell}_s$ and $\widetilde{\ell}_{s-1,s}$. This is projectively dual to the intersection point of $L_s $ with $L_{s-1,s}$, which lies on $Q$. Therefore, $L$ is tangent to $Q^*$, and hence $\widetilde{\ell}_s(y)-\widetilde{\ell}_{s-1,s}(y)$ is either positive or negative on $\operatorname{int}(C^*)$. It is easy to verify that by construction we have that $\widetilde{\ell}_s(y)-\widetilde{\ell}_{s-1,s}(y)$ is positive on $\operatorname{int}(C^*)$, see Figure~\ref{fig:lemma_proof}.

    We now consider the case when $ab > 1$, i.e. when $-\widetilde{\ell}_s(y)\widetilde{\ell}_{s-1,s}(y)>q^*(y)$. If $a <0$, i.e. $\widetilde{\ell}_s(y)\leq 0$, then by construction one can verify that $\widetilde{\ell}_{s-1,s}(y)\leq 0$, see Figure~\ref{fig:lemma_proof}. Since $q^*(y)>0$ on $\operatorname{int}(C^*)$, this region is empty. On the other hand, if $a>0$, then $\delta(a,b) = 1$ in~\eqref{eq:arctan_addition}. Hence, in this case the resulting function is always non-negative as $\arctan$ is bounded between $-\pi/2$ and~$\pi/2$.
\end{proof}

\begin{figure}[t]
\centering
  \includegraphics[width=0.42\linewidth]{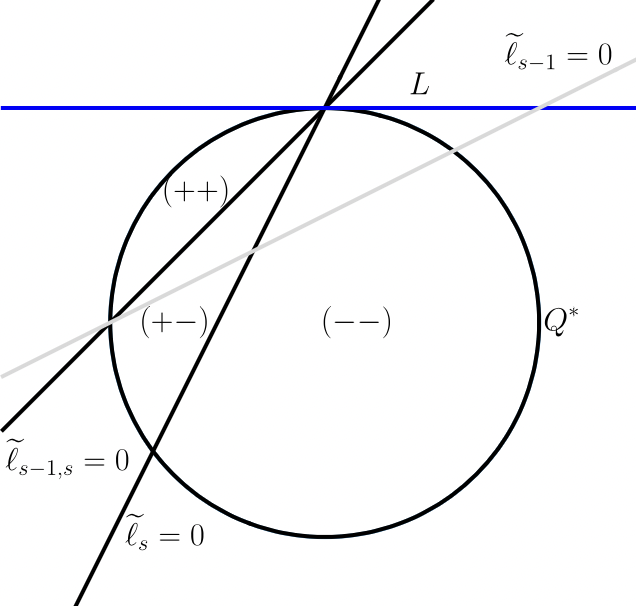}
 \caption{Figure for the argument in the proof of Lemma~\ref{lemma:mu_nonnegative}. The sign-regions $(++),(+-)$ and $(--)$ inside $C^*$ are with respect to $(\widetilde{\ell}_s(y_1,y_2,1),\widetilde{\ell}_{s-1,s}(y_1,y_2,1))$. The line $L$ is cut out by $\widetilde{\ell}_s(y)-\widetilde{\ell}_{s-1,s}(y)=0$, and we have that $\widetilde{\ell}_s(y)-\widetilde{\ell}_{s-1,s}(y) \geq 0$ for every $y \in C^*$.}
 \label{fig:lemma_proof}
\end{figure}

From this result, we deduce the following.

\begin{theorem}\label{thm:polycones_CM}
    Every hyperbolic polycon $(\mathbb{P}^2,P(r,s))$ is a completely monotone positive geometry.
\end{theorem}

\begin{proof}
    A polycon $P=P(r,s)$ is specified, up to cyclic rotation, by a sequence of positive integers $(s_1,\dots,s_r)$ with $s=\sum_{i=1}^r s_i$, where $s_i$ denotes the number of consecutive linear edges between two successive curvy edges on the conic $Q$. As usual, let us denote by $C$ the hyperbolicity cone of $Q$. Using the argument as in the proof of Proposition~\ref{prop:ext_tr_con_lines}, we have the canonical form triangulation $\Omega_{\widehat{P}} = \sum_i \Omega_{\widehat{P}_i}$, where $P_i $ is the polycon of type $(1,s_i)$ bounded by the same $s_i$ consecutive lines bounding $P$, such that $P \subset P_i$. We can then compute the measure $\mu_{\widehat{P}}$ representing $\Omega_{\widehat{P}}$. We find that $\mu_{\widehat{P}}$ is constant equal to one on $\widehat{P}^* \setminus C^*$, while 
    \begin{equation}\label{eq:mu_r_s}
        \mu_{\widehat{P}}(y) = \sum_{i=1}^r \mu_{\widehat{P}_i}(y) = \frac{1}{2} + \frac{1}{\pi} \arctan( \frac{f_{\widehat{P}}(y)}{g_{\widehat{P}}(y) \sqrt{q^*(y)} } ) + k(y) \, , \quad \forall \, y \in \operatorname{int}(C^*) \, ,
    \end{equation}
    where $f_{\widehat{P}}$, $g_{\widehat{P}}$ are homogeneous polynomials of degree $2s-r$ and $2s-r-1$, respectively\footnote{In the second equality of ~\eqref{eq:mu_r_s} we implicitly used the equation $\arctan(\frac{1}{x}) = {\rm sign}(x) \, \frac{\pi}{2} - \arctan(x)$, valid for every $x \in \mathbb{R} \setminus\{0\}$.}, and $k(y) \in \mathbb{Z}$ is constant on $C^* \setminus (V(f_{\widehat{P}}) \cup V(g_{\widehat{P}}))$.
    In~\eqref{eq:mu_r_s} the measure $\mu_{\widehat{P}_i}$ has the form as in ~\eqref{eq:mu_penta}, with the $s_i$ lines $\ell_i$ being those bounding $P_i$. By Lemma~\ref{lemma:mu_nonnegative} we have that $\mu_{\widehat{P}_i}$ is non-negative for every $1 \leq i \leq r$, and hence~\eqref{eq:mu_r_s} is also non-negative. Note that~\eqref{eq:mu_r_s} gives a formula for computing the measure of any polycon of type $(r,s)$. 
\end{proof}

\begin{definition}\label{def:dual_symbol}
    Given a polycon $P$ of type $(r,s)$, we call the homogeneous polynomials $f_{\widehat{P}}$ and $g_{\widehat{P}}$ in~\eqref{eq:mu_r_s} the \textit{dual letters}\footnote{This terminology is motivated by physics, where the \textit{letters} of an integral describe its singular loci and are used to construct the arguments of the functions expressing them~\cite{Goncharov:2010jf, Duhr:2011zq}.} of~$P$. These determine the measure $\mu_{\widehat{P}}$ representing the canonical function of $P$, up to a piecewise integer constant function $k(y)$ on $\operatorname{int}(C^*)$.
\end{definition}

Note that the dual letters can be computed from ~\eqref{eq:mu_penta} and ~\eqref{eq:mu_r_s} solely from $\ell_i$ and $q$, which in turn allow to compute $\ell_{ii+1}$, and then by repeated application of the addition formula for the inverse tangent function~\eqref{eq:arctan_addition}.
On the other hand, it would be interesting to provide a geometric understanding of dual letters, in terms of their interpolation conditions with the algebraic boundary of $\widehat{P}^*$, the lines $\widetilde{\ell}_i$ and the adjoint hypersurface of $P$, see for example Figure~\ref{fig:pizza_slice2}.

\begin{eg}[The maximally-curvy $2r$-gon]\label{eg:curvy_k_gons}
Let us consider the case of as many curvy edges as possible, namely when $r=s$.
Then, $P(r,r)$ looks like a curvy $2r$-gon, with $r$ linear edges sequentially alternating to $r$ curvy edges on the conic, see Figure~\ref{fig:line_conic_polypols}. Following the proof of Theorem~\ref{thm:polycones_CM}, we find $\Omega_{\widehat{P}(r,r)} = \sum_{i=1}^r \Omega_{\widehat{P}_i}$, where the polycons $P_i$ are of type $(1,1)$.
The dual semialgebraic set $P(r,r)^*$ looks like a disk with $r$ horns, see Figure~\ref{fig:line_conic_polypols}. On the horns, outside $C^*$, the measure $\mu_{\widehat{P}(r,r)}$ representing $\Omega_{\widehat{P}(r,r)}$ is constant equal to one. Inside the $C^*$ it is equal to
\begin{equation}\label{eq:mu_k_gon}
    \mu_{\widehat{P}(r,r)}(y) = \frac{r}{2} + \frac{1}{\pi} \sum_{i=1}^r \arctan\left( \frac{\widetilde{\ell}_i(y)}{\sqrt{q^*(y)}}\right) \, , \quad \forall \, y \in \operatorname{int}(C^*) \, .
\end{equation}

Let us rewrite and plot the measure for the case of $r=4$. We take the eight vertices of $P(4,4)$ to be evenly distributed on the unit circle, starting from $(0,1,1)$, see Figure~\ref{fig:line_conic_polypols}.
We can then write~\eqref{eq:mu_k_gon} as
\begin{equation}\label{eq:mu_k_4}
    \frac{1}{2} - \frac{1}{\pi}\arctan \left(
 \tfrac{ 17 y_1^4 - 128 y_1^3 y_2 + 34 y_1^2 y_2^2 + 128 y_1 y_2^3 + 17 y_2^4 - 
   2 (33 + 20 \sqrt{2}) (y_1^2 + y_2^2) y_3^2 + (49 + 40 \sqrt{2}) y_3^4}{8 \sqrt{2 + \sqrt{2}}
    y_3  (-7 (y_1^2 + y_2^2) + (7 + 4 \sqrt{2}) y_3^2) \sqrt{y_3^2 -y_2^2 - y_1^2}
   }\right) \, .
\end{equation}
This gives explicit expressions for the dual letters of this polycon, see also Figure~\ref{fig:line_conic_polypols}.
    
\end{eg}

\begin{figure}[t]
\begin{minipage}[c]{0.31\linewidth}
\includegraphics[width=\linewidth]{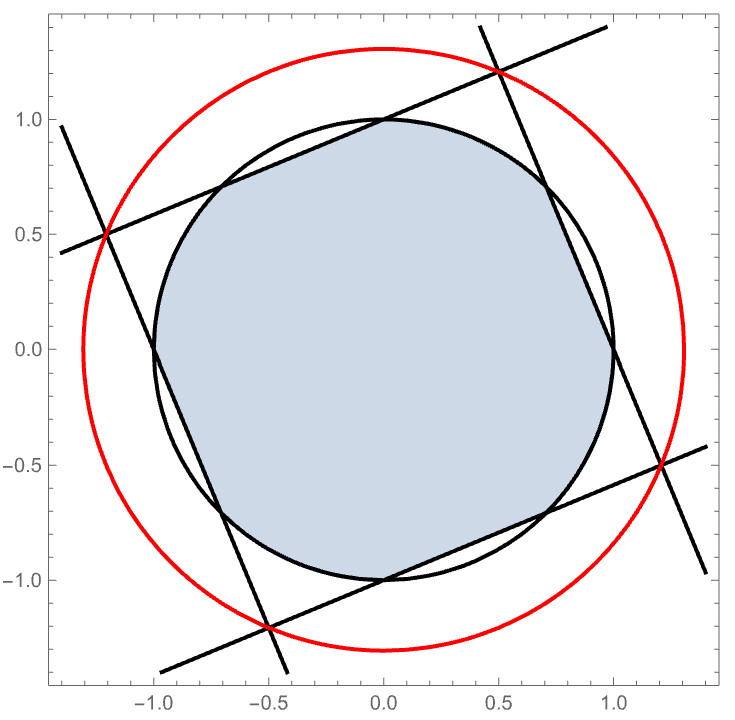}
\end{minipage}
\hfill
\begin{minipage}[c]{0.31\linewidth}
\includegraphics[width=\linewidth]{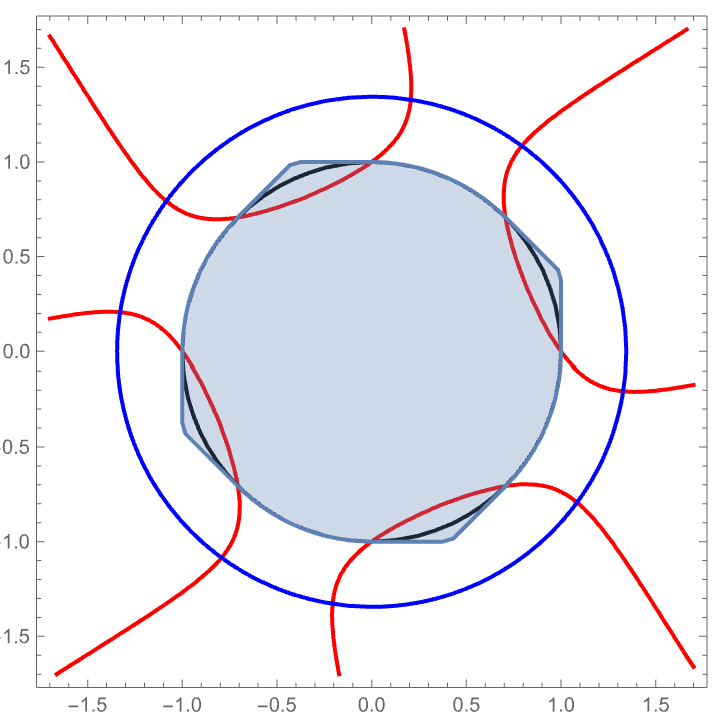}
\end{minipage}%
\hfill
\begin{minipage}[c]{0.37\linewidth}
\includegraphics[width=\linewidth]{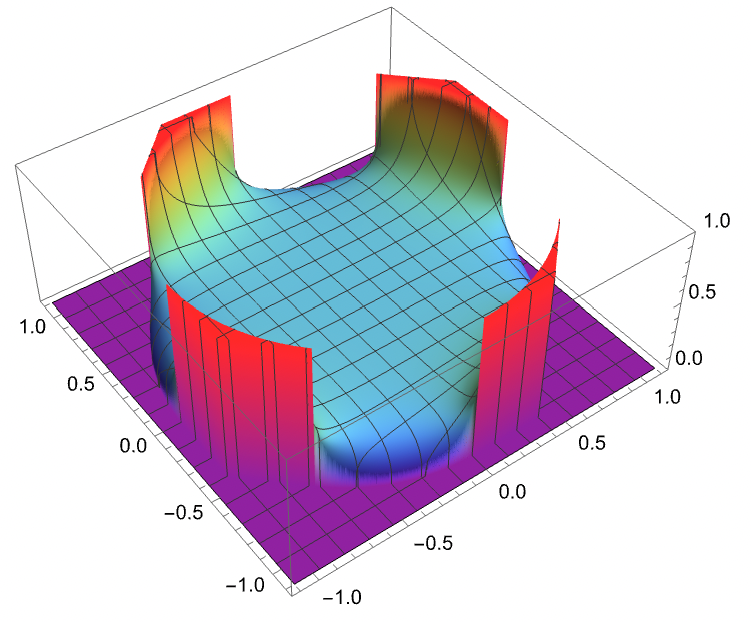}
\end{minipage}%
\caption{On the left, an octagonal polycon $P(4,4)$ as in Example~\ref{eg:curvy_k_gons} with its adjoint curve in red (which consists also of the line at infinity $x_3=0$). In the middle, its dual $P(4,4)^*$, with additional curves: in red is the vanishing locus of the quartic in the numerator of~\eqref{eq:mu_k_4}, while in blue that of the quadric in the denominator (the denominator involves also the line at infinity $y_3=0$). On the right, a plot of the graph of the measure $\mu_{\widehat{P}(4,4)}$, see~\eqref{eq:mu_k_4}.}
\label{fig:line_conic_polypols}
\end{figure}

\subsection{More conics}
\label{subsec:Any number of lines and conics}

We now extend the construction to several conics. Let $Q_1,\ldots,Q_t\subset\mathbb P^2_{\mathbb R}$
have hyperbolicity regions $C_1,\ldots,C_t$, and assume
\begin{equation}
    C:=\bigcap_{j=1}^t C_j\neq\emptyset \, .
\end{equation}
We assume that distinct conics meet transversely in two or four real points, that no three conics meet at one point, and that every boundary line meets every conic in two distinct real points. Under these hypotheses, the resulting semialgebraic sets are positive geometries by~\cite[Theorem 2.15]{Polypols}. Complete monotonicity can only occur for regions contained in $\mathbb P(C)$, by Corollary~\ref{cor:CM_PG_hyperbolic}.

\begin{definition}\label{def:polycon_r_s_t}
A \textit{polycon of type $(r,s,t)$} is a full-dimensional semialgebraic set $P(r,s,t)\subset\mathbb P(C)$
whose boundary consists of $r$ conic arcs supported on the $Q_j$ and $s$ line segments.
\end{definition}

For $t=1$ this reduces to the polycons of the previous subsection. The following extends Proposition~\ref{prop:ext_tr_con_lines}.

\begin{proposition}\label{prop:ext_tr_con_lines_conics}
    Let $P=P(r,s,t)$ be a connected convex polycon of type $(r,s,t)$. Then there exists a signed canonical-form decomposition as in Proposition~\ref{prop:ext_tr_con_lines}.
\end{proposition}

\begin{proof}
    We argue by induction on $t$, the case $t=1$ being Proposition~\ref{prop:ext_tr_con_lines}. Choose a conic $Q_1$ and subdivide every component of $\partial P\cap Q_1$ at all its intersections with the other conics. For each resulting arc $S_k$, its interior lies in $\bigcap_{j>1}C_j$. We may choose finitely many tangent lines to $Q_1$ along $S_k$ such that the cap $R_k$ between $S_k$ and these lines, adjacent to $P$, is still contained in $\bigcap_{j>1}C_j$. Indeed, this is immediate on compact subarcs away from the endpoints. Near an endpoint $p\in Q_1\cap Q_j$, if $q_j=0$ cuts out $Q_j$, transversality gives, for a local parametrization $\gamma(u)$ of $S_k$, $q_j(\gamma(u))=cu+O(u^2)$ with $c>0$, whereas the tangent approximation to $Q_1$ differs from $\gamma$ by $O(u^2)$. Hence sufficiently short tangent pieces remain in $C_j$. Compactness gives a finite choice valid for all $j$.

    Replacing all arcs on $Q_1$ by these tangent pieces produces a convex polycon $P_1\supset P$ with one fewer conic in its boundary. Repeating the construction for a second conic gives $P_2$, and $P_3:=P_1\cup P_2$ has two fewer conics. As in Proposition~\ref{prop:ext_tr_con_lines},
    \begin{equation}
        \Omega_{\widehat{P}}=\Omega_{\widehat{P}_1}+\Omega_{\widehat{P}_2}-\Omega_{\widehat{P}_3}\, .
    \end{equation}
    The claim follows then by induction on $t$.
\end{proof}

Proposition~\ref{prop:ext_tr_con_lines_conics} reduces the representing measure of any such connected convex polycon to the one-conic case. In particular, we have that
\begin{equation}
    \mu_{\widehat P}=1
\qquad\text{on}\qquad
\widehat P^*\setminus\bigcup_{j=1}^t C_j^* \, ,
\end{equation}
while on $\bigcup_j C_j^*$ it is piecewise given by finite signed sums of the inverse-tangent densities.

We illustrate this for two conics $Q_1,Q_2$ with hyperbolicity cones $C_1,C_2$.

\begin{figure}[t]
\centering
\begin{minipage}[c]{0.32\linewidth}
\includegraphics[width=\linewidth]{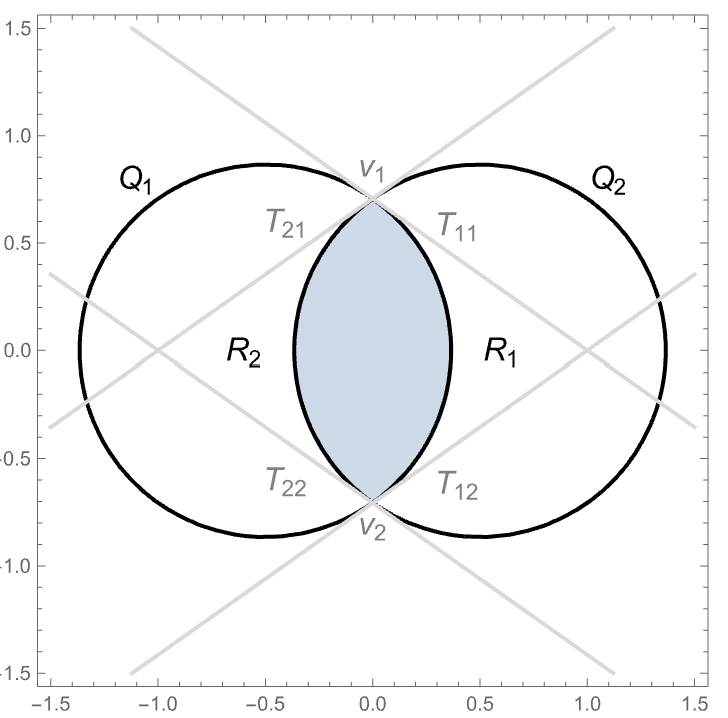}
\end{minipage}
\hfill
\begin{minipage}[c]{0.32\linewidth}
\includegraphics[width=\linewidth]{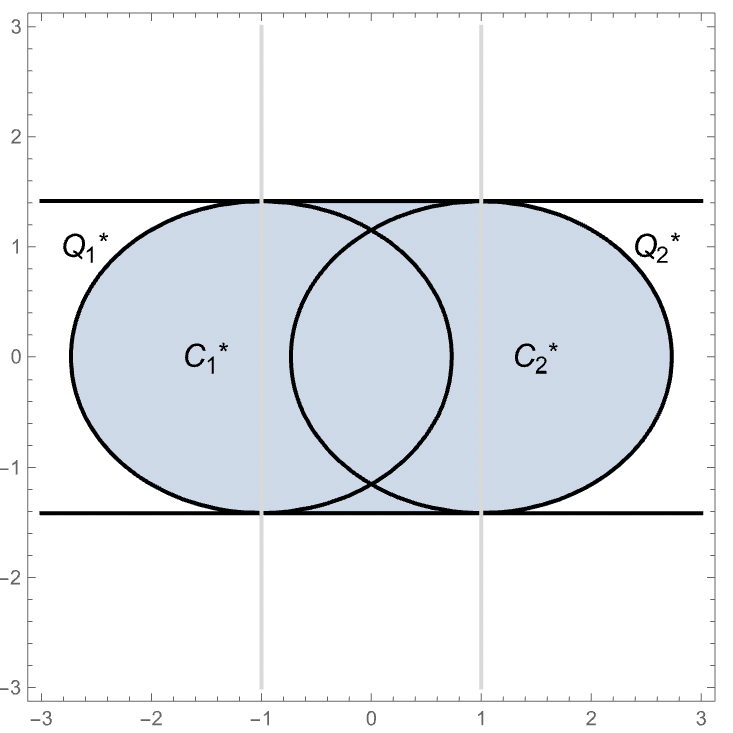}
\end{minipage}%
\hfill
\begin{minipage}[c]{0.35\linewidth}
\includegraphics[width=\linewidth]{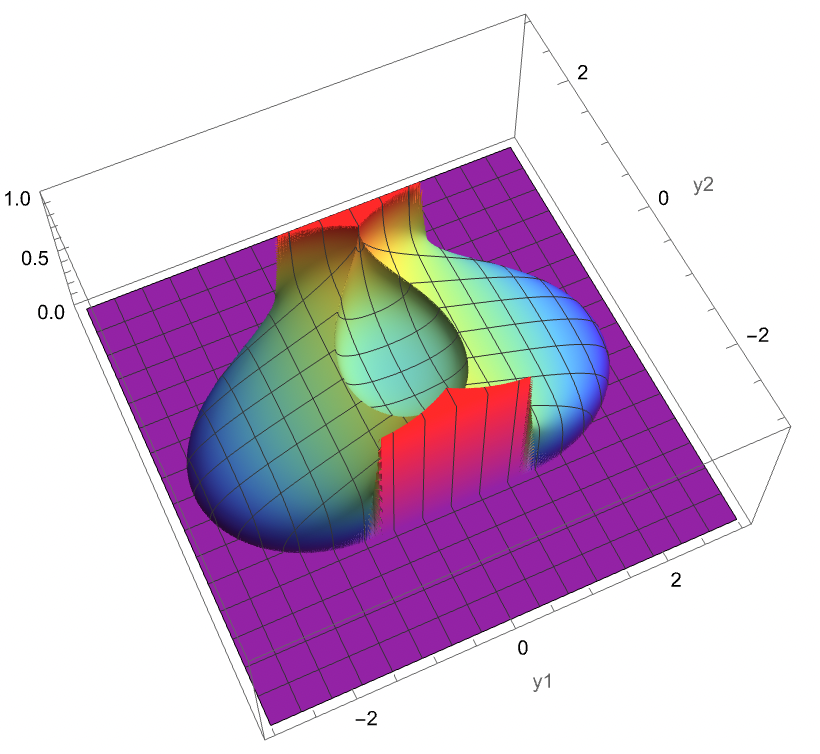}
\end{minipage}
\caption{On the left, a curvy two-gon $P$, a polycon of type $(2,0,2)$, from Example~\ref{eg:curvy_2gon}. This is a positive geometry and has an adjoint curve given by the line at infinity ${x_3=0}$. In the middle, its dual $P^*$ and on the right the plot of the measure $\mu_{\widehat{P}}$ for the canonical function of $P$, see~\eqref{eq:meas_2con}. The plot is non-negative and provides numerical evidence for Conjecture~\ref{conj:hyperbolic_polycons_CM}.}
\label{fig:2con}
\end{figure}

\begin{eg}[Curvy two-gon]\label{eg:curvy_2gon}
Assume that $Q_1$ and $Q_2$ meet in two distinct real points $v_1,v_2$, and consider the semialgebraic cone $\widehat P=C_1\cap C_2$.
For $i,j\in\{1,2\}$, let $T_{ij}$ be the tangent to $Q_i$ at $v_j$, and let $R_i$ be the cap bounded by $Q_i,T_{i1},T_{i2}$ adjacent to $P$. Setting
\begin{equation}
    P_i=P\cup R_i \, ,\qquad i=1,2,
\qquad
P_3=P\cup R_1\cup R_2 \, ,
\end{equation}
see Figure~\ref{fig:2con_tr}, gives the canonical-form decomposition
\begin{equation}
    \Omega_{\widehat P}
=
\Omega_{\widehat P_1}
+
\Omega_{\widehat P_2}
-
\Omega_{\widehat P_3} \, .
\end{equation}
Here $P_1,P_2$ are polycons of type $(1,2)$, while $P_3$ is the quadrilateral bounded by the four tangent lines. Hence $\mu_{\widehat P_3}$ is the indicator density of $P_3^*$. 
It is then easy to verify that
    \begin{equation}\label{eq:meas_2con}
        \mu_{\widehat{P}}(y)=
\begin{cases}
1 \, , & y \in \widehat{P}^* \setminus \bigl(C_1^* \cup C_2^*\bigr) \, ,\\[2pt]
\mu_{\widehat{P}_1}(y) \, , & y \in C_1^* \setminus C_2^* \, ,\\[2pt]
\mu_{\widehat{P}_2}(y) \, , & y \in C_2^* \setminus C_1^* \, ,\\[2pt]
\mu_{\widehat{P}_1}(y)+\mu_{\widehat{P}_2}(y)-1 \, , & y \in C_1^* \cap C_2^* \, ,
\end{cases}
    \end{equation}
where $\mu_{\widehat P_i}$ is given by~\eqref{eq:mu_pizza_slice2}. Positivity on $C_1^*\cap C_2^*$
is not manifest from this expression; in this example we verify it numerically, see Figure~\ref{fig:2con}.

\begin{figure}[t]
\centering
\begin{minipage}[c]{0.25\linewidth}
\includegraphics[width=\linewidth]{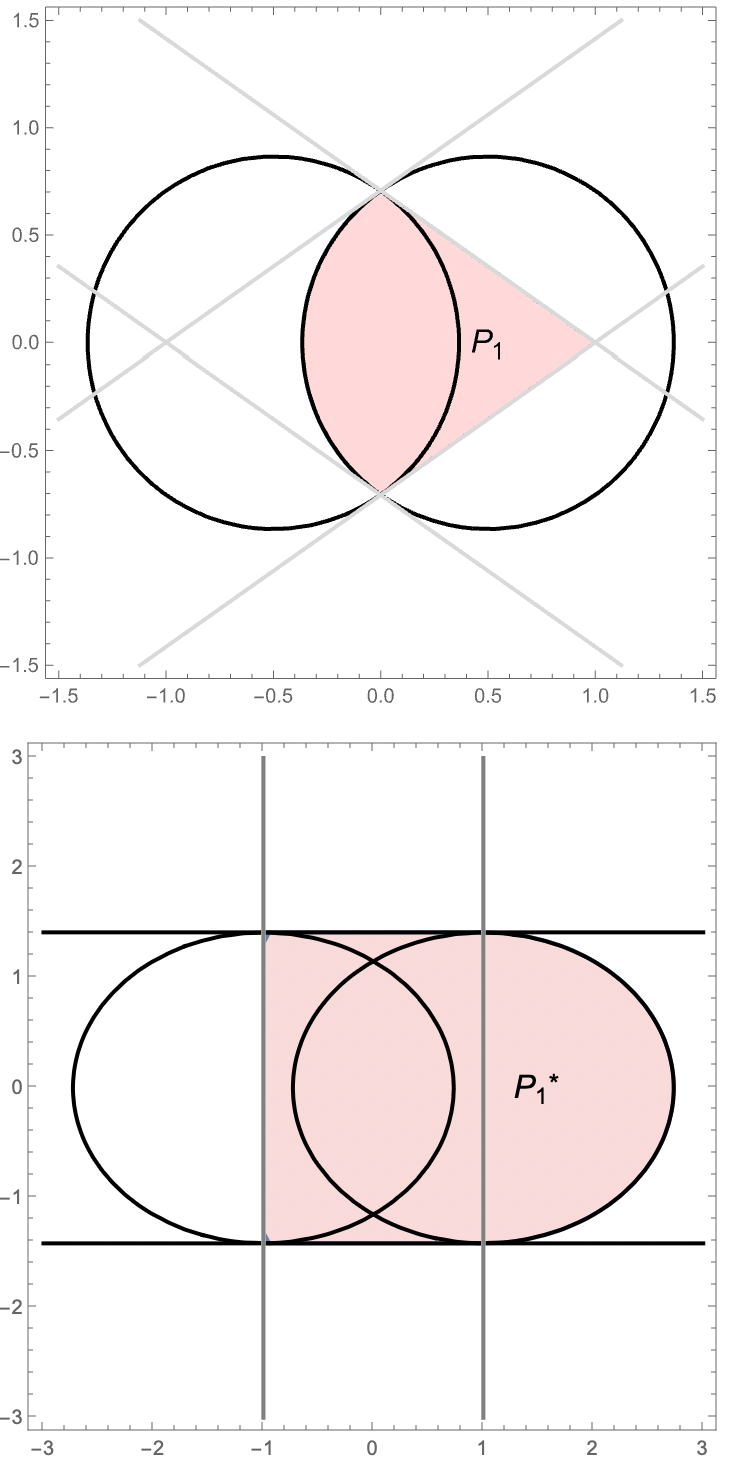}
\end{minipage}
\hfill
\begin{minipage}[c]{0.24\linewidth}
\includegraphics[width=\linewidth]{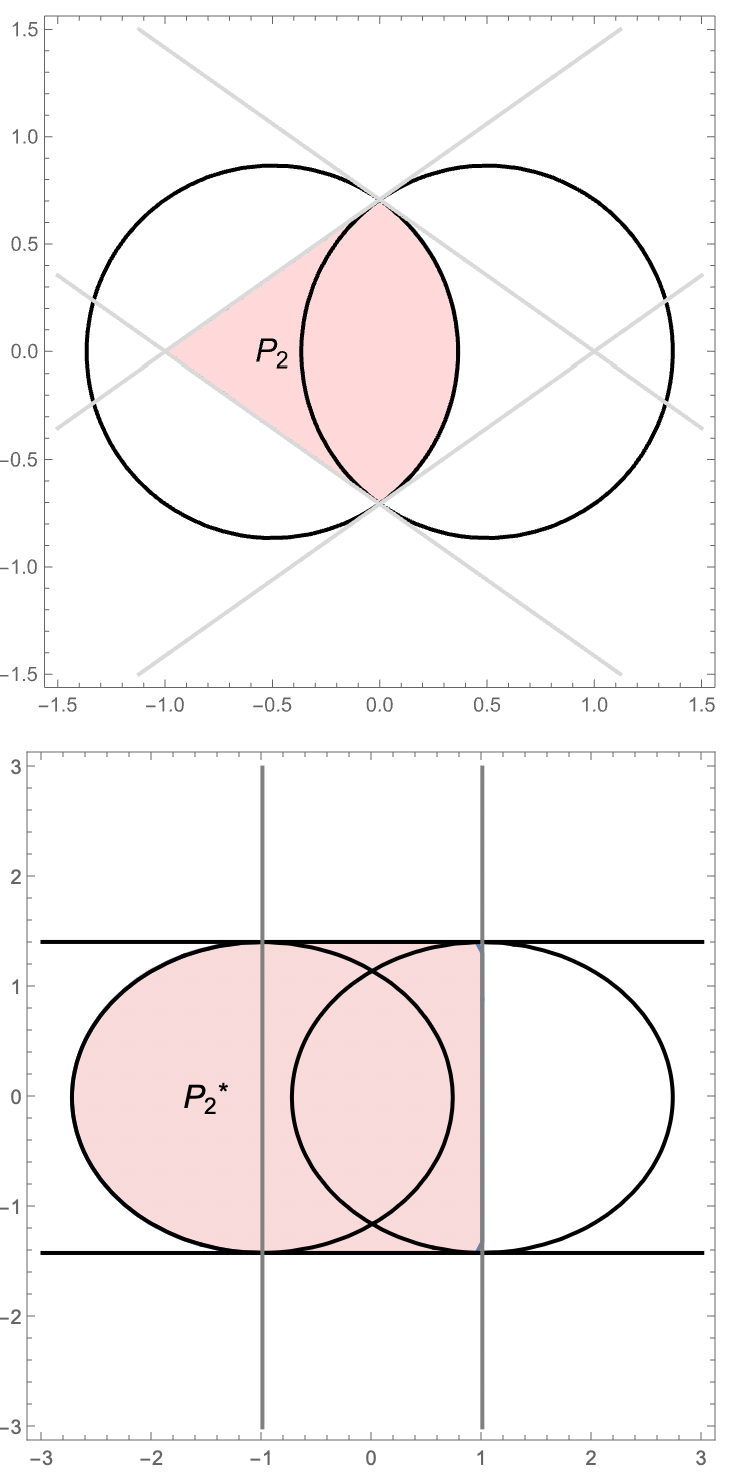}
\end{minipage}%
\hfill
\begin{minipage}[c]{0.24\linewidth}
\includegraphics[width=\linewidth]{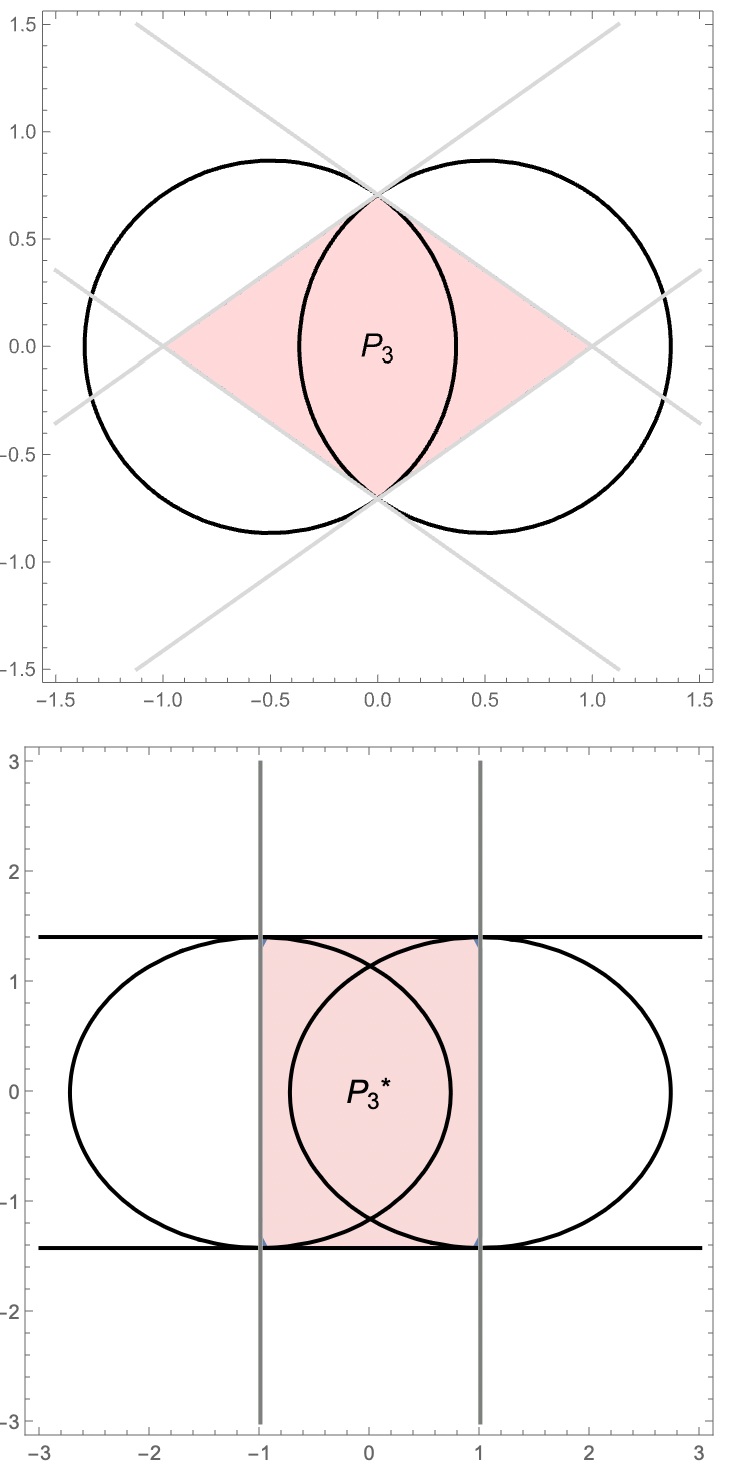}
\end{minipage}
\caption{The polycons $P_i$, above, together with their duals $P_i^*$, below, see Example~\ref{eg:curvy_2gon}. These form a canonical form triangulation for the curvy two-gon in Figure~\ref{fig:2con}, and yield the formula~\eqref{eq:meas_2con} for the representing measure.}
\label{fig:2con_tr}
\end{figure}

\end{eg}

\begin{figure}[t]
\centering
\begin{minipage}[c]{0.32\linewidth}
\includegraphics[width=\linewidth]{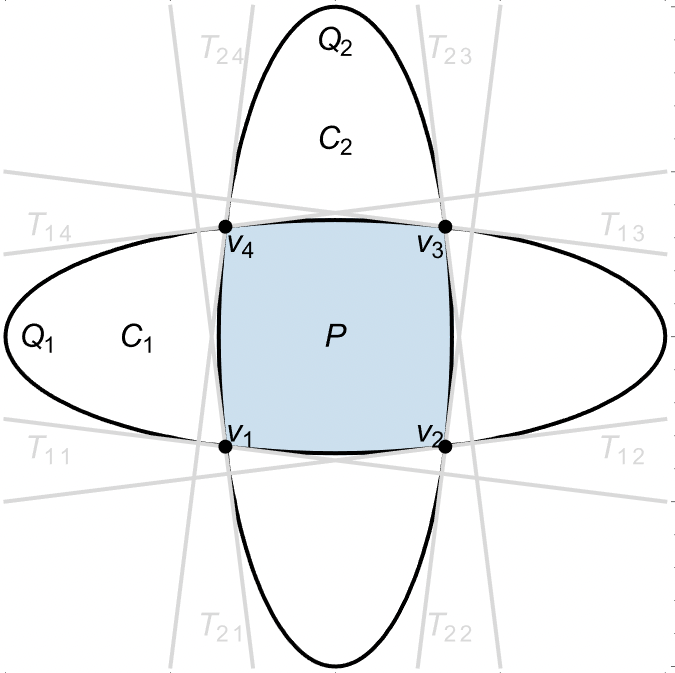}
\end{minipage}
\hfill
\begin{minipage}[c]{0.32\linewidth}
\includegraphics[width=\linewidth]{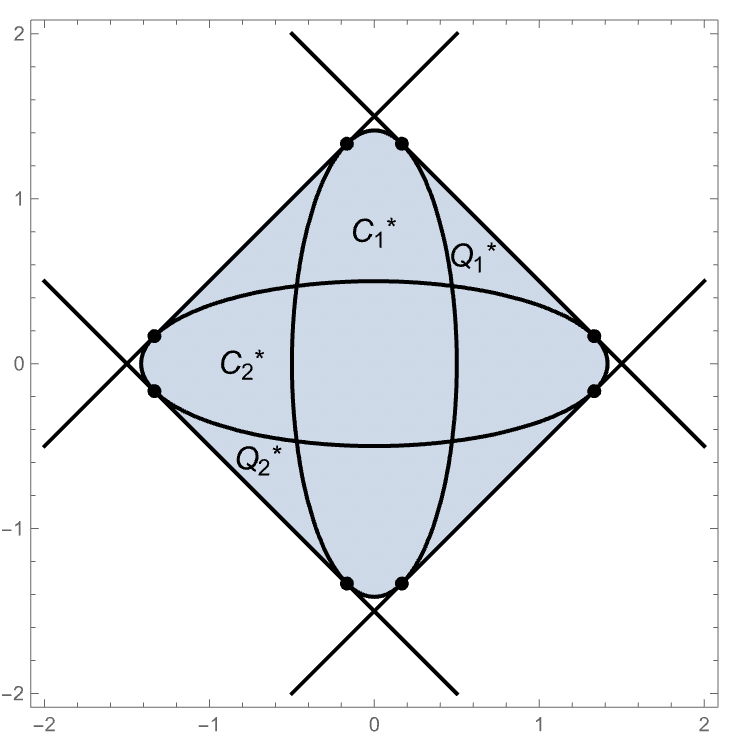}
\end{minipage}%
\hfill
\begin{minipage}[c]{0.35\linewidth}
\includegraphics[width=\linewidth]{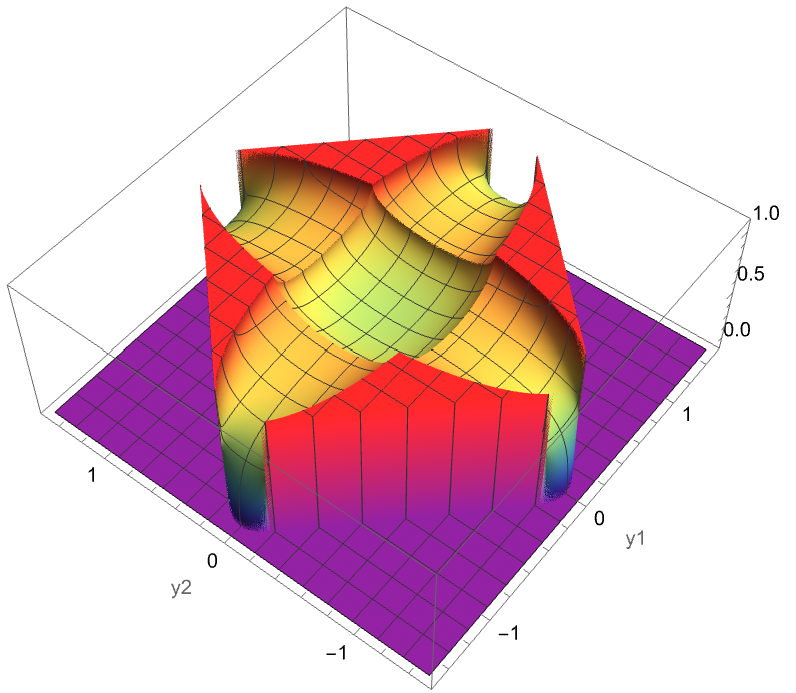}
\end{minipage}
\caption{On the left, a curvy four-gon $P$ as in Example~\ref{eg:curvy_4gon}. This is a positive geometry with an adjoint curve given by the line at infinity $x_3=0$. In the middle, its dual and on the right a plot of the measure for its canonical function, see~\eqref{eq:meqs_4con}. The plot is non-negative and provides numerical evidence for Conjecture~\ref{conj:hyperbolic_polycons_CM}.}
\label{fig:curvy_4gon}
\end{figure}

\begin{eg}[Curvy four-gon]\label{eg:curvy_4gon}
Suppose $Q_1$ and $Q_2$ meet in four distinct real points $v_1,\ldots,v_4$, ordered cyclically along $\widehat P=C_1\cap C_2.$
Let $T_{ji}$ denote the tangent to $Q_j$ at $v_i$. For $i=1,3$, let $P_i$ be the polycon bounded by $Q_2,T_{1,i},T_{1,i+1}$; for $i=2,4$, let $P_i$ be bounded by $Q_1,T_{2,i},T_{2,i+1}$, with indices understood modulo four. Choose each $P_i$ to contain $P$, and let $P_5$ be the octagon bounded by the eight tangent lines. Then
\begin{equation}\label{eq:form_tr_4gon}
\Omega_{\widehat P}
=
\sum_{i=1}^{4}\Omega_{\widehat P_i}
-
\Omega_{\widehat P_5} \, .
\end{equation}
The dual sets $P^*_i$ all lie inside $P^*$: for example, $P^*_5$ is the octagon given by the convex hull of the black points in the middle of Figure~\ref{fig:curvy_4gon}.
   From~\eqref{eq:form_tr_4gon} we compute
    \begin{equation}\label{eq:meqs_4con}
        \mu_{\widehat{P}}(y)=
   \begin{cases}
1 \, , & y \in \widehat{P}^* \setminus \bigl(C_1^* \cup C_2^*\bigr) \, ,\\[2pt]
\mu_{\widehat{P}_1}(y) + \mu_{\widehat{P}_3}(y) \, ,  & y \in C_1^* \setminus C_2^* \, ,\\[2pt]
\mu_{\widehat{P}_2}(y) + \mu_{\widehat{P}_4}(y) \, , & y \in C_2^* \setminus C_1^* \, ,\\[2pt]
\mu_{\widehat{P}_1}(y)+\mu_{\widehat{P}_2}(y) + \mu_{\widehat{P}_3}(y) + \mu_{\widehat{P}_4}(y)-1 \, , & y \in C_1^* \cap C_2^* \, ,
\end{cases}
\end{equation}
where each $\mu_{\widehat P_i}$ is given by~\eqref{eq:mu_pizza_slice2}. As in the two-gon example, positivity on $C_1^*\cap C_2^*$
is not manifest and is verified numerically in Figure~\ref{fig:curvy_4gon}. We also find numerical agreement with the formulas of Wagner~\cite[Section 6]{wagner4} for the associated fundamental solution, after applying the numerator differential operator as in Theorem~\ref{thm:Riesz_p/q}.

\end{eg}

The signed decomposition gives the representing measure explicitly, while positivity in the multi-conic case remains conjectural. The examples above motivate:
\begin{conj}\label{conj:hyperbolic_polycons_CM}
Every hyperbolic polycon $P(r,s,t)$ is a completely monotone positive geometry.
\end{conj}

\subsection{The nodal cubic}
\label{subsec:Cubics}

We now consider a positive geometry bounded by a nodal cubic. Let
\begin{equation}\label{eq:cubic}
p(x)
=
-2x_2^2x_3+(x_1+x_3)^2(x_3-x_1)
=
\det
\begin{pmatrix}
2x_3 & x_1+x_3 & 0\\
x_1+x_3 & x_1+x_3 & x_2\\
0 & x_2 & x_1+x_3
\end{pmatrix}.
\end{equation}
The cubic $V(p)$ has a node at $[-1:0:1]$ and is hyperbolic with respect to $(0,0,1)$, where the matrix above is positive definite. We denote the corresponding hyperbolicity cone by $\widehat P$; see Figure~\ref{fig:cubic}. The cone $\widehat P$ is minimal spectrahedral, and $P=\mathbb P(\widehat P)$ is a positive geometry with canonical function~\cite[Section 5.4]{Brown:PG_Hodge}
\begin{equation}\label{eq:form_cubic}
\Omega_{\widehat P}(x)=\frac4{p(x)} \, .
\end{equation}
Hence $(\mathbb P^2,P)$ is completely monotone by Corollary~\ref{cor:simplex_det_is_CM}. We now compute its measure.

By Theorem~\ref{thm:hyp_PDE}, the Riesz distribution of $p^{-1}$ is the fundamental solution of the hyperbolic operator $p(-i\partial)$; here it is a positive measure as $p^{-1}$ is completely monotone. For smooth hyperbolic cubics in three variables, Wagner~\cite{wagner1,wagner2} computed this fundamental solution explicitly. It is supported on the dual cone, is constant on an open \emph{lacuna}, and outside the lacuna is expressed through a complete elliptic integral of the first kind.

Our cubic~\eqref{eq:cubic} is singular, so Wagner's formula does not apply directly. A smooth plane cubic would moreover not define a positive geometry of the present type: as a genus-one curve it carries a nonzero holomorphic one-form, obstructing uniqueness of the canonical form. We compute the representing measure from the spectrahedral description using the formulas below Proposition~\ref{Prop:det_is_CM}. The dual cone $\widehat P^*$, shown in the middle panel of Figure~\ref{fig:cubic}, has algebraic boundary $V(q)\cup V(y_1-y_3),$
where
\begin{equation}\label{eq:dual_cubic_pol}
q(y)= -27 y_2^2 y_3^2 + 16 y_2^4 - 18 y_2^2 y_1 y_3
+13 y_2^2 y_1^2 + 8 y_1^3 y_3 + 8 y_1^4 \, .
\end{equation}
The quartic $V(q)$ is the projective dual of $V(p)$, while $y_1-y_3=0$ is dual to the node. Following~\cite{wagner2}, we define the lacuna as the open semialgebraic cone
\begin{equation}\label{eq:lacuna}
\mathcal L
=\operatorname{int}(\widehat P^*)\cap\{q>0\} \, .
\end{equation}

\begin{figure}[t]
\centering
\begin{minipage}[c]{0.31\linewidth}
\includegraphics[width=\linewidth]{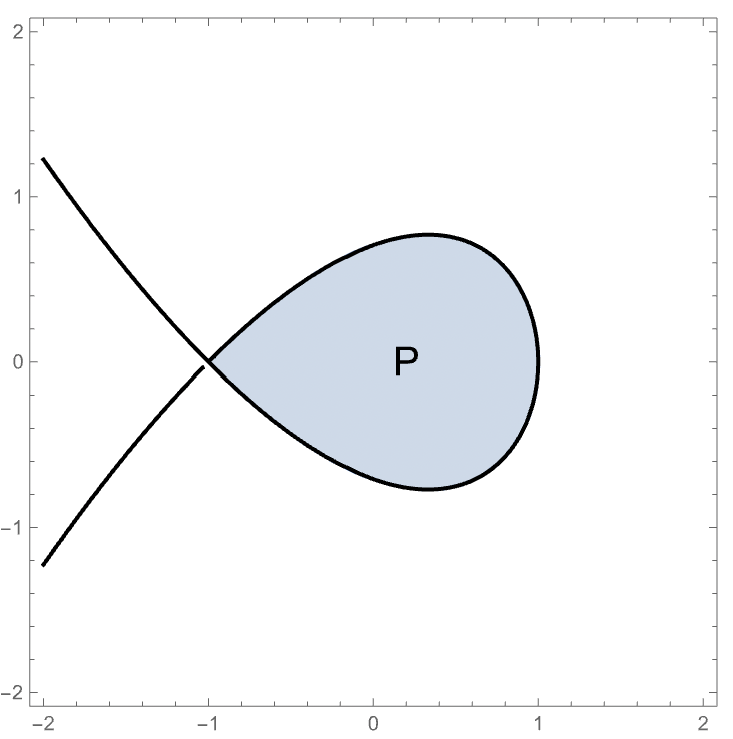}
\end{minipage}
\hfill
\begin{minipage}[c]{0.31\linewidth}
\includegraphics[width=\linewidth]{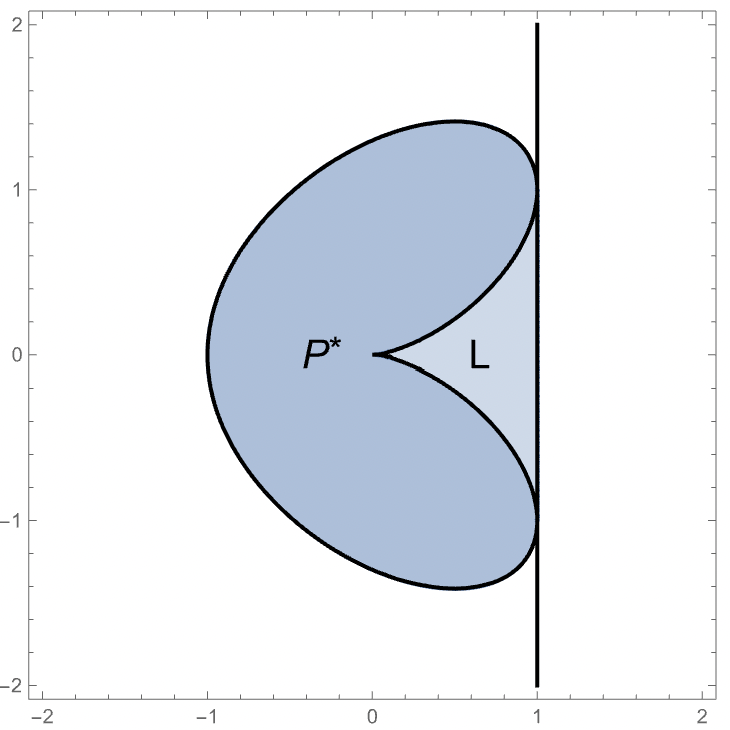}
\end{minipage}%
\hfill
\begin{minipage}[c]{0.35\linewidth}
\includegraphics[width=\linewidth]{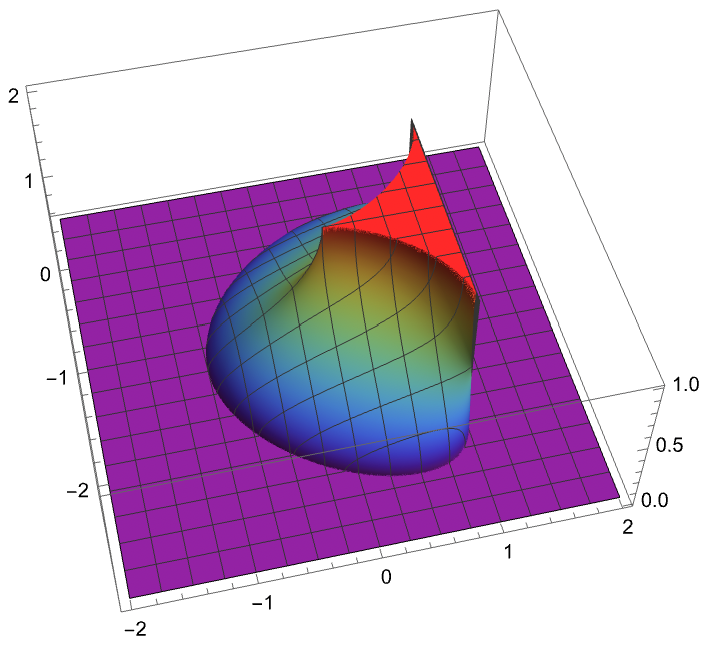}
\end{minipage}
\caption{On the left the nodal cubic cut out by ~\eqref{eq:cubic} with its hyperbolicity region~$P$, shaded in blue, on the affine slice $x_3=1$. In the middle, the dual $P^*$ given by the whole shaded region. The lightly-shaded region is the lacuna $\mathcal L$, and in black one sees the two components of the algebraic boundary: the one of degree four given by the vanishing of~\eqref{eq:dual_cubic_pol}, and the line $y_1=1$ (equivalently $y_1-y_3=0$ in homogeneous coordinates) dual to the node of the cubic. On the right we plot the graph of the measure in~\eqref{eq:cubic_mu_period_2}.}
\label{fig:cubic}
\end{figure}

\begin{proposition}\label{prop:nodal_cubic_riesz}
The representing measure of $\Omega_{\widehat P}$ in~\eqref{eq:form_cubic} is given by
\begin{equation}\label{eq:cubic_pushforward}
4\,(L\circ\phi_2)_*\!\left(\pi^{-3}dZ\right) \, ,
\qquad
\phi_2(Z)=ZZ^T \, ,\quad Z\in\mathbb R^{3\times2} \, ,
\end{equation}
with $L$ as in~\eqref{eq:L_projection}. On $\operatorname{int}(\widehat P^*)$ it has density
\begin{equation}\label{eq:cubic_mu_period}
\mu_{\widehat P}(y)
=
\int_{\mathcal R(y)}
\frac{4r\,dr\,dt}
{\pi^2\sqrt{(1-t^2)h(t,r,y)}} \, ,
\end{equation}
where
\begin{equation}
    h(t,r,y)
=
-y_2^2-4r^2\bigl(r^2-y_1-r t\sqrt{2(y_3-y_1)}\bigr) \, ,\
\mathcal R(y)=\{r>0,\,-1<t<1,\,h>0\} \, .
\end{equation}
For $y\in\operatorname{int}(\widehat P^*)\setminus\overline{\mathcal L},$
the polynomial $H_y(r):=h(1,r,y)$ has exactly two positive roots $r_-(y)<r_+(y)$, and we have the inequalities
\begin{equation}\label{eq:cubic_t0}
-1<t_0(r,y):=
\frac{4r^4-4r^2y_1+y_2^2}
{4r^3\sqrt{2(y_3-y_1)}}<1 \, ,
\qquad r_-<r<r_+ \, .
\end{equation}
Hence, we can write~\eqref{eq:cubic_mu_period} as the one-fold integral
\begin{equation}\label{eq:cubic_mu_period_2}
\mu_{\widehat P}(y)
=
\int_{r_-(y)}^{r_+(y)}
\frac{2^{5/4}\,dr}
{\pi^2\sqrt r\,(y_3-y_1)^{1/4}}
K\!\left(
\frac{h(1,r,y)}
{8r^3\sqrt{2(y_3-y_1)}}
\right),
\end{equation}
where the complete elliptic integral of the first kind is
\begin{equation}
    K(z)=\int_0^{\pi/2}
\frac{d\varphi}{\sqrt{1-z\sin^2\varphi}} \, .
\end{equation}
\end{proposition}

\begin{proof}
Since in the notation of Proposition~\ref{Prop:det_is_CM} $m=3$ and $\alpha=1$,~\eqref{eq:cubic_pushforward} corresponds to~\eqref{eq:singular_riesz_pushforward}. Writing the three rows of $Z$ as $u,v,w\in\mathbb R^2$, we have that
\begin{equation}
    L(ZZ^T)
=
\bigl(
2u\!\cdot\!v+|v|^2+|w|^2,\,
2v\!\cdot\!w,\,
2|u|^2+2u\!\cdot\!v+|v|^2+|w|^2
\bigr) \, ,
\end{equation}
and the density id the integral on the fibre given by
\begin{equation}\label{eq:mu_delta}
    \mu_{\widehat P}(y)
=
4\pi^{-3}\int dZ\,
\delta^{(3)}\!\bigl(y-L(ZZ^T)\bigr) \, .
\end{equation}
Subtracting the first component from the third fixes $|u|^2=(y_3-y_1)/2.$
Using the $O(2)$ action on the columns of $Z$, we can parametrize
\begin{equation}
    u=\sqrt{\frac{y_3-y_1}{2}}\,(1,0) \, ,
\qquad
v=r(\cos\theta,\sin\theta) \, ,
\end{equation}
and decompose $w$ into components parallel and orthogonal to $v$. The remaining delta functions in~\eqref{eq:mu_delta} and their Jacobians, with $t=-\cos\theta$, give~\eqref{eq:cubic_mu_period}.

For $y\in\operatorname{int}(\widehat P^*)\setminus\overline{\mathcal L}$, one has $q(y)<0$. Indeed, $q=0$ is the relative boundary of the lacuna, whose only real projective singular point is $[0:0:1]$, and $q(\varepsilon,0,1)=8\varepsilon^3(1+\varepsilon)$
changes sign there. Moreover, the discriminant of $H_y$ equals $1024\,y_2^2q(y)$.
If $y_2\neq0$, choose $B_0\succ0$ with $L(B_0)=y$. Since
\[
L:\operatorname{Sym}_3(\mathbb R)\longrightarrow\mathbb R^3
\]
maps a six-dimensional space to a three-dimensional one, $\ker L\neq0$. Moving from $B_0$ along a nonzero kernel direction until reaching the boundary of the positive-semidefinite cone produces a matrix $ZZ^T$ of rank at most two with the same image $y$. Thus $H_y(r)\geq0$ for some $r>0$. Since its discriminant is negative, $H_y$ has exactly two real roots; because $H_y(0)<0$ and its leading coefficient is negative, both are positive.

If $y_2=0$, then
\begin{equation}
    H_y(r)
=
4r^2\bigl(-r^2+r\sqrt{2(y_3-y_1)}+y_1\bigr) \, ,
\end{equation}
and $q(y)<0$ implies $y_1<0<y_1+y_3$, again giving two positive nonzero roots. In both cases $H_y(-r)<0$ for $r>0$. Finally,
\begin{equation}
    h(t,r,y)
=
4r^3\sqrt{2(y_3-y_1)}\,\left(t-t_0(r,y) \right) \, ,
\end{equation}
which proves~\eqref{eq:cubic_t0}; integrating~\eqref{eq:cubic_mu_period} over $t$ gives~\eqref{eq:cubic_mu_period_2}.
\end{proof}

By~\eqref{eq:cubic_pushforward}, the measure is supported in $\widehat P^*$. On the lacuna it is constant:
\begin{equation}
    \mu_{\widehat P}(y)=1 \, ,
\qquad y\in\mathcal L \, ,
\end{equation}
as follows by evaluating~\eqref{eq:cubic_mu_period} at $y=(1/2,0,1)$. The quartic $q$ controlling the lacuna also controls the discriminant of $H_y$. We plot the density obtained from~\eqref{eq:cubic_mu_period_2} in Figure~\ref{fig:cubic}.

The density is smooth away from its singular support, which is contained in $V(q)\cup\partial\widehat P^*.$
It is continuous away from the cone over $\partial\mathcal L\cap\partial\widehat P^*.$
Unlike the smooth-cubic case of~\cite{wagner2}, where the fundamental solution is continuous away from the origin, the singular cubic allows the lacuna boundary to meet the boundary of the propagation cone. 
 
\break

\section{Open problems}
\label{sec:Open problems}

We conclude with several open questions.

\begin{question}[Refining the class of completely monotone positive geometries]
Corollary~\ref{cor:CM_PG_hyperbolic} shows that every completely monotone positive geometry is a hyperbolicity region of its algebraic boundary. Is the converse true? Namely, if $P$ is hyperbolic, is $\Omega_{\widehat P}=q/p$
completely monotone on $\widehat P$?

This is stronger than Conjecture~3.5 of~\cite{Exponential_varieties}, which predicts complete monotonicity of $p^{-\alpha}$ for sufficiently large $\alpha$. Here we ask about the distinguished exponent $\alpha=1$ and, more generally, about the effect of the adjoint numerator $q$. The claim holds for simplex-like hyperbolic determinantal positive geometries by Corollary~\ref{cor:simplex_det_is_CM}; does it extend to non-simplex-like hyperbolic positive geometries?
\end{question}

\begin{question}[Properties of the Riesz distribution]
For a hyperbolic positive geometry, the associated Riesz distribution is smooth away from its singular support; see Remark~\ref{rem:regularity_of_E}. Can one characterize its microlocal singular structure along the wave front set and, when it is represented by a measure or locally integrable function, its boundary behavior on the dual cone?
\end{question}

\begin{question}[Inverse moment problem]
The representing measure in~\eqref{eq:Choquet} need not have finite total mass and therefore is not generally a probability measure. When finite, it can be normalized; more generally, exponential tilting by a point of the cone gives a probability distribution if the Laplace transform converges. Under what conditions is the resulting moment problem determinate for rational canonical functions of positive geometries?
\end{question}


\begin{question}[Computing the measure]
For minimal spectrahedral positive geometries, Subsection~\ref{subsec:determinantal representations and spectrahedral shadows} gives a general construction of the representing measure, but explicit evaluation typically leads to nontrivial period integrals. Can one develop a systematic method for computing these periods, for example using $D$-module techniques~\cite[Section 7]{Pos_Certif}?
\end{question}

\begin{question}[Finding positive geometries with computable measures]
Can one identify classes of positive geometries whose representing measures admit explicit expressions in known transcendental functions? A natural family is given by positive geometries in projective space bounded by hyperplanes and quadrics. Which function classes are required in this case?
\end{question}

\begin{question}[Complete monotonicity on subvarieties]
How should complete monotonicity be defined for functions on semialgebraic subsets of real projective varieties? Our main motivation is the amplituhedron, a semialgebraic subset of a Grassmannian~\cite[Section 6.6]{Positive_geometries}, for which an appropriate notion of complete monotonicity could provide a framework for dual-volume representations of canonical forms.
\end{question}

\section{Acknowledgements}
We would like to thank Lorenzo Baldi, Clemens Brüser, Hadleigh Frost, Johannes Henn, Joris Koefler, Khazhgali Kozhasov, Mateusz Michałek, Lizzie Pratt, Giulio Salvatori, Rainer Sinn, Bernd Sturmfels for helpful discussions. This project is funded by the European Union (ERC, UNIVERSE PLUS, 101118787). Views and opinions expressed are those of the authors only and do not necessarily reflect those of the European Union or the European Research Council Executive Agency. Neither the European Union nor the granting authority can be held responsible for them.

\appendix

\printbibliography

@article{Scott_2014,
  author = {Scott, Alexander D. and Sokal, Alan D.},
  title = {Complete monotonicity for inverse powers of some combinatorially defined polynomials},
  journal = {Acta Mathematica},
  volume = {213},
  year = {2014}
}

@article{Exponential_varieties,
  author = {Michałek, Mateusz and Sturmfels, Bernd and Uhler, Caroline and Zwiernik, Piotr},
  title = {Exponential varieties},
  journal = {Proceedings of the London Mathematical Society},
  volume = {112},
  year = {2016}
}

@article{Guler_Hyperbolic_pol,
  author = {Güler, Osman},
  title = {Hyperbolic polynomials and interior point methods for convex programming},
  journal = {Mathematical Operations Research},
  volume = {22},
  year = {1997}
}

@article{Polypols,
  author = {Kohn, Kathlén and Piene, Ragni and Ranestad, Kristian and Rydell, Felix and Shapiro, Boris and Sinn, Rainer and Sorea, Miruna-Stefana and Telen, Simon},
  title = {Adjoints and canonical forms of polypols},
  journal = {Documenta Mathematica},
  volume = {30},
  number = {2},
  pages = {275--346},
  year = {2025}
}

@article{10.2969/jmsj/05210161,
author = {Hideyuki ISHI},
title = {{Positive Riesz distributions on homogeneous cones}},
volume = {52},
journal = {Journal of the Mathematical Society of Japan},
number = {1},
publisher = {Mathematical Society of Japan},
pages = {161 -- 186},
year = {2000},
}

@article{Positive_geometries,
   title={Positive geometries and canonical forms},
   volume={2017},
   number={11},
   journal={Journal of High Energy Physics},
   publisher={Springer Science and Business Media LLC},
   author={Arkani-Hamed, Nima and Bai, Yuntao and Lam, Thomas},
   year={2017},
   month=nov }

@article{Kummer_2015,
   title={A Note on the Hyperbolicity Cone of the Specialized Vámos Polynomial},
   volume={144},
   number={1},
   journal={Acta Applicandae Mathematicae},
   publisher={Springer Science and Business Media LLC},
   author={Kummer, Mario},
   year={2015},
   month=dec, pages={11–15} }

@article{Helton_Linear,
  title={Linear matrix inequality representation of sets},
  author={Helton, J William and Vinnikov, Victor},
  journal={Communications on Pure and Applied Mathematics: A Journal Issued by the Courant Institute of Mathematical Sciences},
  volume={60},
  number={5},
  pages={654--674},
  year={2007},
  publisher={Wiley Online Library}
}

@article{the_amplituhedron,
    author = "Arkani-Hamed, Nima and Trnka, Jaroslav",
    title = "{The Amplituhedron}",
    journal = "JHEP",
    volume = "10",
    pages = "030",
    year = "2014"
}

@article{ABHY_original,
    author = "Arkani-Hamed, Nima and Bai, Yuntao and He, Song and Yan, Gongwang",
    title = "{Scattering Forms and the Positive Geometry of Kinematics, Color and the Worldsheet}",
    journal = "JHEP",
    volume = "05",
    pages = "096",
    year = "2018"
}

@article{ABJM_amplituhedron,
    author = "He, Song and Huang, Yu-tin and Kuo, Chia-Kai",
    title = "{The ABJM Amplituhedron}",
    journal = "JHEP",
    volume = "09",
    pages = "165",
    year = "2023",
    note = "[Erratum: JHEP 04, 064 (2024)]"
}

@article{Correlahedron,
    author = "Eden, Burkhard and Heslop, Paul and Mason, Lionel",
    title = "{The Correlahedron}",
    reportNumber = "DCPT-16-59",
    journal = "JHEP",
    volume = "09",
    pages = "156",
    year = "2017"
}

@article{Cosmoehdra,
    author = {Arkani-Hamed, Nima and Figueiredo, Carolina and Vaz{\~a}o, Francisco},
    title = {Cosmohedra},
    journal = {Journal of High Energy Physics},
    volume = {2025},
    number = {11},
    pages = {029},
    year = {2025}
}

@article{Cosmological_polytopes,
    author = "Arkani-Hamed, Nima and Benincasa, Paolo and Postnikov, Alexander",
    title = "{Cosmological Polytopes and the Wavefunction of the Universe}",
    year = "2017"
}

@incollection{Lam:_PG_notes,
    author = {Lam, Thomas},
    title = {An invitation to positive geometries},
    booktitle = {Open Problems in Algebraic Combinatorics},
    series = {Proceedings of Symposia in Pure Mathematics},
    volume = {110},
    pages = {159--180},
    publisher = {American Mathematical Society},
    year = {2024}
}

@article{Ranestad:what_is_PG,
  author = {Kristian Ranestad and Bernd Sturmfels and Simon Telen},
  title = {What is positive geometry?},
  journal = {Le Mathematiche},
  volume = {80},
  pages = {3--16},
  year = {2025}
}

@article{Fevola:Pos_Geom,
  author = {Fevola, Claudia and Sattelberger, Anna-Laura},
  title = {Algebraic and positive geometry of the universe: from particles to galaxies},
  journal = {Notices of the American Mathematical Society},
  volume = {72},
  number = {8},
  pages = {808--817},
  year = {2025}
}

@article{Brown:PG_Hodge,
  author = {Brown, Francis and Dupont, Clément},
  title = {Positive geometries and canonical forms via mixed Hodge theory},
  journal = {Communications in Mathematical Physics},
  volume = {406},
  number = {11},
  pages = {267},
  year = {2025}
}

@article{Hodges,
  author = {Andrew Hodges},
  title = {Eliminating spurious poles from gauge-theoretic amplitudes},
  journal = {Journal of High Energy Physics},
  volume = {05},
  pages = {135},
  year = {2013}
}

@article{Gaetz,
  author = {Christian Gaetz},
  title = {Canonical forms of polytopes from adjoints},
  journal = {arXiv preprint},
  year = {2025}
}

@article{widder2015laplace,
  author = {David Vernon Widder},
  title = {Laplace transform (PMS-6)},
  journal = {Princeton University Press},
  year = {2015}
}

@article{Henn:CM,
  author = {Johannes Henn and Prashanth Raman},
  title = {Positivity properties of scattering amplitudes},
  journal = {Journal of High Energy Physics},
  volume = {04},
  pages = {150},
  year = {2025}
}

@article{Choquet,
  author = {Gustave Choquet},
  title = {Deux exemples classiques de représentation intégrale},
  journal = {Enseignement Mathématique},
  volume = {15},
  number = {2},
  pages = {63--75},
  year = {1969}
}

@article{Pos_Certif,
  author = {Khazhgali Kozhasov and Mateusz Michałek and Bernd Sturmfels},
  title = {Positivity certificates via integral representations},
  journal = {Facets of Algebraic Geometry},
  volume = {2},
  pages = {84--114},
  year = {2019}
}

@article{Garding_1959,
  author = {Lars Gårding},
  title = {An inequality for hyperbolic polynomials},
  journal = {Journal of Mathematics and Mechanics},
  volume = {8},
  number = {6},
  pages = {957--965},
  year = {1959}
}

@article{1Garding_1970,
  author = {M. F. Atiyah and R. Bott and L. Gårding},
  title = {Lacunas for hyperbolic differential operators with constant coefficients I},
  journal = {Acta Mathematica},
  volume = {124},
  pages = {109--189},
  year = {1970}
}

@article{Garding_1951,
  author = {Lars Gårding},
  title = {Linear hyperbolic partial differential equations with constant coefficients},
  journal = {Acta Mathematica},
  volume = {85},
  pages = {1--62},
  year = {1951}
}

@article{sinn2015algebraic,
  author = {Rainer Sinn},
  title = {Algebraic boundaries of convex semi-algebraic sets},
  journal = {Research in the Mathematical Sciences},
  volume = {2},
  number = {1},
  pages = {3},
  year = {2015}
}

@article{branden2014hyperbolicity,
  author = {Petter Brändén},
  title = {Hyperbolicity cones of elementary symmetric polynomials are spectrahedral},
  journal = {Optimization Letters},
  volume = {8},
  number = {5},
  pages = {1773--1782},
  year = {2014}
}

@book{gelʹfand1968generalized,
  author = {Izrailʹ Moiseevič Gelʹfand and Georgij Evgenʹevič Šilov},
  title = {Generalized functions. Vol. 2, Spaces of fundamental and generalized functions},
  publisher = {Academic Press},
  year = {1968}
}

@article{wagner1,
  author = {Peter Wagner},
  title = {A fundamental solution of N. Zeilon's operator},
  journal = {Mathematica Scandinavica},
  pages = {273--287},
  year = {2000}
}

@article{wagner2,
  author = {Peter Wagner},
  title = {Fundamental solutions of real homogeneous cubic operators of principal type in three dimensions},
  journal = {Acta Mathematica},
  volume = {182},
  pages = {283--300},
  year = {1999}
}

@article{wagner4,
  author = {Peter Wagner},
  title = {On the fundamental solutions of a class of hyperbolic quartic operators in dimension 3},
  journal = {Annali di Matematica Pura ed Applicata},
  volume = {184},
  number = {2},
  pages = {139--159},
  year = {2005},
  publisher = {Springer}
}

@book{hormander2,
  author = {Lars Hörmander},
  title = {The analysis of linear partial differential operators. 2. Differential operators with constant coefficients},
  series = {Grundlehren der mathematischen Wissenschaften 257},
  publisher = {Springer},
  address = {Berlin},
  year = {1983},
  isbn = {3540121390}
}

@book{hormander1,
  author = {Lars Hörmander},
  title = {The analysis of linear partial differential operators. 1. Distribution theory and Fourier analysis},
  series = {Grundlehren der mathematischen Wissenschaften 256},
  publisher = {Springer},
  address = {Berlin},
  year = {1983},
  isbn = {3540121048}
}

@book{ampl_and_pos_gr,
  author = {Nima Arkani-Hamed and Jacob L. Bourjaily and Freddy Cachazo and Alexander B. Goncharov and Alexander Postnikov and Jaroslav Trnka},
  title = {Grassmannian geometry of scattering amplitudes},
  publisher = {Cambridge University Press},
  year = {2016},
}

@article{Banerjee:2018tun,
    author = "Banerjee, Pinaki and Laddha, Alok and Raman, Prashanth",
    title = "{Stokes polytopes: the positive geometry for $\phi^{4}$ interactions}",
    journal = "JHEP",
    volume = "08",
    pages = "067",
    year = "2019"
}

@article{Raman:2019utu,
    author = "Raman, Prashanth",
    title = "{The positive geometry for $\phi^{p}$ interactions}",
    journal = "JHEP",
    volume = "10",
    pages = "271",
    year = "2019"
}

@article{Aneesh:2019cvt,
    author = "Aneesh, P. B. and Banerjee, Pinaki and Jagadale, Mrunmay and Rajan, Renjan and Laddha, Alok and Mahato, Sujoy",
    title = "{On positive geometries of quartic interactions: Stokes polytopes, lower forms on associahedra and world-sheet forms}",
    journal = "JHEP",
    volume = "04",
    pages = "149",
    year = "2020"
}

@article{Jagadale:2020qfa,
    author = "Jagadale, Mrunmay and Laddha, Alok",
    title = "{On positive geometries of quartic interactions: one loop integrands from polytopes}",
    journal = "JHEP",
    volume = "07",
    pages = "136",
    year = "2021"
}

@article{Jagadale:2022rbl,
    author = "Jagadale, Mrunmay and Laddha, Alok",
    title = "{Towards Positive Geometries of Massive Scalar field theories}",
    year = "2022"
}

@article{Jagadale:2023hjr,
    author = "Jagadale, Mrunmay and Laddha, Alok",
    title = "{Positive Geometries of S-matrix without Color}",
    year = "2023"
}

@article{guler1996barrier,
  title={Barrier functions in interior point methods},
  author={G{\"u}ler, Osman},
  journal={Mathematics of Operations Research},
  volume={21},
  number={4},
  pages={860--885},
  year={1996},
  publisher={INFORMS}
}

@article{Herrmann:2020qlt,
    author = "Herrmann, Enrico and Langer, Cameron and Trnka, Jaroslav and Zheng, Minshan",
    title = "{Positive geometry, local triangulations, and the dual of the Amplituhedron}",
    journal = "JHEP",
    volume = "01",
    pages = "035",
    year = "2021"
}

@article{Duhr:2011zq,
    author = "Duhr, Claude and Gangl, Herbert and Rhodes, John R.",
    title = "{From polygons and symbols to polylogarithmic functions}",
    reportNumber = "IPPP-11-56, DCPT-11-112",
    journal = "JHEP",
    volume = "10",
    pages = "075",
    year = "2012"
}

@article{Goncharov:2010jf,
    author = "Goncharov, Alexander B. and Spradlin, Marcus and Vergu, C. and Volovich, Anastasia",
    title = "{Classical Polylogarithms for Amplitudes and Wilson Loops}",
    reportNumber = "BROWN-HET-1602",
    journal = "Phys. Rev. Lett.",
    volume = "105",
    pages = "151605",
    year = "2010"
}

@article{Ranestad:adjoint,
    author = "Ranestad, Kristian and Sinn, Rainer and Telen, Simon",
    title = "{Adjoints and Canonical Forms of Tree Amplituhedra}",
    journal = "Math. Scand.",
    volume = "130",
    pages = "433--466",
    year = "2024"
}

@article{Mazzucchelli:2025kzy,
    author = {Mazzucchelli, Elia and Pratt, Elizabeth},
    title = {Exterior Cyclic Polytopes and Convexity of Amplituhedra},
    month = {7},
    year = {2025}
}

@article{Ferro,
    author = "Ferro, Livia and Lukowski, Tomasz and Orta, Andrea and Parisi, Matteo",
    title = "{Towards the Amplituhedron Volume}",
    journal = "JHEP",
    volume = "03",
    pages = "014",
    year = "2016"
}

@article{KohnRanestad2020,
  author  = {Kohn, Kathl{\'e}n and Ranestad, Kristian},
  title   = {Projective geometry of Wachspress coordinates},
  journal = {Foundations of Computational Mathematics},
  volume  = {20},
  pages   = {1135--1173},
  year    = {2020},
}

@article{BruserWeigert2025,
  author  = {Br{\"u}ser, Clemens and Weigert, Julian},
  title   = {Geometry of adjoint hypersurfaces for polytopes},
  year    = {2025},
  note    = {Accepted to Algebraic Combinatorics}
}

@book{BPT2013,
  editor    = {Blekherman, Grigoriy and Parrilo, Pablo A. and Thomas, Rekha R.},
  title     = {Semidefinite optimization and convex algebraic geometry},
  series    = {MOS-SIAM Series on Optimization},
  volume    = {13},
  publisher = {Society for Industrial and Applied Mathematics},
  address   = {Philadelphia},
  year      = {2013}
}

@article{MaWeigert2025,
  author  = {Ma, Rourou and Weigert, Julian},
  title   = {Complete monotonicity of log-functions},
  journal = {Analysis and Mathematical Physics},
  volume  = {15},
  number  = {6},
  pages   = {141},
  year    = {2025}
}
\end{document}